\documentclass[12pt,reqno]{amsart}
\usepackage{amssymb,amsmath,bm}
\usepackage{graphicx}
\usepackage{caption}
\usepackage{color}
\usepackage[usenames,dvipsnames]{xcolor}
\usepackage{tikz}
\def\d{{\rm d}}

\def\R {\mathbb{R}}

\def\e {\varepsilon}

\def\la {\lambda}
\def\g{\gamma}
\def\G{\Gamma}

\def\s {\sigma}

\def\O{\Omega}
\def\We{W_{\rm e}}
\def\Wh{W_{\rm  p}}

\def\be{\boldsymbol e}

\def\bz{\boldsymbol z}
\def\bze{\boldsymbol z_\varepsilon}
\def\bzehat{\widehat{\boldsymbol z}_\varepsilon}

\def\bA{\boldsymbol A}
\def\bB{\boldsymbol B}
\def\bC{\boldsymbol C}
\def\bD{\boldsymbol D}
\def\bG{\boldsymbol G}
\def\bL{\boldsymbol L}
\def\bB{\boldsymbol B}
\def\bP{\boldsymbol P}
\def\bPi{\boldsymbol \Pi}

\def\bQ{\boldsymbol Q}
\def\bP{\boldsymbol P}
\def\bR{\boldsymbol R}

\def\bN{\boldsymbol N}
\def\bE{\boldsymbol e}
\def\bF{\boldsymbol F}
\def\bS{\boldsymbol S}
\def\bT{\boldsymbol T}

\def\bFe{\boldsymbol F_{\rm e}}
\def\bFp{\boldsymbol P}
\def\bFpk{\boldsymbol P_k}
\def\bCe{\boldsymbol C_{\rm e}}
\def\bCp{\boldsymbol C_{\rm p}}
\def\bCpe{\boldsymbol C_{\rm p \varepsilon}}
\def\bCpehat{\widehat{\boldsymbol C}_{\rm p \varepsilon}}
\def\bCpetild{\widetilde{\boldsymbol C}_{\rm p \varepsilon}}

\def\dotCp{\dot{\boldsymbol C}_{\rm p}}

\def\one{\boldsymbol{I}}
\def\bbC{\mathbb{C}}
\def\bbH{\mathbb{H}}
\def\bbW{\mathbb{W}}
\def\bbM{\mathbb{M}}
\def\calD{\mathcal D}
\def\calE{\mathcal E}

\def\calS{\mathcal S}
\def\calQ{\mathcal Q}
\def\calY{\mathcal Y}
\def\calW{\mathcal W}
\def\calZ{\mathcal Z}
\def\TT{\top}
\def \l {\langle}
\def \r {\rangle}

\newtheorem{proposition}{Proposition}[section]
\newtheorem{theorem}[proposition]{Theorem}

\newtheorem{lemma}[proposition]{Lemma}
\theoremstyle{definition}

\numberwithin{equation}{section}
\newcommand{\epsi}{\varepsilon}

\newcommand{\disp}{\displaystyle}
\newcommand{\Nz}{{\mathbb N}}
\newcommand{\Rz}{{\mathbb R}}
\newcommand{\Rd}{{\mathbb R}^3}
\newcommand{\Rzn}{{\mathbb R}^{3 \times 3}}
\newcommand{\Rzd}{\Rz^{3\times 3}_{\rm dev}}
\newcommand{\Rzs}{\Rz^{3\times 3}_{\rm sym}}
\newcommand{\Rzsp}{\Rz^{3\times 3}_{\rm sym+}}
\newcommand{\haz}{\widehat}

\newcommand{\dx}{\,\text{\rm d}x}
\newcommand{\dt}{\,\text{\rm  d}t}

\newcommand{\SL}{\text{\rm SL}}

\newcommand{\SO}{\text{\rm SO}}
\newcommand{\GLp}{\text{\rm GL}^+}
\newcommand{\GLps}{\text{\rm GL}^+_{\rm sym}}
\newcommand{\MM}{\SL_{\rm sym}^+}
\newcommand{\obCp}{\overline \bC_{\rm p}}
\newcommand{\bCpk}{\bC_{{\rm p} k}}
\newcommand{\bCpks}{\bC_{{\rm p} k}^{1/2}}
\newcommand{\hbCp}{\haz \bC_{\rm p}}

\newcommand{\weak}{\rightharpoonup}
\newcommand{\weakto}{\rightharpoonup}

\newcommand{\tr}{\textrm{tr}\,}
\newcommand{\cof}{\textrm{cof\,}}
\newcommand{\dev}{{\rm dev}\,}
\def \no#1#2#3 {{\bf #1} (#3), #2.}
\def \eds#1#2#3 {#1, #2, #3.}
\title[]
{ Finite Plasticity in $\bFp^\TT\! \bFp$}%\\
\author[]
{D. Grandi, U. Stefanelli}
\address{Faculty of Mathematics, University of Vienna
\newline\indent
Oskar-Morgenstern-Platz 1, A-1090 Vienna, Austria}
\begin{document}
\begin{abstract}
 We discuss a finite-plasticity model based on the symmetric tensor
  $\bFp^\TT\! \bFp$ instead of the classical plastic strain $\bFp$. Such a model structure arises from assuming that the
  material behavior is  invariant with respect to frame
  transformations of the 
  intermediate configuration. The resulting variational model is
  lower-dimensional, symmetric, and based solely on the
  reference configuration. We discuss the existence of energetic
  solutions both at the material-point level and for the quasistatic
  boundary-value problem.  These solutions are constructed as
  limits of time discretizations. Eventually, the linearization of the model for
  small deformations is ascertained via a
  rigorous evolutive-$\Gamma$-convergence argument. 
\end{abstract}
\maketitle
%
%
%\tableofcontents
%=====================

\section{Introduction}

The inelastic behavior of polycrystalline solids is classically
described in terms of the deformation gradient $\bF$ with
respect to the reference configuration \cite{Gurtin81}. As the elastic response is
observed to be largely independent from prior plastic distortion of
the crystalline structure, the deformation gradient is usually
decomposed into an elastic and a plastic part.  While this decomposition is
additive in the small-deformation regime, at finite strains, a
multiplicative decomposition  $\bF =
\bFe\bFp$ is used instead \cite{Kroener,Lee69}. Here $\bFe$ is the elastic deformation tensor, describing
indeed the elastic response of the medium, and $\bFp$ is the plastic
deformation tensor, encoding the information on the plastic state
. Although other options have been
advanced, see for instance
\cite{Clifton72,Davoli-Francfort15,Lubarda99,Nemat-Nasser79}, the multiplicative
decomposition has now turned to be the reference in finite
plasticity.
A justification for this
decomposition has been recently provided in
\cite{Conti-Reina,Conti-Reina2} on the basis
micromechanical considerations. 

Based on the multiplicative decomposition, the
elastoplastic evolution of the medium is described by the time
evolution of $\bFe$ and $\bFp$.  This results from the competition between
energy-storage   and plastic-dissipation mechanisms
\cite{Nemat-Nasser04,Simo-Hughes98}.
A first  structural restriction to the modeling choice   is
that of {\it frame indifference}
\cite{Gurtin81}, imposing indeed the elastic state of the
material to be completely represented in terms of the so-called {\it (right)
Cauchy-Green tensor} $ \bFe^{\TT}\!\bFe$. 
 Moving from   by this
observation, the focus of this paper is to address the possibility
of formulating and analyzing
finite plasticity in terms of the corresponding {\it plastic} Cauchy-Green
tensor $\bFp^\TT \!\bFp$ instead of $\bFp$.

Finite-plasticity models in terms of $ \bFp^\TT
\!\bFp$ bear significant advantages with respect to formulations in $\bFp$. At first,
variables are symmetric and positive definite, reducing indeed the degrees of freedom of
the problem. Furthermore, $\bCp$ being symmetric plays a significant
computational role allowing the use of efficient algorithms, especially
in connection with power- and exponential-matrix
evaluations \cite{Dettmer-Reese}. 
Secondly, a model in $ \bFp^\TT
\!\bFp$ is fully defined on the reference configuration of the
medium, avoiding the necessity of any
intermediate configuration, a commonly controversial  issue 
\cite{Naghdi90}. 
Moreover, such a fully Lagrangian formulation is better adapted to finite-element
approximations, which are then to be defined for all variables on the fixed reference
configuration. Triggered by these appealing features,
finite-plasticity models based on $ \bFp^\TT
\!\bFp$  have already attracted attention
\cite{Simo93,Lion97,Miehe95,Reese&al08}. The reader is referred to the recent 
 \cite{Neff15} where a comparative study of these as well as the current model is
 provided. In the context of shape memory
 materials, a model based on $ \bFp^\TT
\!\bFp$ is advanced in \cite{Eva09,Eva10} and
 variationally reformulated and analyzed in \cite{Fri-Ste09}. 

The aim of this paper is to provide a comprehensive discussion of
finite plasticity expressed in terms of $ \bFp^\TT
\!\bFp$, both at the modeling and at the analytical level.  
Starting from classical associative
finite plasticity in terms of $\bFp$, we discuss a general frame allowing
an equivalent reformulation in $ \bFp^\TT
\!\bFp$. This relies on a quite natural {\it plastic-invariance
assumption}, translating indeed the indifference of the model with
respect to rotations of the intermediate configuration. Quite
remarkably, the model in $ \bFp^\TT
\!\bFp$ turns out to be associative with respect to the new variables as well.

The variational structure of the model  allows  us to prove the existence of
variational solutions of {\it energetic type}
\cite{Fra-Mie06,Mie05,Mie-The04}. At the material-point level we
obtain such an existence under very general assumptions on the model
ingredients, in particular on the coercivity of the energy. We then
turn to the analysis of the  quasistatic-evolution setting  resulting from the
combination of the constitutive material relation with the equilibrium
system. Also at the
quasistatic evolution level energetic solutions are proved to exists,
provided the energy is polyconvex \cite{Ball77} and augmented by a gradient term of the form
$|\nabla(\bFp^\TT\!\bFp)|$. Such a term describes nonlocal plastic
effects and is inspired to the by-now classical {\it gradient
  plasticity} theory \cite{Fleck1,Fleck2,Aifantis91}. In particular,
its occurrence turns out to be crucial in order to prevent the formation
of plastic microstructures and ultimately ensures the necessary
compactness for the analysis. In the context of the mathematical
analysis of finite plasticity, the introduction of suitable regularizing
terms on the plastic variables
seems at the moment unavoidable. The {
\it incremental} problem has been tackled under the regularizing effect of
a term in  
${\rm curl}\, \bFp$   in \cite{Mie-Mull06}
An existence result in finite (incremental)
plasticity without gradient terms is in  \cite{Mie04} where however
 substantial  restrictions on modeling choices are 
imposed.  The only other
available  quasistatic-evolution existence result is for the
formulation in $\bFp$
\cite{Mai-Mie09} and features 
a regularizing term in $\nabla \bFp$ as well. 

The existence results, both at the material-point and the
quasistatic-evolution level, results from passing to the limit in
implicit time-discretization schemes. As a by-product we hence devise the
convergence of such schemes both in terms of solution trajectories and
of energy and dissipation.

A second important focus of our analysis is  the rigorous
justification of the classical linearization approach for small
deformations. Within the small-deformation regime it is indeed
customary to leave the nonlinear finite-strain frame and resort to
linearized theories. This model reduction is classically
justified by heuristic Taylor-expansion arguments. Here, we aim
instead at providing a rigorous linearization proof by means of an
{\it evolutionary $\Gamma$-convergence} analysis in the spirit of the
general abstract theory of \cite{MieRouSte08}. This rigorous limiting procedure is devised both at
the material-point and at the quasistatic-evolution level.  Note that a rigorous convergence
result in case of  the  $\bFp$-based formulation was provided in
\cite{MieSte13}. Our results  extends  that of \cite{MieSte13}
 to 
the case of $\bFp^\TT\!\bFp$-plasticity. In contrast  with  \cite{MieSte13} we discuss here the convergence of solutions for
which existence is known. This involves the additional difficulty of
discussing the convergence of the gradient terms.

We  describe  the constitutive model and the role of
plastic-rotation indifference on Section \ref{sec:const}. Section
\ref{sub:es} provides a minimal toolbox on energetic
formulations of rate-independent systems and the corresponding
approximation via evolutionary $\Gamma$-convergence. The existence of
energetic solutions of the constitutive model at the material-point
level is discussed in Section \ref{sec:solve_const} and the
corresponding small-deformation limit is presented in Section \ref{sec:lin_const}. The
quasistatic-evolution problem is introduced in Section
\ref{sec:quasi}. The corresponding existence result is then presented in Section \ref{sec:energetic} whereas
Section \ref{sec:linear} contains the detail of the
small-deformation limit.

\section{Constitutive model}\label{sec:const}
%===========================================
%
We devote this section to the specification of the finite-plasticity 
model in study.  As already commented in the Introduction, this corresponds to classical
 associative finite plasticity under an 
 invariance assumption with respect to plastic rotations. We limit
 ourselves at introducing the constitutive relation, referring indeed the reader
to the monographs \cite{Gurtin10,Nemat-Nasser04,Simo-Hughes98} for
additional material and detail on 
finite-plasticity formulations.  

 Before going on let  us record here that finite plasticity is to-date a still controversial
 subject \cite{Naghdi90}. It is not our intention to contribute new
 mechanical arguments to the ongoing
 discussion. On the contrary  our aim is to present the possibly simplest model
 in $\bCp$ showing a sound variational structure. The specific form of
 our constitutive model seems to be new \cite{Neff15}. Still, we
 believe that the main interest in
 this rather simplified case relies the on the quite detailed mathematical
 analysis that such a variational structure allows.

\subsection{Tensors} We focus  on the three-dimensional setting and systematically use
boldface symbols in order to indicate 
$2$-tensors in $\Rz^3$. The corresponding space is denoted by
$\Rz^{3\times 3}$. Given $\boldsymbol A\in \Rz^{3\times 3}$ we
classically define its trace as $\tr \boldsymbol A : = A_{ii}$
(summation convention), its deviatoric part as
${\rm dev}\boldsymbol A  = \boldsymbol A - (\tr \boldsymbol A) \one
/3$ where $\one$ is the identity $2$-tensor, and its (Frobenius)
 norm as $|\boldsymbol A|^2 := \tr (\boldsymbol A ^\TT\! \boldsymbol A)
 $ where the symbol $\TT$ denotes transposition. The contraction product between $2$-tensors is $\boldsymbol A {:}
 \boldsymbol B  := A_{ij}B_{ij}$ and we classically denote the
 scalar product of vectors in $\Rd$ by $a{\cdot} b:=a_ib_i$. The
 symbols $\Rzs$ and $\Rzsp$ stand for the  subsets  of $\Rz^{3 \times
   3}$ of symmetric tensors and   of  symmetric positive-definite tensors,
 respectively. Moreover, $\Rzd$ indicates the space of symmetric
 deviatoric tensors, namely $\Rzd:=\{\boldsymbol A \in  \Rzs  \; |\; \tr
 \boldsymbol A =0\}$.   We shall use also the following tensor sets
\begin{align*}
  &\SL:=\{\boldsymbol A \in \Rzn \;|\; \det \boldsymbol A
  =1\},\\
&\SO:=\{\boldsymbol A  \in  \SL  \;|\; \boldsymbol A^{-1} =
  \boldsymbol A ^\TT  \},\\
&\GLp:=\{\boldsymbol A  \in \Rzn \;|\; \det \boldsymbol A >0 \},\\
&\GLps:= \GLp\cap \Rzs, \\
 &\MM:= \SL \cap \Rzsp.
\nonumber
\end{align*}
The tensor $\cof \bA$ is the {\it cofactor matrix} of $\bA $. For $\bA$ 
invertible we have that $\cof \bA=(\det\bA)\,\bA^{-\TT}$. For any symmetric positive-definite matrix $\bA\in \R^{3\times 3}_{\rm{sym}+}$,
the real power $\bA^s$ is classically defined,  for any $s\in\R$, in terms of its
eigenvalues $(\lambda_1,\lambda_2,\lambda_3)$, $\lambda_i>0$ and 
$$
\tr \bA^s= \lambda_1^s+\lambda_2^s+\lambda_3^s, \quad\det\bA^s= (\lambda_1\lambda_2\lambda_3)^s.
$$
In particular, 
the square root $\boldsymbol A^{1/2}$ is uniquely defined in
$\GLps$. The matrix logarithm $\log \bCp$ is globally uniquely defined
 in  $ \MM$. In particular, one has that $\tr (\log
\bCp) = \log (\det \bCp) = 0$ for all $\bCp \in \MM$.
Given any symmetric, positive-definite 4-tensor $\bbC$ we denote by
$|\bA|^2_\bbC:= \bA{:} \bbC \bA$ the corresponding induced (squared) norm on
$\Rzs$.  The product $\bbC \bA$ is here classically defined
as $(\bbC \bA)_{ij} := \bbC_{ij\ell k}A_{\ell k}$. For
all $3$-tensors $\bD$ (arising in this context as gradients of $2$-tensor fields) we define $|\bD|^2 := D_{ijk}^2$, the partial
trasposition $(\bD^\TT)_{ijk} = D_{jik}$, and the product with the
2-tensor $\bA$ as $(\bD \bA)_{ijk} =
D_{ij\ell}A_{\ell k}$ and $(\bA \bD)_{ijk} = A_{i\ell}D_{\ell
  j k}$. Note that, along with these definitions, $|\bD \bA|, |\bA \bD|
\leq |\bA|\, |\bD|$.

In the following we denote by $\partial \varphi$ the subdifferential
of the smooth or  of the  convex, proper, and lower semicontinuous function
$\varphi: E \to (-\infty,\infty]$ where $E$ is a normed space with dual
$E^*$ and duality pairing $\langle \cdot, \cdot \rangle $
\cite{Brezis73}. In particular, $ y^* \in \partial \varphi (x)$ iff
$\varphi(x) <\infty$ and 
$$ \langle y^*,w{-}x\rangle \leq \varphi(w) - \varphi(x) \quad
\forall w \in E.$$

A {\it caveat} on notation: in the following we use the same symbol $c$ in order to indicate a
generic constant, possibly depending on data and varying from line to
line.

\subsection{Deformation}
We consider an elastoplastic body occupying the reference configuration $\Omega
 $, which is assumed to be a nonempty, open, connected, and bounded subset
of $\Rd$ with  Lipschitz boundary
$\partial \Omega$. The three-dimensional setting is here chosen for
the sake of notational definiteness only: both modeling and analysis
could be reformulated in one or two dimensions. 

The deformation of the body is described by
$y:\O\rightarrow \Rd$ and is assumed to be such that the deformation
gradient $\bF:=\nabla y$ is almost everywhere defined and belongs to
$\GLp$. The deformation gradient $\bF$ is classically decomposed
as \cite{Kroener,Lee69}
\begin{equation}\label{mult}
\bF=\bFe\bFp
\end{equation}
where $\bFe$ denotes the {\it elastic} part of $\bF$ and $\bFp$ its
{\it plastic}
part.  In particular, the plastic tensor $\bFp$ describes the internal plastic state
of the material and fulfills $$\det\bFp=1$$ in order
to express the {\it isochoric} nature of plastic deformations, as 
customary in metal plasticity
\cite{Simo-Hughes98}.  The heuristics for the multiplicative
decomposition \eqref{mult} resides in the classical chain rule: in case $y$ can be interpreted as a composition $y_{\rm e} \circ
y_{\rm p}$ of an elastic and a plastic deformation, the set
$y_{\rm p}(\Omega)$ is termed {\it intermediate} (or {\it structural})
configuration  and $\nabla y = \nabla y_{\rm e}(y_{\rm p})\nabla
y_{\rm p}$.  Note nonetheless that the tensors $\bFe$ and $\bFp$ need not
be gradients as the compatibility conditions ${\rm curl}\, \bFe =
\boldsymbol 0$ and  ${\rm curl}\, \bFp =
\boldsymbol 0$ may not hold. Correspondingly, the intermediate
configuration can be understood in a local sense only \cite{Naghdi90}. We refer to the recent \cite{Conti-Reina,Conti-Reina2} for a
justification of the multiplicative decomposition \eqref{mult} in
two dimensions consisting in a kinematic analysis
of elastoplastic deformation in plastic-slip and dislocation systems.

 The {\it (right) Cauchy-Green} symmetric tensors associated to the three deformation gradients are
defined by
\begin{equation*}
 \bC:=\bF^\TT\!\bF\in \GLps,\quad \bCe:=\bFe^\TT\!\bFe\in \GLps,\quad
 \bCp:=\bFp^\TT\!\bFp\in \MM.
\end{equation*}
In particular, we
have that  $\det \bCp = (\det\bFp)^2=1$. Note that these tensors are all true
tensorial quantities, all defined on the reference configuration,
whereas $\bF$, $\bFe$, $\bFp$ are two-points tensors.
% The local state of the
% elastoplastic body will be equivalently described by the pairs  
% \begin{equation*}
% (\bF, \bFp)\in\GLp \times \SL  \ \ \text{or} \ \   (\bF, \bCp)\in\GLp \times \MM.
% \end{equation*}
%or, equivalently, by the pair $(\bC,\bCp) \in\GLp\cap \Rzs \times \MM,$
 % Note that
% $\MM$ is a connected five-dimensional closed submanifold of
% $\Rzd$, while $\bFp$ belongs to $\SL$, which is eight-dimensional.
%We observe that it is uniquely defined the
% positive-definite square root $\bCp^{1/2}\in\MM $.\\
% For any matrix $\bA\in\R^{3\times 3}$, we denote its (Frobenius) norm by
% \begin{equation}
% |\bA|:=(\bA:\bA)^{1/2}=[\tr{(\bA^\TT\!\bA)}\,]^{1/2}.
% \end{equation}
%
%-----------------------------
\subsection{Energy}
%-----------------------------
%
The evolution of the elastoplastic body is governed by the
interplay between energy-storage mechanisms and plastic-dissipative
effects. We assume  from the very beginning the response of the medium
to be {\it hyperelastic}
\cite{Truesdell-Noll65} and 
start by specifying the {\it energy density} of the medium by imposing
the additive decomposition  
\begin{equation}\label{eq:energy0}
   \We(\bFe)+\Wh(\bFp)
\end{equation}
into an {\it elastic} and  a {\it plastic} (or {\it hardening})  energy term. 

The elastic energy density $\We: \GLp\to[0,\infty)$ is required to be
$C^1$ and
\emph{frame indifferent} \cite{Truesdell-Noll65}, namely  
\begin{equation}\label{eq:frame-ind}
\We(\bR\bFe)=\We(\bFe)\quad \forall \bR\in\SO.
\end{equation}
Frame indifference implies that the elastic energy can be expressed
solely in terms of the tensor $\bCe$. Indeed, given $\bFe\in \GLp$ by
polar decomposition there exists a rotation matrix $\bR\in
\SO$ such that $\bFe=\bR\bCe^{1/2}$ and 
$$\We(\bFe) = \We(\bR^\TT\!\bFe) =
\We (\bCe^{1/2}) =: \haz \We(\bCe)$$
 where now
$\widehat{\We}:\GLps\to[0,\infty)$. 

We admit here hardening effects of a purely kinematic nature. These are modulated by
the {\it  plastic-energy density} $\Wh:\SL \to
[0,\infty]$, which we assume to be $C^1$ on its domain. 
 Let us  explicitly remark that we are not considering here additional
internal hardening dynamics. In particular, isotropic hardening is not
directly included in our frame. Our choice is  motivated by
the mere sake of simplicity. Additional internal parameters could be
considered as well.

\subsection{ Plastic-rotation indifference}
The crucial assumption of our analysis is that the  material
behavior  is invariant
by plastic rotations. This invariance is formulated as  
\begin{equation}
  \label{plas_inv}
  \We(\bFe \bQ) = \We(\bFe), \quad \Wh(\bQ\bFp) = \Wh(\bFp) \quad
  \forall \bQ \in \SO
\end{equation}
for all $\bFe \in \GLp$ and $\bFp \in \SL$. The condition on $\We$
corresponds to {\it isotropy},  whereas  $\Wh$ can be
nonisotropic instead. The condition on $\Wh$ is then nothing but frame
indifference with respect to the intermediate configuration.

By using the  polar decomposition   $\bFp=\bQ\bCp^{1/2}$ for
$\bQ \in \SO$ we have 
$$
\bFe=\bF\bFp^{-1}=\bF\bCp^{-1/2}\bQ^\TT.
$$
 The isotropy of $\We$ from \eqref{plas_inv}  then yields $\We(\bFe)=\We(\bF\bCp^{-1/2})$. By
combining frame indifference and isotropy  of $\We$  one can equivalently rewrite the elastic energy density
as  
\begin{align*}
\We(\bFe)&=\We(\bF \bFp^{-1})=\We(\bF\bCp^{-1/2}) =\widehat\We\big((\bF
\bCp^{-1/2})^\TT \bF
\bCp^{-1/2}\big)\\
&=\widehat\We(\bCp^{-1/2}\bC\bCp^{-1/2}).
\end{align*}
 On the other hand, the invariance of $\Wh$ under plastic
rotations \eqref{plas_inv} entails  that  $\Wh (\bFp) = \Wh(\bCp^{1/2})$. We
hence define the function  $\haz \Wh: \MM \to
[0,\infty]$ by 
$$ \haz \Wh(\bCp) := \Wh(\bCp^{1/2})$$
and rewrite the {\it energy density}  
\eqref{eq:energy0} as
\begin{align*}
 W(\bC,\bFp) &=\haz W(\bC,\bCp) =\haz \We(\bFp^{-\TT}\bC \bFp^{-1}) +\Wh(\bFp) \\&= \widehat{\We}(\bCp^{-1/2}\bC\bCp^{-1/2})
 +\haz \Wh(\bCp). 
\end{align*}
 The  state of the system is  hence  described by
the pair $$%(\bC,\bFp)\in \GLps\times \SL  \ \ \ \text{or} \ \ \
(\bC,\bCp)\in \GLps\times \MM.$$
Henceforth, we systematically employ the hat superscript in order
to identify quantities written in terms of the Cauchy-Green tensors $\bCe$
and $\bCp$.

%
%
%---------------------
\subsection{Constitutive relations}
%---------------------
%

In order to introduce constitutive relations we shall here follow the
classical {\it Coleman-Noll procedure} \cite{Coleman-Noll63}. Let us assume sufficient
smoothness and
compute
\begin{align}
  &\frac{\d}{\d t} W(\bC,\bFp) = \partial_{\bC} W {:}\dot \bC  + \partial_{\bFp}
  W{:}\dot \bFp=:
\bS{:}\frac12 \dot \bC - \bN{:} \dot \bFp  \nonumber \\
&= 2 \bFp^{-1}  \partial_{\bCe} \haz \We(\bCe) \bFp^{-\TT}{:}\frac12
\dot \bC - \Big(  2 \bCe\partial_{\bCe} \haz \We(\bCe) \bFp^{-\TT}-
2 
\bFp \partial_{\bCp} \haz \Wh(\bCp) \Big){:} \dot \bFp  \label{Clausius1}
\end{align}
or,  by using the variables $(\bC,\bCp)$, 
\begin{align}
 & \frac{\d}{\d t} {\haz W}(\bC,\bCp) =  \partial_{\bC} W {:}\dot \bC  + \partial_{\bCp}
  W{:} \dot \bCp=:
\bS{:}\frac12 \dot \bC - \bT{:} \frac12 \dot \bCp  \nonumber \\ 
&=2 \bFp^{-1}  \partial_{\bCe} \haz \We(\bCe) \bFp^{-\TT}{:}\frac12
\dot \bC-\big(2 \bFp^{-1}
  \bCe \partial_{\bCe}\haz \We(\bCe)\bFp^{-\TT} - 2 \partial_{\bCp}
  \haz \Wh(\bCp) \big) {:}\frac12 \dot \bCp. \label{Clausius2}
\end{align}
We have  here used  the coaxiality of $\bCe$ and $\partial_{\bCe} \haz
\We(\bCe)$, which is in turn a consequence of  the isotropy of
$\We$ from \eqref{plas_inv},  as well
as the symmetry of $\partial_{\bCp} \haz \Wh$. The above computation shows that the classical  {\it
  second Piola-Kirchhoff stress} tensor $\bS$ 
\begin{equation}\label{eqn:piola}
\bS:=2 \bFp^{-1} \partial_{\bCe} \haz \We(\bCe) \bFp^{-\TT}\in\Rzs
\end{equation} 
 is the conjugate variable to
$\bC$ whereas the evolution of $\bFp$ and $\bCp$ is driven by the corresponding
conjugated {\it thermodynamic forces}
\begin{align*}
\bN &:= - \partial_{\bFp} W(\bC,\bFp) =   2 \bCe\partial_{\bCe}
\haz \We(\bCe) \bFp^{-\TT} -
2 
\bFp \partial_{\bCp} \haz \Wh(\bCp), \\
 \bT &:=  -2\partial_{\bCp} W(\bC,\bCp)  =  2 \bFp^{-1}
  \bCe \partial_{\bCe}\haz \We(\bCe)\bFp^{-\TT} - 2 \partial_{\bCp}
  \haz \Wh(\bCp),
\end{align*}
 respectively.  The expression of $\bT$ follows
 along the same computations of
 \cite{Fri-Ste09}. In particular, we use the fact that, for all $\bB\in \Rzn$,
\begin{equation}
  \label{T4}
  \bFp^{-\TT} \bB \bFp^{-1} = \bFp^{-\TT}   (\partial_{\bCp} \bFp^{\TT}{:}\bB) + (\partial_{\bCp} \bFp {:} \bB) \bFp^{-1},
\end{equation}
in order to compute
\begin{align}
& \partial_{\bCp}\haz \We(\bCp^{-1/2}\bC\bCp^{-1/2}){:}\bB=\partial_{\bCe}\haz \We(\bCe) {:}\partial_{\bFp}(\bFp^{-\TT}\bC\bFp^{-1}){:}
\partial_{\bCp}\bFp {:}\bB \nonumber\\
&=-\partial_{\bCe}\haz \We(\bCe) {:}\bFp^{-\TT}(\partial_{\bCp}\bFp^\TT{:}\bB)\bFp^{-\TT}\bC\bFp^{-1}
-\partial_{\bCe}\haz \We(\bCe) {:} \bFp^{-\TT}\bC\bFp^{-1} (\partial_{\bCp}\bFp{:}\bB)\bFp^{-1}\nonumber\\
&=-\partial_{\bCe}\haz \We(\bCe)\bCe{:}\bFp^{-\TT}(\partial_{\bCp}\bFp^\TT{:}\bB)-\bCe\partial_{\bCe}\haz \We(\bCe){:}(\partial_{\bCp}\bFp{:}\bB)\bFp^{-1}
\nonumber\\
&\stackrel{\eqref{T4}}{=}-\bCe\partial_{\bCe}\haz \We(\bCe) {:}\big(\bFp^{-\TT}(\partial_{\bCp}\bFp^\TT{:}\bB)+(\partial_{\bCp}\bFp{:}\bB)\bFp^{-1}\big)
\nonumber\\
&=-\bCe\partial_{\bCe}\haz \We(\bCe){:}\bFp^{-T}\bB\bFp^{-1}=-\bFp^{-1}\bCe\partial_{\bCe}\haz \We(\bCe)\bFp^{-\TT}{:}\bB \nonumber.
\end{align}
Note that the tensors $\bN$ and $\bT$  fulfill the basic relation
\begin{equation}
  \label{NPT}
  \bN = \bFp \bT.
\end{equation}
As   the  tensors
$\bCe \partial_{\bCe} \haz \We(\bCe)$ and $\partial_{\bCp}
  \haz \Wh(\bCp)$ are symmetric,  the tensor $\bT$ is symmetric as
  well.

%
%
%---------------------
\subsection{Flow rule in terms of $ \bFp$}
%---------------------
%
The plastic evolution is formulated in terms of a given  {\it yield
  function} $\phi =\phi(\bFp,\bN): \SL\times \Rzn \to \Rz$ whose sublevel $\{\phi(\bFp,\bN) \leq 0\}$
represents the {\it elastic domain}. We assume that  
for all given $\bFp \in \SL$ the yield function $ \bN \mapsto
\phi(\bFp,\bN) $ is convex
and that $\phi(\bFp,\boldsymbol 0)<0$. 

Given the conjugacy of $\bN$ and $\bFp$ from
\eqref{Clausius1}, we classically prescribe the  flow rule  in
complementarity form as
\begin{equation}\label{eq:flowrule}
\displaystyle \dot\bFp=\dot z\,\partial_{\bN}\phi(\bFp,\bN),\quad 
\displaystyle \dot z\geq 0,\quad \phi\leq 0,\quad \dot z \phi=0. 
\end{equation}
This position falls within the class of
{\it associated}  plasticity models  for  the rate $\dot\bFp$ is prescribed to belong to the
normal cone of the yield surface $\{\phi(\bFp,\bN)=0\}$.
By dualization, this can be equivalently reformulated as 
\begin{equation}
\bN \in \partial_{\dot \bFp}  R(\bFp, \dot
\bFp) \label{combine1}
\end{equation}
where the {\it infinitesimal dissipation} $ R(\bFp, \dot \bFp)$ is the Legendre conjugate of the
indicator function of the  elastic domain $\{\phi(\bFp,\bN)\leq 0\}$ with
respect to its second argument. 

%
%---------------------
\subsection{Flow rule in terms of $ \bCp$}
%---------------------
%
 We now aim at rewriting  the flow rule
\eqref{eq:flowrule} in
terms of $\bCp$ only. This follows  again by assuming plastic-rotation
invariance for the plastic-dissipation mechanism, namely
\begin{equation}
  \label{plas_inv2}
  \phi(\bQ\bFp,\bQ \bN) = \phi(\bFp,\bN) \quad \forall \bQ \in \SO
\end{equation}
and all $\bFp \in \SL$ and $\bN \in \Rzn$. Again, invariance with
respect to rotations in $\bFp$ corresponds to frame indifference
in the intermediate configuration. On the other hand, under assumption
\eqref{plas_inv}  the tensor  $\bQ\bN$ is the
thermodynamic force conjugated to $\bQ \bFp$ via \eqref{Clausius1}.

By using \eqref{plas_inv2}, the decomposition $\bFp = \bR \bCp^{1/2}$,
and relation \eqref{NPT} we define
$$ \phi(\bFp, \bN) = \phi(\bCp^{1/2},\bCp^{1/2}\bT) =: \haz
\phi(\bCp,\bT)$$
and note that the function $\bT \mapsto \haz \phi(\bCp,\bT)$ is convex
and $\haz
\phi(\bCp,\boldsymbol 0) < 0$ for all given $\bCp \in \MM$. We aim now
at showing that a flow rule in terms of $\dotCp$ follows from the flow
rule \eqref{eq:flowrule}.  As $\bT$ is symmetric,  we compute
\begin{align*}
  \dot \bFp = \dot z \partial_{\bN} \phi(\bFp,\bN) \stackrel{\eqref{NPT}}{=}  \dot
  z \partial_{\bN} \haz \phi(\bCp,\bN^\TT\! \bFp^{-\TT}) =  \bFp^{-\TT}
  \dot z \partial_{\bT} \haz \phi(\bCp,\bT)
\end{align*}
where we have also used that $\partial_{\bT} \haz \phi(\bCp,\bT)$ is
symmetric.
Then, one has that 
\begin{equation}\label{eq:f-C}
\displaystyle \dot\bCp=\dot \bFp^\TT\! \bFp + \bFp^\TT\! \dot \bFp = 2 \dot z\,\partial_{\bT}\haz \phi (\bCp,\bT),\quad 
\displaystyle \dot z\geq 0,\quad \haz \phi\leq 0,\quad \dot z \haz \phi=0.
\end{equation}
 This  can be also expressed in the equivalent dual form
\begin{equation}
\bT\in\partial_{\dotCp} \haz R(\bCp,\dotCp)\label{eq:flow0}
\end{equation}
where the {infinitesimal dissipation} $\haz R(\bCp,\dotCp)$ is the Legendre conjugate of the
indicator function of the {elastic domain} $\{\haz \phi(\bCp,\bT)\leq 0\}$ with
respect to $\bT$.
By using the plastic-rotation
invariance \eqref{plas_inv2} we deduce that 
\begin{align*}
  R(\bFp,\dot \bFp) &= \sup \{ \bN {:} \dot \bFp \;|\; \phi(\bFp,
  \bN) \leq 0\}= \sup \{ \bN {:} \dot \bFp  \;|\;\phi(\bQ\bFp,
  \bQ\bN) \leq 0\}\\
&= \sup \{ \bQ\bN {:} \bQ\dot \bFp  \;|\; \phi(\bQ\bFp,
  \bQ\bN) \leq 0\} = R(\bQ\bFp,\bQ \dot \bFp) \quad \forall \bQ \in\SO.
\end{align*}
Correspondingly, we have that 
\begin{align}
  \haz R(\bCp,\dotCp) &= \sup \{ \bT {:} \dotCp  \;|\; \haz \phi(\bCp,
  \bT) \leq 0\}= \sup \{ \bT {:} \dotCp  \;|\;\phi(\bCp^{1/2},
  \bCp^{1/2}\bT) \leq 0\}\nonumber\\
&= \sup \{ \bCp^{1/2}\bT {:} \bCp^{-1/2}\dotCp  \;|\; \phi(\bCp^{1/2},
  \bCp^{1/2}\bT) \leq 0\}\ = R(\bCp^{1/2},\bCp^{-1/2}\dotCp).\label{acco}
\end{align}

Before closing this subsection, let us remark that the combination of frame \eqref{eq:frame-ind} and
plastic-rotation indifference \eqref{plas_inv}, \eqref{plas_inv2} 
entail that the model is invariant under the trasformations
$ \bFe \to \bQ \bFe \bR$ and $ \bFp \to \bR \bFp$
with respect to all $\bQ,\, \bR \in \SO$.
This invariance is already advocated in
\cite{Casey-Naghdi80,Green-Naghdi71} as a natural requirement in
relation with the multiplicative decomposition $\bF = \bFe\bFp$, see
also \cite[Formula (4.5)]{Naghdi90}.

%
%
%---------------------
\subsection{ Choice of the yield function}
%---------------------
%
 We shall now leave the abstract discussion of the previous
subsections and choose the {\it yield
  function} as 
\begin{equation*}\label{eq:f-F}
\phi(\bFp, \bN):=|\dev(\bN\bFp^{\TT})|-r.
\end{equation*}
 Here $r>0$ is a given {\it yield threshold} activating the plastic
evolution.  The latter choice of yield function is inspired by the
classical von Mises theory and has to be traced back to {\it Mandel}
\cite{Mandel72}, see also \cite{Gurtin10}. In particular, for all
given $\bFp \in \SL$, the function $\bN \mapsto \phi(\bFp,\bN)$ is convex
and $\phi(\bFp,\boldsymbol 0) =-r<0$. Moreover, $\phi$ fulfills 
plastic-rotation invariance 
\eqref{plas_inv2}. Correspondingly, the flow rule \eqref{eq:flowrule}
is here specified as
\begin{equation}
\dot\bFp \bFp^{-1}\in\left\{
\begin{array}{ll}
\displaystyle \dot z \, \frac{\dev(\bN \bFp^{\TT})}{|\dev(\bN
  \bFp^{\TT})|} \quad &\textrm{for }
\dev(\bN
  \bFp^{\TT}) \neq 0,\\[3mm]
\dot z\,   \Big\{\bA\in\Rzd \;|\; |\bA|\leq 1\Big\}\ &\textrm{for } \dev(\bN
  \bFp^{\TT}) =0.
\end{array}
\right.\label{combine}
\end{equation}
The infinitesimal dissipation $ R(\bFp, \dot \bFp)$ reads
\begin{align*}
   R(\bFp, \dot \bFp) &= \sup\big\{ \dot \bFp{:} \bN  \;|\; \phi(\bFp,
  \bN)\leq 0\big\}\\
&=\sup\big\{ \dot \bFp \bFp^{-1}{:} \bB  \;|\; |\dev{\bB}| \leq r\big\} = 2\widetilde R (\dot \bFp \bFp^{-1})
\end{align*}
with
\begin{equation}\label{eq:tildaR}
\widetilde R(\bA):=\left\{
\begin{array}{ll}
\disp\frac{r}{2}|\bA|\quad &\textrm{ if }  \tr(\bA)=0,\\
\infty & \textrm{ else}.
\end{array}
\right.
\end{equation}
Note that, owing to the definition of $\bN$ from \eqref{Clausius1},
the flow rule in terms of $\bFp$ reads
\begin{equation}
  \label{eq:222}
  \partial_{\dot \bFp}  R(\bFp, \dot \bFp) + \partial_{\bFp}
  W(\bC, \bFp) \ni \boldsymbol 0.
\end{equation}

Let us now rewrite the flow rule in terms of $ \bCp$. We have that 
\begin{equation*}
\haz \phi(\bCp,\bT)=|\dev(\bCp^{1/2} \bT\bCp^{1/2})|-r 
\end{equation*}
hence the flow rule \eqref{eq:f-C} reads
\begin{align}
\dot \bCp  
&\in 2 \left\{
\begin{array}{ll}
\displaystyle \dot z \, \bFp^{\TT} \frac{\dev(\bN \bFp^{\TT})}{|\dev(\bN
  \bFp^{\TT})|} \bFp \quad &\textrm{for }
\dev(\bN
  \bFp^{\TT}) \neq 0,\\[4mm]
\dot z\,   \Big\{\bA\in\Rzd \;|\;  |\bA|\leq 1\Big\}\ &\textrm{for } \dev(\bN
  \bFp^{\TT}) =0,
\end{array}
\right. \nonumber\\[4mm]
&
=
2\left\{
\begin{array}{ll}
\displaystyle \dot z \, \bCp^{1/2} \frac{\dev(\bCp^{1/2}\bT\bCp^{1/2})}{|\dev(\bCp^{1/2}\bT\bCp^{1/2})|} \bCp^{1/2} \quad &\textrm{for }
\dev(\bCp^{1/2}\bT\bCp^{1/2})\neq 0,\\[4mm]
\dot z\,   \Big\{\bA\in\Rzd \;|\;  |\bA|\leq 1\Big\}\ &\textrm{for } \dev(\bCp^{1/2}\bT\bCp^{1/2})=0.
\end{array}
\right.\label{eq:flowrule20}
\end{align}
Equivalently, by dualization we rewrite the flow rule in the form of \eqref{eq:flow0} 
where the {infinitesimal dissipation} $\haz R(\bCp,\dotCp)$ reads
\begin{align*}\label{eq:R}
\haz R(\bCp,\dotCp)= R(\bCp^{1/2}, \bCp^{-1/2}\dotCp)=\widetilde R(\bCp^{-1/2}\dot\bCp\bCp^{-1/2}) 
\end{align*}
in accordance with \eqref{acco}. In fact, given the specific form of
$\widetilde R$ we also have that  
\begin{equation}\label{eq:R1}
\haz R(\bCp,\dotCp)=\widetilde
R(\bCp^{-1}\dot\bCp)=\widetilde
R(\dot\bCp \bCp^{-1}).
\end{equation}

% 
% In particular, the function $R$ is both left- and right-invariant, namely
% \begin{equation}\label{eq:R2}
% R(\bCp,\dotCp)=R(\bCp\bC,\dotCp\bC)=R(\bC\bCp,\bC\dotCp)\quad\forall\bC\in\MM.
% \end{equation}
% questa non mi pare buona perhce' $\bC \bCp$ non e` simmetrico e la $R$
% e` definita solo sui simmetrici. pero` dopo comunque non si usa.
% 

\subsection{Material constitutive relation}

By the definition of $\bT$ from \eqref{Clausius2}, the flow rule \eqref{eq:flow0} takes the  compact form
\begin{equation}\label{eq:flowrule22}
 \partial_{\dot\bCp}\haz R(\bCp,\dotCp) + \partial_{\bCp}\haz W(\bC,\bCp) \ni \boldsymbol 0.
\end{equation}
Eventually, we have proved that the flow rule written in terms of
$\dot \bFp$ (namely \eqref{eq:flowrule},
\eqref{combine1}, \eqref{combine}, or \eqref{eq:222}) and the flow rule written in terms of
$\dot \bCp$ (namely
\eqref{eq:f-C}, \eqref{eq:flowrule20}, \eqref{eq:flow0}, or
\eqref{eq:flowrule22}) are  such  that solutions
$\bFp$ of the former correspond to solutions $\bCp= \bFp^\TT\! \bFp$ of
the latter. In the following we hence concentrate on the formulation
\eqref{eq:flowrule22}.

Let us stress that the flow rule \eqref{eq:flowrule20} induces an evolution  
in $\MM$.
First of all,  as
$\dot\bCp\bCp^{-1}=2\dot z\bCp^{1/2}\bD\bCp^{-1/2}$ for some $ \bD
\in\Rzd$ with $|\bD| \leq 1$, we have that
$$
\tr(\dot\bCp\bCp^{-1})=2\dot z\,\tr(\bCp^{1/2}\bD\bCp^{-1/2})=2\dot
z\,\tr\bD=0.
$$
This implies by Jacobi's formula that $$\frac{\d}{\d t}\det\bCp=\tr(\dot\bCp\bCp^{-1})=0
.$$ Hence, the evolution preserves the 
determinant constraint.  This in particular entails that eigenvalues
cannot change sign along smooth evolutions,
so that 
positive definiteness is also conserved.
Secondly,
it is clear from the above expression \eqref{eq:flowrule20} that
$\dot\bCp$ is symmetric, so that evolution preserves symmetry as $\bT$ is symmetric. 
 Note that the preservation of the determinant constraint
follows solely from the choice of the flow rule. On the other hand, the symmetric character
of the evolution is a combined effect of the form of the flow rule and of
the energy.

As already commented in the Introduction, the possibility of
reformulating the constitutive
model in terms of $\bCp$ instead of using $\bFp$ is quite advantageous
in terms of computational complexity. Indeed, $\bCp$ belongs to the
five-dimensional connected manifold $\MM$ whereas $\bFp$ is in $ \SL$
which is eight dimensional. Moreover,  this brings also a
computational advantage as matrix computations such as
exponentials, logarithms, and powers are  considerably faster  on $\MM$, see
Appendix A. Finally, the fully Lagrangian formulation in $\bCp$
requires no intermediate configurations. In particular, space
discretizations can be based on the reference configuration only. The reader is referred to the
recent \cite{Neff15} for a comparative discussion of the many
finite-plasticity model based on $\bCp$ available in the literature. The main
result of  \cite{Neff15} consists in proving in the isotropic case that all these
constitutive relations coincide, and coincide to the one of this
paper. Recall however that no isotropy in $\Wh$
is assumed throughout our analysis.

\subsection{Dissipativity}

Thanks to constitutive relation \eqref{eqn:piola}, 
we can express the dissipative character of the model by observing
that 
\begin{align*} 
\frac{\d}{\d t} \haz W(\bC,\bCp) - \bS{:}\frac12 \dot \bC =- \bT{:} \frac12 \dot \bCp \leq 
\haz R(\bCp,\boldsymbol 0)-\haz R(\bCp,\dotCp)=-\haz R(\bCp,\dotCp)\leq 0
\end{align*} 
where we have exploited the very definition of subdifferential and \eqref{eq:flow0}.
In particular, for all
sufficiently smooth evolutions we have that $$\frac{\d}{\d t} \haz
W(\bC,\bCp) \leq \bS{:}\frac12 \dot \bC .$$ 

\subsection{Formulation via the logarithmic plastic strain} By using
the isomorphism
$$\log : \MM \to \Rzd ,$$
the
material constitutive model \eqref{eq:flowrule22} can be equivalently
reformulated in the variables 
$$ (\bC, \log \bCp) \in \GLps \times \Rzd.$$
 An interesting feature of this choice  is that the internal variable
$\log \bCp$ takes values in the linear space $\Rzd$.
% We shall take advantage of
% this fact in discussing the linearization of the model for small
% strains, see
% Sections \ref{sec:lin_const} and \ref{sec:linear} below.

%=====================================================================================
\section{Aside on energetic solutions}\label{sub:es} 
%=====================================================================================
%
We collect here some notation and tools for the
{\it energetic solvability} of general
rate-independent systems. The notion of {\it energetic solutions} is
by now classical \cite{Fra-Mie06,Mie05,Mie-The04} and we limit 
ourselves  in collecting  the minimal material needed along the analysis, by referring to the
above-mentioned papers for all details, motivations, generalizations, and proofs.

Given a product of 
complete metric spaces $\calQ=\calY \times \calZ$, an energy functional  $\calE :\calQ\times [0,T] \to
(-\infty,\infty]$, a dissipation distance
$\calD:\calZ\times \calZ \to [0,\infty]$ (see below for  specific
 assumptions), and an initial datum $q_0\in\calQ$, 
we say that a trajectory $q=(y,z) :[0,T] \to \calQ$ is an {\it energetic
  solution} corresponding to $(\calQ,\calE,\calD)$ starting from $q_0$
if $q(0)=q_0$ and,
for all $t \in [0,T]$, the following two 
conditions are satisfied
\begin{align}
  %&\text{\rm Global stability:} \nonumber\\
&q(t) \in \calS(t):=\Big\{q \in \calQ \;|\: \calE(q,t)<\infty   \ \
\text{and} \ \
\calE(q,t)\leq \calE(\haz q,t) + \calD(q,\haz q) \ \forall \haz q \in
\calQ\Big\},\label{eq:st}\\%[4mm]
%&\text{\rm Energy balance:} \nonumber\\
&\calE(q(t),t) + \textrm{Diss}_{\calD,[0,t]}(q) =
\calE(q(0),0) + \int_0^t \partial_\tau \calE(q(\tau),\tau) \d \tau.\label{eq:eb}
\end{align}
In the latter, we have  denoted  the {\it total dissipation} on
$[0,t]$   by $$\textrm{Diss}_{\calD,[0,t]}(q):=\sup\left\{
  \sum_{i=1}^N\calD(q(t_{i-1}),q(t_i))\right\}$$ where of course
$\calD$ is intended to act just on the $z$ component of $q$ and 
  the supremum is taken over
all partitions $\{0=t_0\leq t_1\leq\ldots\ldots t_N=t\}$ of
$[0,t]$. Condition \eqref{eq:st} is usually referred to as {\it global
  stability}. It expresses the optimality of the current state $q(t)$
against possible competitors $\haz q$ with respect to the
complementary energy,
augmented by the dissipation from $q(t)$ to $\haz
q$. Relation \eqref{eq:eb} imposes the balance between the actual
complementary energy $\calE(q(t),t) $ plus  total  dissipation
$\textrm{Diss}_{\calD,[0,t]}(q)$ and initial energy $\calE(q(0),0) $
plus work of the external actions $\int_0^t \partial_\tau \calE(q(\tau),\tau) \d \tau$. It hence corresponds to {\it energy conservation}.

We refer the reader to the above mentioned classical
references and especially to the recent monograph \cite{Mielke-Roubicek15} for a detailed discussion on the relevance of such a weak
notion of solvability. Here we limit ourselves in observing that the
energetic formulation is totally derivative-free, as no gradients of the
functionals $\calE$ and $\calD$ nor of the trajectory $q$  are
involved.  As
such, it appears to be particularly well-suited for the nonsmooth case
at hand.

A second important property of energetic solutions is that they
naturally arise as limits of {\it incremental minimizations}. Assume to be
given a partition $\{0=t_0<t_1<\dots<t_N=T\}$ of the interval $[0,T]$.
One is interested in incrementally solving the minimization
problems 
\begin{equation}
q_i=\textrm{Argmin}\left\{\calE(q,t_i)+\calD(q_{i-1},q)\;| \; q \in \calQ\right\}\quad \text{for} \
i=1,\ldots,N.
\label{IMP}
\end{equation}
These can be tackled by direct variational
methods and, in particular, have at least a solution
$(q_0,q_1,\ldots,q_{N})$ under suitable coercivity and
lower-semicontinuity assumptions for  the {\it incremental}
functional  $q \mapsto
\calE(q,t_i)+\calD(q_{i-1},q)$. Then, one investigates the convergence
of piecewise constant interpolants $\bar
q^k(t)$ of sequences $(q_0,q_1^k,\ldots,q_{N_k}^k)$ of solutions of
 \eqref{IMP} corresponding to partitions
$\{0=t_0^k<t_1^k<\ldots<t_{N_k}^k=T\}$ with time step $\tau^k=\max(t^k_i{-}t^k_{i-1})$ tending to zero. Under
some specific qualification, see Lemma \ref{lem:abstract-ex} below,
one can prove that 
$\bar q^k$  admits  a convergent subsequence, whose limit is indeed an
energetic solution corresponding to $(\calQ,\calE,\calD)$ \cite{Mie08}.
We shall leave aside the discussion
on the actual capability of energetic solutions of  reproducing
actual  physical behaviors \cite{Mie15,Rou15,Ste09} and limit ourselves in recording that
time-discretization is the most common tool for  
calculating approximate rate-independent evolutions.  The study of
energetic solutions bears a clear relevance for these are limits of
time-discretizations. 

In our analysis, we will  use  the following general
existence result \cite{Mie05,Mie08}. 

\begin{lemma}[Existence result]\label{lem:abstract-ex}
Assume that 
\begin{align}
  &\calD \ \text{is lower
  semicontinuous and nondegenerate in $z$, namely,}\nonumber\\
&  \calD(q,\haz q)\leq \calD(q,\widetilde
q)+\calD(\widetilde q, \haz q) \quad \forall q,\haz q,\widetilde q\in \calQ\quad \text{and} \quad \calD(q,\haz q)=0\
\Leftrightarrow \ z = \haz z,\nonumber\\
&\min\{\calD(q_n,q),\calD(q,q_n)\} \to 0 \ \Rightarrow \ z_n \to z,\label{Diss_ass}\\[3mm]
&\calE\ \text{has compact sublevels and controlled power
  $\partial_t \calE$,
  namely,}\nonumber\\
& 
\{\calE \leq\lambda\}\textrm{ is
  compact in} \ \calQ \quad \forall t\in[0,T],\ \forall \lambda\in\R,\nonumber\\
&\exists c_1>0\;\forall(q_*,t_*) \textrm{ such that }
\calE(q_*,t_*)<\infty:\nonumber\\\nonumber &\calE(q_*,\cdot)\in
C^1(0,T)\textrm{ and } |\partial_t\calE(q_*,t)|\leq
c_1(1{+} \calE(q_*,t))\;\forall t\in[0,T],\\ &\partial_t
\calE : \{ \calE \leq c_1 \} \to \Rz \ \text{is
  continuous},\label{En_ass}\\[3mm]
&\text{Stable states are closed:} \nonumber\\
& q_k\in\calS(t_k) \ \text{and} \ (q_k , t_k ) \to (q_*,t_*) \
\Rightarrow \  q_*\in\calS(t_*). \label{closure_ass}
\end{align}

\noindent Then, for any $q_0\in\calS(0)$ and every partition
$\{0=t_0^k<t_1^k<\ldots<t_{N^k}^k=T\}$  with time step
$\tau^k=\max(t^k_i{-}t^k_{i-1})  $  the incremental minimization problems 
\begin{equation*}
q_i=\textrm{\rm Argmin}\left\{\calE(q,t_i^k)+\calD(q_{i-1}^k,q)\;|\; q
  \in \calQ\right\}\quad \text{for} \
i=1,\ldots,N^k
\label{IMP2}
\end{equation*}
admit a solution $\{q_0,q_1^k,\dots,q_{N^k}^k\}$. As $\tau^k \to 0$,
the corresponding piecewise backward-constant interpolants $ t
\mapsto \bar
q^k(t)$ on the partition admit a not relabeled subsequence such that,
for all $t\in [0,T]$,
$$\bar q^k(t) \to q(t), \quad \textrm{\rm Diss}_{\calD,[0,t]}(\bar q^k) \to
\textrm{\rm Diss}_{\calD,[0,t]}(q), \quad \calE(\bar q^k(t),t) \to
\calE(q(t),t)$$
and $\partial_t \calE(\bar q^k(\cdot),\cdot) \to \partial_t \calE(
q(\cdot),\cdot) $ in $L^1(0,T)$ where $q$ in an energetic solution
corresponding to $(\calQ,\calE,\calD)$ starting from $q(0)$. 
\end{lemma}

In connection with the small-deformation case, we shall be confronted
with the issue of the stability of energetic
solutions with respect to limits. In this concern, we shall be using the
 following  convergence tool
from \cite{MieRouSte08}.

\begin{lemma}[Evolutive $\Gamma$-convergence]\label{lem:evol} Assume
  to be given $\calD_n: \calQ \times \calQ \to [0,\infty]$,
  $\calE_n:\calQ \times [0,T] \to (-\infty,\infty]$ for $n \in
  \Nz_\infty:=\Nz \cup \{\infty\}$ such that $\calD_n$ and $\calE_n$
  fulfill \eqref{Diss_ass} and \eqref{En_ass} uniformly with respect to
  $n\in \Nz_\infty$. Moreover, for all $n\in \Nz$ let $q_{0n} \in \calS_n(0)$ be given,
  where $\calS_n(t)$ indicates the set of stable states at time $t\in [0,1]$
  corresponding to $(\calQ,\calE_n,\calD_n)$, and assume that 
\begin{align}
& \calD_\infty(q , \haz q) \leq \inf \{\liminf_{n\to \infty}
 \calD_n( q_n , \haz q_n) \;|\;\forall (q_n, \haz q_n) \to (
 q,\haz q)\}\quad \forall  q,\, \haz q \in \calQ,\label{gammainfD}\\[2mm]
&  \calE_\infty(  q,t) \leq \inf \{\liminf_{n\to \infty}
 \calE_n( q_n,t_n) \;|\;\forall   (q_n,t_n)  \to (q,t)\}\quad \forall (q,t) \in \calQ\times [0,T],\label{gammainfE}\\[2mm]
  &\forall (q_n,t_n) \to (q,t)  \ \text{s.\,t.} \  q_n \in \calS_n(t_n) \
  \forall \haz q \in \calQ \ \exists \haz q_n \in \calQ: \nonumber\\
&\limsup_{n\to \infty} \big( \calE_n(\haz q_n,t_n){-}
  \calE_n(q_n,t_n){+}\calD_n(\haz q_n,q_n)\big) \leq \calE_\infty(\haz q,t){-}
  \calE_\infty(q,t){+}\calD_\infty(\haz q,q).\label{closure}
\end{align}
Let 
$q_n$ be energetic solutions
corresponding to $(\calQ,\calE_n,\calD_n)$ starting from $q_{0n}$.  
If the initial values are {\rm well-prepared}, namely
$$q_{0n} \to q_{0\infty} \ \ \text{and} \ \ \calE(q_n(0),0) \to
\calE_\infty(q_\infty(0),0),$$
 there exists a not relabeled subsequence $q_n$ such
that  $q_n(t) \to q_\infty(t)$ for all $t \in [0,T]$  where $q_\infty$ is an energetic solution
corresponding to $(\calQ,\calE_\infty,\calD_\infty)$ starting from
$q_{0\infty}$.
Moreover, we have that, for all $t \in [0,T]$,
$$ \textrm{\rm Diss}_{\calD_n,[0,t]}(q_n) \to
\textrm{\rm Diss}_{\calD_\infty,[0,t]}(q_\infty), \quad \calE_n( q_n(t),t) \to
\calE_\infty(q_\infty(t),t)$$
and $\partial_t \calE_n(\bar q_n(\cdot),\cdot) \to \partial_t \calE_\infty(
q_\infty(\cdot),\cdot) $ in $L^1(0,T)$.
\end{lemma}

In particular, the stability of energetic solution in the limit
$n\to \infty$ follows whenever one checks the {\it two separate}
$\Gamma$-$\liminf$ inequalities \eqref{gammainfD}-\eqref{gammainfE}
and the {\it mutual recovery sequence} condition \eqref{closure}.

%
%
%===========================================
\section{Energetic solvability of the constitutive model}\label{sec:solve_const}
%===========================================
%

We devote this section to the discussion of the existence of energetic
solutions to the constitutive model
\eqref{eq:flowrule22}  at the material-point level.   Assume to be given an initial state $\bC_{\rm
  p,0}\in \MM$ as well as the deformation history $t \in [0,T] \mapsto
\bC(t) \in \GLps$. We are here interested in finding energetic
solutions $t \in [0,T] \mapsto
\bCp(t) \in \MM$ to the evolution problem \eqref{eq:flowrule22}, namely
   \begin{equation}\label{eq:flowrule3}
\partial_{\dot\bCp}\haz R(\bCp,\dotCp)+ \partial_{\bCp}\haz W(\bC(t),\bCp) \ni \boldsymbol 0, \quad
\bCp(0)=\bC_{\rm p,0}.
\end{equation} 
To this aim, we set $q = z = \bCp$ (the
elastic variable $y$ plays no role here, indeed simplifying the
argument) and indicate the complementary energy as $E(\bCp,t):=
\haz W(\bC(t),\bCp)$.

We replace the infinitesimal dissipation $R$
 by the \emph{dissipation metric}   $D:\MM
\times \MM \to [0,\infty]$ 
defined through the formula 
\begin{align}\label{eq:D}
&D(\bCp,\hbCp):=\inf\bigg\{\int_0^1\haz R(\bCp(t),\dot\bCp(t))\d t\;|\; \bCp \in
C^1(0,1;\MM),\nonumber\\
&\hspace{50mm}\bCp(0)=\bCp,\;\bCp(1)=\hbCp\bigg\}.
\end{align}
As the function $\haz R(\bCp,\cdot)$ is smooth for $\dot\bCp\neq 0$, positively 
$1$-homogeneous, and has strictly convex square  power,  $D$ results in a 
\emph{Finsler metric} \cite{MRS10}. In particular,  $D$ 
fulfills the  {\it triangle inequality}. Moreover, the actual choice
of $\haz R$ entails that $D$ is symmetric as well, see \eqref{eq:R1}. % Let us stress that $D$
%is defined on $\MM\times \MM$. % Given $\bCp, \, \hbCp
% \in \MM$ such that $\bCp \hbCp = \hbCp\bCp$ we have that 
% \begin{align*}\label{eq:D-inv0}
% D(\bCp,\hbCp)=D(\one,\hbCp\bCp^{-1})=D(\one,\bCp^{-1}\hbCp)=D(\bCp\hbCp^{-1},\one)=D(\hbCp^{-1}\bCp,\one).
% \end{align*}
% On the other hand, the value $D(\one,\hbCp\bCp^{-1})$ is not defined
% if $\bCp$ and $\hbCp$ do not commute.

Before moving on, let us record some basic continuity and boundedness properties of the
dissipation metric $D$ in the following lemma.

\begin{lemma}\label{lem:D0}
The map $D $ fulfills
\begin{equation}
  \label{lipschitz_char}
  D(\bCp,\hbCp) \leq \widetilde R(\log \bCp {-} \log \hbCp) \quad \forall \bCp,
  \, \hbCp \in \MM.
\end{equation}
In particular, $D$ is locally Lipschitz continuous and we have
the bound
\begin{equation}\label{eq:D-pointbound0}
D(\bCp,\hbCp)\leq  2  r(|\bCp|{+}|\hbCp|{+}6)\quad \forall \bCp,
  \, \hbCp \in \MM.
\end{equation}
%\item [ii)] The distance $D$ satisfies the coercivity property (?)
%\begin{equation}
%D(\bC_{\rm p0},\bC_{\rm p1})\geq c|\bC_{\rm p0}-\bC_{\rm p1}|(|\bC_{\rm p0}|^{-1}\wedgwe |\bC_{\rm p1}|^{-1})
%\label{eq:D-coer}
%\end{equation}
\end{lemma}
\begin{proof}
Given $\bCp\in\MM$, let $ \bL  = \log \bCp\in \Rzd$ and define the curve
$t \in [0,1]\mapsto \bC(t)=\exp(t \bL ) \in \MM$  connecting $\one$ and $\bCp$. Note
that  $\tr(\bC^{-1}\dot\bC)=\tr( \bL )=0$, so that $\haz
R(\bCp(t),\dot\bCp(t))=r| \bL | /2$, and
$$D(\one,\bCp)\leq\int_0^1\haz R(\bCp(t),\dot\bCp(t)) \,\d t=\disp\frac{r}{2}| \bL |.$$
An analogous argument entails that $D(\bCp,\one)\leq  r| \bL
| /2 $. 
Let now $\la_3\geq \la_2\geq\la_1> 0$ with $\la_1\la_2\la_3=1$ be the eigenvalues
of $\bCp$. Then, $\mu_i=\log \la_i $ are the eigenvalues of  
$ \bL $. As we have that $\mu_1 + \mu_2 + \mu_3=0$, we deduce
$$
| \bL |\leq  |\mu_1| + |\mu_2| + |\mu_3|\leq 4 \log \la_3\leq 4(\la_3-1)\leq 4|\bCp{-}\bm 1|.
$$
Hence
\begin{equation*}%\label{eq:Destimate0}
D(\one,\bCp)\vee D(\bCp,\one)\leq  2 r|\bCp{-}\bm 1|.
\end{equation*}
By the triangle inequality, we hence obtain   estimate
\eqref{eq:D-pointbound0}. 

Let now $\bCp, \, \hbCp \in \MM$ be given and define $ \bL  = \log \bCp$
and $\haz \bL  = \log \hbCp$. The curve $  t \in [0,1] \mapsto \bA(t):= \exp(t\haz
 \bL  + (1{-}t)  \bL  )\in \MM$ connects $\bCp$ and $\hbCp$  and
 it is such that $\tr (\bA^{-1}\dot \bA)=\tr(\bL{-}\haz \bL)=0$.  Hence 
$$D(\bCp,\hbCp)\leq\int_0^1 \widetilde R(\bA^{-1}(t) \dot \bA(t)
) \d t=\widetilde R( \bL {-}\haz \bL ) = \widetilde R(\log \bCp{-} \log \hbCp)$$
so that the local Lipschitz continuity of $D$ follows from that of the
logarithm, see Appendix A.\end{proof}

We now focus on energetic solutions corresponding to 
$(\MM,E,D)$. 
In the following we will make use of the  assumption 
\begin{align} \label{Kirchhoff}
 &   |\bFe^\TT\!\partial_{\bFe} \We (\bFe)| \leq c_2(1
{+} \We(\bFe)) \quad \forall \bFe \in \GLp
\end{align} 
for some positive constant $c_2$.
Assumption \eqref{Kirchhoff} entails the controllability of the tensor
$\bFe^\TT\! \partial_{\bFe} \We(\bFe)$ by means of the energy. It is a
crucial condition in finite-deformation theories \cite{Bal84b,
  Bal02} and, in particular, is compatible with polyconvexity (see
later on). Let us record that condition \eqref{Kirchhoff}  has already been considered in
the quasistatic context \cite{Fra-Mie06, Mai-Mie09,MieSte13} and that
it  implies 
\begin{equation}\label{Kirchhoff2}
\ |\partial_{\bFe} \We (\bFe) \bFe^\TT| \leq c(1
{+} \We(\bFe)) \quad \forall \bFe \in \GLp
\end{equation}
for some $c$, depending on $c_2$. 
This implication has been proved in \cite[Prop. 2.3]{Bal02} for any
{frame-indifferent} energy function $\We (\bFe)$. With a
completely similar argument, one can prove that, for 
 isotropic functions $\We(\bFe)$, \eqref{Kirchhoff2}
implies \eqref{Kirchhoff} so that these two conditions are equivalent
in the frame of \eqref{plas_inv}. 
We remark that \eqref{Kirchhoff}-\eqref{Kirchhoff2}
imply that $\We$ has polynomial growth \cite[Prop. 2.7]{Bal02}.
Let us
anticipate that additional assumptions on $\We$, in particular polyconvexity and
coercivity, will be introduced later in Section
\ref{sec:energetic}. In addition to the control \eqref{Kirchhoff} we
require $\haz \Wh$ to be coercive.  Namely, we ask that
\begin{equation}
  \label{eq:3}
 \text{the sublevels of $\haz \Wh$ are  compact.}
\end{equation} 
This coercivity condition will  be progressively strengthened
along the analysis, see
\eqref{coercivity_of_Wh2}, \eqref{coercivity_joint}, and \eqref{eq:32} later on.   

 Owing to the abstract existence result of Lemma
\ref{lem:abstract-ex} (here indeed simplified by the fact that $q=z=\bCp$) we have the
following.

\begin{theorem}[Energetic solvability of the constitutive material
  relation]\label{thm:existence0} Assume \eqref{Kirchhoff} and \eqref{eq:3}.
Let the deformation $t\mapsto \bC(t) \in C^1(0,T)$ and the initial
state $\bC_{\rm p,0} \in \calS(0)$ be given,
where $\calS(t)$ denotes stable states at time $t\in [0,T]$ with
respect to  $(\MM,E,D)$. Then, there
exists an energetic solution corresponding to $(\MM,E,D)$ starting from $\bC_{\rm p,0}$.
More precisely, for all partitions
$\{0=t_0^k<t_1^k<\ldots<t_{N^k}^k=T\}$  with time step
$\tau^k=\max(t^k_i{-}t^k_{i-1})  $ the incremental minimization problems 
\begin{equation*}
\bC_{{\rm p},i}=\textrm{\rm Argmin}\left\{E(\bCp,t_i^k) +
  D(\bC_{{\rm p},i-1},\bC_{{\rm p},i})\;|\; \bCp \in \MM \right\}\quad \text{for} \
i=1,\ldots,N^k
\label{IMP22}
\end{equation*}
admit a solution $ \{\bC_{{\rm p},0},\bC_{{\rm p},1}^k,\dots,\bC_{{\rm
    p}, N^k}^k\}$ and, as $\tau^k \to 0$, the corresponding piecewise
backward-constant interpolants  $t\mapsto \obCp^k(t)$  on the partition admit a not relabeled subsequence such that,
for all $t\in [0,T]$,
$$ \obCp^k(t) \to \bCp(t), \quad \textrm{\rm Diss}_{[0,t]}( \obCp^k) \to
\textrm{\rm Diss}_{[0,t]}(\bCp), \quad E( \obCp^k(t),t) \to
 E(  \bCp(t),t), $$
and $\partial_t E( \obCp^k(\cdot),\cdot) \to \partial_t E(\bCp(\cdot),\cdot)  $ in $L^1(0,T)$ where $\bCp$ in an energetic solution.
\end{theorem}

\begin{proof}
  We limit ourselves in checking the assumptions of  Lemma
  \ref{lem:abstract-ex}.   
Conditions \eqref{Diss_ass} follow from Lemma \ref{lem:D0}. As  $\haz \We \geq 0$ and $\haz
  \Wh$ is coercive  by  \eqref{eq:3},   the
  compactness of the sublevels of $E(\cdot,t)$  ensues.  The closure of
  the stable states \eqref{closure_ass} is a consequence of the
  continuity of $E$ and $D$, see again Lemma \ref{lem:D0}.

We are left with the treatment of the power $\partial_tE(\bCp,t)$. 
Let us start by computing 
\begin{align}
 &\partial_t E(\bCp,t) = \partial_t \haz \We(\bCe) = \partial_{\bCe}
  \haz W(\bCe){:}\dot \bCe\nonumber \\&= \partial_{\bCe}
  \haz W(\bCp^{-1/2} \bC(t) \bCp^{-1/2}){:} \bCp^{-1/2} \dot \bC(t) \bCp^{-1/2}.\label{restart}
\end{align}
As $\haz \We \in C^1$ and the square root is
continuous \cite[pag. 23]{Gurtin81}, the continuity of the map $\bCp \mapsto
\partial_tE(\bCp,t) $  follows. 

 In order to prove the bound on the power  in terms of the
energy, recall that $\dot \bCe = \bFp^{-\TT}\dot \bC \bFp^{-1}$. Hence,
we have that 
\begin{align*}
  &\partial_t E(\bCp,t) = \partial_{\bCe}
  \haz W(\bCe){:}\dot \bCe = \frac12 (\bFe^{-1} \partial_{\bFe}
  W(\bFe)){:} (\bFp^{-\TT}\dot \bC \bFp^{-1})\\
&= \frac12 (\partial_{\bFe} W(\bFe) \bFe^{\TT}){:}(\bFe^{-\TT}
\bFp^{-\TT} \dot \bC \bP^{-1} \bFe^{-1}) \\
&= \frac12 (\partial_{\bFe}
W(\bFe) \bFe^{\TT}){:}(\bF^{-\TT} \dot \bC \bF^{-1}).
\end{align*}
% let $\bR,\, \bQ \in \SO $ be given such that $\bFp = \bR
% \bCp^{1/2}$ and $ \bF = \bQ \bC^{1/2}$. Note that $\bCp^{-1/2}\dot
% \bC(t) \bCp^{-1/2} = \bR^\TT \dot \bCe(t) \bR$.  From 
% \eqref{restart} and the fact that $
% 2\partial_{\bCe} \haz \We(\bCe)=\bFe^{-1}\partial_{\bFe}\We(\bFe)
% $  we obtain
% \begin{align*}
%  &\partial_t E(\bCp,t) = \partial_{\bCe}
%   \haz W(\bCe){:}\dot \bCe = \bFe \partial_{\bCe}
%   \haz W(\bCe)\bFe^\TT{:} \bFe^{-\TT}\dot \bCe \bFe^{-1}\\
% &= \frac12 \partial_{\bFe} \We(\bFe) \bFe^\TT{:} \bFe^{-\TT}\bR
% \bCp^{-1/2}\dot \bC  \bCp^{-1/2} \bR^\TT \bFe^{-1}\\
% &= \frac12 \partial_{\bFe} \We(\bFe) \bFe^\TT{:}  \bF^{-\TT} \bFp^\TT \bR
% \bCp^{-1/2}\dot \bC\bCp^{-1/2} \bR^\TT \bFp \bF^{-1}\\
% &=\frac12 \partial_{\bFe} \We(\bFe) \bFe^\TT{:}  \bF^{-\TT}  \dot
% \bC  \bF^{-1}\\
% &=\frac12 \bQ \partial_{\bFe} \We(\bFe) \bFe^\TT \bQ^\TT{:}  \bQ
% \bF^{-\TT}\dot \bC\ \bF^{-1}\bQ^\TT\\
% &=\frac12 \bQ \partial_{\bFe} \We(\bFe) \bFe^\TT \bQ^\TT{:} \bC^{-1/2} \dot \bC \bC^{-1/2}.
% \end{align*}
Note that the map  $t\mapsto \bF^{-\TT}(t) \dot \bC(t) \bF^{-1}(t)$  is bounded
as $t \mapsto \bC(t) \in \GLps$ is $C^1$ and $|\bF^{-1}| = |\bC^{-1/2}|$. 
By exploiting the control \eqref{Kirchhoff2} we get 
\begin{align}
|\partial_t E(\bCp,t) | &\leq \frac12 c_2(1{+} \We(\bFe))|\bC^{-1/2}(t)
\dot \bC(t) \bC^{-1/2}(t)|  \nonumber\\
&\leq c (1{+}\haz \We(\bCe)) \leq c (1{+}
E(\bCp,t))\label{power_control} 
\end{align}
which delivers the required bound.
\end{proof}
\section{Small-deformation limit for the constitutive model}\label{sec:lin_const}
%===========================================
%
 
We turn  now  our attention to the study of the small-deformation
case. The main result of this Section is a rigorous linearization 
limit for the
constitutive model at the material-point level. This will follow from an application of the
evolutive $\Gamma$-convergence Lemma
\ref{lem:evol}. 
Linearization arguments are classically based on Taylor
expansions for energy and dissipation densities. Here we concentrate
instead on the proof of a variational convergence result. Indeed, we are here proving not only
that the driving functionals are converging but, more
significantly, that the whole trajectories converge. This brings to a
rigorous {\it variational justification} of the linearization approach. 

In order to tackle the small-deformation situation, we concentrate on suitably rescaled differences between $\bC$ or $\bCp$ and the 
identity. In particular, given $\e>0$ we
reformulate the problem in the variables
\begin{equation}\label{forma} \bE : = \frac{1}{2\e} (\bC {-}\one) \in \Rzs, \quad \bz :=\frac{1}{2\e} \log\bCp
 \in\Rzd.
\end{equation}
The tensor $\bE$ is nothing but the $\e$-rescaled {\it Green-Saint Venant}
strain. By
assuming 
$y ={\rm id} + \epsi u$ where $u$ is the rescaled displacement of the
body, one has
$$ \bC = (\one {+} \epsi \nabla u)^\TT (\one {+} \epsi \nabla u) =
\one + 2\epsi \nabla u^{\rm sym} + \epsi^2 \nabla u^\TT \nabla u.$$
In particular $\nabla u^{\rm sym}=(\nabla u {+}\nabla u^\TT)/2$ corresponds to $\bE$
to first order.

The choice for $\bz$ is in the same spirit and corresponds to the
$\e$-rescaled {\it Hencky logarithmic (plastic) strain}. Indeed $\bCp = \exp(2\e
\bz)$ so that $\bCp \sim \one +2 \e \bz$ to first order, in analogy
with the definition of $\bC=\one + 2\e \bE$. The different choice for
$\bz$ is motivated by the nonlinear nature of the state space
$\MM$. In particular, we  use here  the fact that the logarithm is
an isomorphism between $\MM$ and $\Rzd$
in order to replace the the nonlinear finite-plasticity state space $\MM$ with the
linear space $\Rzd$, corresponding indeed to the small-deformation
limit. This is crucial in order to avoid the $\e$-dependence in the state spaces.

By using the equivalent variables \eqref{forma} we introduce the
rescaled energy density $W_\epsi : \Rzs \times \Rzd \to [0,\infty]$ as
\begin{align*} W_\epsi (\bE, \bz) &:= \frac{1}{\epsi^2}\haz W(\bC,\bC_{\rm
  p})
\\&\stackrel{\eqref{forma}}{=} \frac{1}{\epsi^2}\haz\We\big(\exp(-\e \bz) (\one{+}2\epsi \bE)
\exp(-\e \bz)\big)  + \frac{1}{\epsi^2} \haz \Wh\big(\exp(2\e \bz) \big) .
\end{align*}

The relevance of this scaling is revealed  for  $\haz \We $
and $\haz \Wh$ twice differentiable at $\one$  by  computing 
Taylor expansions. In particular, by assuming with no loss of
generality that the densities are
normalized so that $\haz \We(\one)  =
\haz \Wh(\one)=0$,  that the reference configuration is stress-free
($\partial_{\bFe}\We(\one)=\boldsymbol 0$), and that the thermodynamic
force $\bT$ conjugated to $\bCp$ vanishes at non-plasticized states
($\partial_{\bCp}\haz \Wh(\one)=\boldsymbol 0$), we compute
\begin{align*}
  \haz\We(\bCp^{-1/2}\bC\bCp^{-1/2}) &= \haz\We\big(\exp(-\e \bz) (\one{+}2\epsi \bE)
\exp(-\e \bz)\big)\\
&=  \frac12 \epsi^2
(\be{-}\bz){:} 4\partial^2_{\bCe} \haz \We(\one)  (\be{-}\bz) + {\rm o}
(\epsi^2)   = \frac12 \epsi^2|\be{-}\bz|^2_{\bbC}+ {\rm o}
(\epsi^2)\\[4mm]
\haz \Wh(\bCp) &= \frac12 \epsi^2 \bz{:}4\partial^2_{\bCp}  \haz \Wh(\one) \bz + {\rm o}
(\epsi^2)=   \frac12 \epsi^2| \bz|^2_{\bbH}+ {\rm o}
(\epsi^2).
\end{align*}
We have here used the fact that $\exp(-\e \bz)= \one -\epsi \bz +
{\rm o} (\epsi) $ and defined the {\it elasticity} $\bbC$
and {\it hardening} tensors $\bbH$ as follows
$$\bbC:=4\partial^2_{\bCe} \haz \We(\one) = \partial^2_{\bFe}
\We(\one), \qquad\bbH:=4\partial^2_{\bCp}
\haz \Wh(\one).$$ 
These fourth-order tensors are clearly symmetric, for they are Hessians. In addition,
 due to frame- and plastic-rotations indifference the tensors  $ \bbC$
 and $\bbH$ present the so-called {\it minor
symmetries} as well, namely
$$ \bbC_{ij\ell k} = \bbC_{\ell kij}= \bbC_{ij k\ell}, \quad  \bbH_{ij\ell k} = \bbH_{\ell kij}= \bbH_{ij k\ell}.$$

As for the dissipation metric,  by rescaling $D$  by $2\epsi$
 we define   $D_\e: \Rzd \times \Rzd \to
[0,\infty]$ as
$$
D_\e(\bz_1,\bz_2):=\frac{1}{2\e}D(\bC_{\rm p1},\bC_{\rm
  p2})\stackrel{\eqref{forma}}{=}
\frac{1}{2\e}D\big(\exp(2\e\bz_1),\exp( 2\e\bz_2)\big).$$
Note that the scaling of the energy  and of the dissipation 
is different for it corresponds for the different homogeneity of these
terms. 

Assume now to be given $t \in [0,T] \mapsto \bE(t) \in
\Rzs\in C^1(0,T)$ and define accordingly  the rescaled complementary 
energy densities $E_\epsi (\bz,t):= W_\e(\bE(t),\bz)$. Moreover, let
the initial values
$\bz_{0\e} \in \calS_\e(0)$ be given, where $\calS_\e(t)$ denotes the
stable states at time $t\in [0,T]$
with respect to $(\Rzd,E_\epsi,D_\epsi)$. By changing back variables
via \eqref{forma} one finds that $\bC_{\rm p,0\e}= \exp(2\epsi \bz_{0\e})
\in \calS(0)$ where the latter denotes the stable states  at  $t=0$
with respect to $(\MM,E/\e^2,D/(2\e))$. In particular, by virtue of Lemma \ref{thm:existence0} there exists an energetic
solution
$t \in [0,T] \mapsto \bz_\e(t) \in
\Rzd$ corresponding to $(\Rzd,E_\epsi,D_\epsi)$ and starting from
$\bz_{0\e}$. We shall term  $\bz_\e$ a  {\it
  finite-plasticity} trajectory in the following.

The focus of this section is to check that finite-plasticity
trajectories $\bz_\epsi$ converge in the small-deformation limit $\epsi
\to 0$ to the unique {\it linearized-plasticity} trajectory. The
limiting linearized model is specified by letting
\begin{align*}
  W_0 (\bE, \bz) := \frac12|\bE{-}\bz|^2_\bbC + \frac12|\bz|^2_\bbH,\quad
  E_0(\bz,t):= W_0(\bE(t),\bz),\quad
D_0(\bz,\haz \bz):= r |\haz \bz {-} \bz|.
\end{align*}
Given $\bz_0\in \Rzd$, one can apply Lemma \ref{lem:abstract-ex}
and find an energetic solution
corresponding to $(\Rzd,E_0,D_0)$ and starting from $\bz_0$.  As $W_0$ is quadratic, the latter energetic solution turns out
to be a strong solution of the  constitutive  relation of
linearized plasticity with linear kinematic hardening
\begin{align}
r \partial|\bz| + (\bbC{+}\bbH) \bz \in \bbC \bE(t), \quad \bz(0)=\bz_0\label{linearized_plasticity}
\end{align}
and it is thus unique \cite{Han-Reddy}. We shall
refer to this solution as the  {\it linearized-plasticity trajectory}
in the following.

 The main result of this section reads as follows. 

\begin{theorem}[Small-deformation limit of the constitutive
  model]\label{sdl} Assume $ \haz \Wh  $ to be coercive in the following
  sense
  \begin{equation}
    \label{coercivity_of_Wh2}
    \haz \Wh\big(\exp(2\bA)\big) \geq
     c_3 |\bA|^2  \quad \forall \bA \in \Rzd
  \end{equation}
where $c_3$ is a positive constant.  Moreover, let 
$\haz \We$ and $\haz \Wh$  have quadratic behavior at identity,
namely
\begin{align}
  &\forall \delta>0 \, \exists c_\delta >0 \, \forall |\bA| \leq
  c_\delta : \nonumber\\
&\left|\haz \We(\one{+}2\bA) {-}\frac12 |\bA|_\bbC^2\right|
  + \left| \haz \Wh\big(\exp(2\bA)\big) {-}\frac12 |\bA|_\bbH^2\right|
  \leq \delta |\bA|^2.  \label{quad_behavior}
\end{align}
Let $\bz_\e$ be finite-plasticity trajectories starting from
{well-prepared} initial data
$\bz_{0\e}\in \calS_\e (0) $, namely
\begin{equation}
  \label{well-prepared}
  \bz_{0\e} \to \bz_0 \in \Rzd \ \ \text{and} \ \ E_\e(\bz_{0\epsi},0)
  \to E_0(\bz_{0},0).
\end{equation}
Then, for all $t\in [0,T]$
$$\bz_\epsi(t) \to \bz(t),\quad {\rm Diss}_{D_\e,[0,t]}(\bz_\e) \to  {\rm
  Diss}_{D_0,[0,t]}(\bz), \quad    E_\e(\bz_{\epsi}(t),t)
  \to  E_0(\bz(t),t)$$
 where $\bz$ is the unique
linearized-plasticity trajectory starting from
$\bz_0$.
\end{theorem}

 Note that  the
coercivity condition \eqref{coercivity_of_Wh2} corresponds to a
quantitative version of the weaker \eqref{eq:3}. Indeed, as $\bA$ is symmetric
and deviatoric, large negative eigenvalues of $\bA$ may arise only
in presence of some large positive eigenvalue. In this case, the norm
the exponential matrix is necessarily large as well. 

Let us also remark that the quadratic behavior \eqref{quad_behavior}
of $\haz \We$ is equivalent to the following
\begin{align} &\forall \delta>0 \, \exists \tilde c_\delta >0 \, \forall |\bA| \leq
  \tilde c_\delta : \quad \left| \We(\one{+}\bA) {-}\frac12 |\bA|_\bbC^2\right|
  \leq \delta |\bA|^2.  \label{quad_behavior2}
\end{align}

Condition \eqref{quad_behavior} implies in particular that $\haz
\We$ and $\haz \Wh$ are twice differentiable at the identity and
\begin{align*}
 &\haz \We (\one)  = \Wh(\one) =0, \ \ \partial_{\bCe}\haz \We
(\one)  = \partial_{\bCp}\haz \Wh(\one) =\boldsymbol 0, \\
& 4\partial^2_{\bCe} \haz \We
(\one) = \bbC, \ \ 4\partial^2_{\bCp} \haz \Wh (\one) = \bbH.
\end{align*}
On the other hand,  these  conditions imply \eqref{quad_behavior}
 in case  $\haz
\We$ and $\haz \Wh$
are $C^2$ in a neighborhood of the identity.

Let us start by preparing some convergence lemmas for the energy
density and the  dissipation metrics.

\begin{lemma}[Convergence of $E_\epsi$] \label{lem:G-limit-E0}
Under the assumptions of Theorem \emph{\ref{sdl}} we have that 
  $E_\epsi \to E_0$ locally uniformly in $\bz$ and uniformly in
  $t$.
\end{lemma}
\begin{proof}
  Let $\bz\in \MM$. We have that
  $\exp(-\e \bz) = \one -\e \bz + \e^2 \bL $, where
  $\bL $ is bounded in terms of $ |\bz| $ only. In particular, we have
  that 
  \begin{align*}
       &\exp(-\e \bz)(\one{+}2\e \bE(t)) \exp(-\e \bz)= (\one
      {-}\e \bz {+} \e^2 \bL ) (\one{+}2\e \bE(t)) (\one {-}\e
      \bz {+} \e^2 \bL)\\
&= \one +\epsi (\bE(t){-}\bz) + \epsi^2 \haz \bL
  \end{align*}
where the matrix $\haz \bL$ is bounded in terms of $\|\bE\|_{L^\infty}$
and $ |\bz| $ only. Let now $\delta>0$ and  $c_\delta>0$ from
\eqref{quad_behavior} be given and let $\epsi$ so small that 
$|\epsi (\bE(t){-}\bz) {+} \epsi^2 \haz \bL|+ |\e
\bz|\leq c_\delta$. Such an $\epsi$ depends on  $\|\bE\|_{L^\infty}$
and $ |\bz| $.  Then, by \eqref{quad_behavior} we have that 
\begin{align*}
 & |E_\epsi(\bz,t){-}E_0(\bz,t)|\\
& = \left| \frac{1}{\epsi^2}\haz \We(
    \one {+}\epsi (\bE(t){-}\bz) {+} \epsi^2 \haz
    \bL)+\frac{1}{\epsi^2} \haz \Wh\big(\exp(2\epsi \bz)\big) -
    \frac12|\bE(t){-}\bz|^2_\bbC - \frac12 |\bz|^2_\bbH \right|\\
&\leq \left|\frac12|\bE(t){-}\bz {+} \epsi \haz
  \bL|^2_\bbC-\frac12|\bE(t){-}\bz|^2_\bbC\right|+\delta |
(\bE(t){-}\bz) {+} \epsi \haz \bL|^2 + \delta |\bz|^2 \leq c(\epsi{+}\delta)
\end{align*}
where the positive constant $c$ depends on  $\|\bE\|_{L^\infty}$ and $
 |\bz| $.  As $\delta>0$ is arbitrary the local uniform convergence
follows. 
\end{proof}
% \begin{proof}
%   Let $\bz\in \MM$ be such that $\exp(2\e \bz) \in K$ and $|\bz|\leq
%   c_z$ for some positive constant $c_z$. We have that
%   $(\one{+}2\e \bz)^{-1/2} = \one -\e \bz + \e^2 \bL \bz^2$, where
%   $\bL $ is bounded in terms of $c_K$ only. In particular, we have
%   that 
%   \begin{align*}
%        &(\one{+}2\e
%       \bz)^{-1/2}(\one{+}2\e \bE(t)) (\one{+}2\e \bz)^{-1/2}= (\one
%       {-}\e \bz {+} \e^2 \bL \bz^2) (\one{+}2\e \bE(t)) (\one {-}\e
%       \bz {+} \e^2 \bL \bz^2)\\
% &= \one +\epsi (\bE(t){-}\bz) + \epsi^2 \haz \bL
%   \end{align*}
% where the matrix $\haz \bL$ is bounded in terms of $\|\bE\|_{L^\infty}, \, c_z$,
% and $c_K$ only. Let now $\delta>0$ and  $c_\delta>0$ from
% \eqref{quad_behavior} be given and let $\epsi$ so small that 
% $|\epsi (\bE(t){-}\bz) {+} \epsi^2 \haz \bL|+ |\e
% \bz|\leq c_\delta$. Such an $\epsi$ depends on  $\|\bE\|_{L^\infty}, \, c_z$,
% and $c_K$ only.  Then, by   \eqref{quad_behavior} we have that 
% \begin{align*}
%  & |E_\epsi(\bz,t){-}E_0(\bz,t)|\\
% & = \left| \frac{1}{\epsi^2}\haz \We(
%     \one {+}\epsi (\bE(t){-}\bz) {+} \epsi^2 \haz
%     \bL)+\frac{1}{\epsi^2} \haz \Wh(\one{+}2\epsi \bz) -
%     \frac12|\bE(t){-}\bz|^2_\bbC - \frac12 |\bz|^2_\bbH \right|\\
% &\leq \left|\frac12|\bE(t){-}\bz {+} \epsi \haz
%   \bL|^2_\bbC-\frac12|\bE(t){-}\bz|^2_\bbC\right|+\delta |
% (\bE(t){-}\bz) {+} \epsi \haz \bL|^2 + \delta |\bz|^2 \leq c(\epsi{+}\delta)
% \end{align*}
% where the positive constant $c$ depends on  $\|\bE\|_{L^\infty}, \, c_z$,
% and $c_K$ only.  As $\delta>0$ is arbitrary the assertion follows.
% \end{proof}

\begin{lemma}[Convergence of $D_\epsi$] \label{lem:G-limit-D0}
Under the assumptions of Theorem \emph{\ref{sdl}} we have that 
  $D_\e\to D_0$ pointwise and  
  \begin{align}
   \widetilde R(\haz \bz{-} \bz) = D_0(\bz,\haz \bz) \leq \liminf_{\epsi \to 0}
D_\e(\bz_\e, \haz \bz_\e) \quad
   \forall (\bz_\e,\haz \bz_\e) \to (\bz, \haz \bz),\label{gamma_inf}
  \end{align}
In particular   $D_\e\to D_0$  in the sense of $\Gamma$-convergence
  \cite{DalMaso93,DeGiorgi79}.
\end{lemma}

%  roba vecchia

%  in $\Rzd\times
%   \Rzd$. More precisely, we have the following
%  \begin{itemize}
%   \item [i)] let $
% (\bz_\e,\widehat\bz_\e)\rightarrow (\bz,\widehat\bz)$; then 
% $$ \liminf_{\e\rightarrow
% 0}D_\e(\bz_\e,\widehat\bz_\e)\geq R(\widehat\bz{-}\bz)
% $$
%   \item [ii)] given $\bz,\widehat\bz\in \R^{3\times 3}_{\rm dev}\cap \R^{3\times 3}_{\rm
% sym}$, let
% $
% \widehat\bz_\e:= \left(
% e^{\frac{\e}{2}(\widehat\bz-\bz)}(\one{+}\e\bz)e^{\frac{
% \e}{2} (\widehat\bz{-}\bz)} -\one\right)/\e.
% $
% Then $(\bz,\widehat\bz_\e)$ is a recovery sequence, namely, ocio
% alla deifinizione forse servono dei 2 
% $$ \liminf_{\e\rightarrow
% 0}D_\e(\bz,\widehat\bz_\e)\leq R(\widehat\bz{-}\bz).
% $$
%  \end{itemize}
% 
% %
\begin{proof}
Let us start by proving the $\Gamma$-$\liminf$ inequality \eqref{gamma_inf}. Assume to be
given $ (\bz_\e,\haz \bz_\e) \to
(\bz, \haz \bz)$ so that, with no loss of generality, $\sup_\e D_\e
(\bz_\e,\haz \bz_\e) <\infty$. % Use the local Lipschitz continuity of
% $D$ from \eqref{lipschitz_char} in order to have that 
% $$ D(\exp(2\e \bz_\e), \exp(2\e \haz \bz_\e)) \leq D(\one{+}2\e
% \bz_\e, \one{+}2\e \haz \bz_\e ) + c\e.$$
For all $\epsi$ small, let $\bC_\epsi
\in C^1(0,1;\MM)$ be such that
$$   (1{-}\epsi)
\int_0^1 \widetilde R(\dot \bC_\e \bC_\e^{-1}) \d t\leq D\big(\exp(2\e \bz_\e), \exp(2\e \haz \bz_\e)\big) $$
along with  $\bC_\e(0)=\exp(2\e \bz_\e)$ and $\bC_\e(1)= \exp(2\e \haz
\bz_\e)$. In particular, $\dot \bC_\e \bC_\e^{-1} \in \Rzd$ almost
everywhere. By possibly reparametrizing time, we can additionally assume
that 
$$ \widetilde R (\dot \bC_\e \bC_\e^{-1}) \leq 2 D\big(\exp(2\e \bz_\e), \exp(2\e \haz \bz_\e)\big). $$
Let us now estimate the distance of $\bC_\e$ from the
identity as follows
\begin{align*}
  &|\bC_\e(t) {-} \one| \leq \int_0^1 |\dot \bC_\e \bC_\e^{-1}|\, |\bC_\e|
  \, \d t + |\bC_\e(0){-}\one|\\
&\stackrel{\eqref{eq:tildaR}}{=} \frac{2}{r} \int_0^1 \widetilde
  R(\dot \bC_\e \bC_\e^{-1}) |\bC_\e| \, \d t + |\exp(2\e  \bz_{0\e}){-}\one|\\
&\leq \frac{2}{r}\sup_{t\in[0,1]}|\bC_\e(t)| \int_0^1 \widetilde
  R(\dot \bC_\e \bC_\e^{-1})  
  \, \d t + c\epsi \leq \frac{2}{r} \sup_{t\in[0,1]}|\bC_\e(t)| \, 2
  D\big(\exp(2\e \bz_\e), \exp(2\e \haz \bz_\e)\big) + c\epsi\\
&\stackrel{\eqref{lipschitz_char}}{\leq}   \frac{2}{r}
\sup_{t\in[0,1]}|\bC_\e(t)|\,  4 \e \widetilde R (\haz
\bz_\e{-} \bz_\e) + c\epsi  = 
  4   \epsi
\sup_{t\in[0,1]}|\bC_\e(t)| \,| \haz \bz_\e{-} \bz_\e|  + c\epsi\\
&\leq   4  \e\big(
  3{+}\sup_{t\in[0,1]}|\bC_\e(t){-}\one| \big)| \haz \bz_\e{-} \bz_\e| + c\epsi.
\end{align*}
Hence, $\bC_\e \to \one$ uniformly. Clearly $\bC_\e^{-1}= \cof
\bCp^\TT \to \cof \one^\TT = \one$ uniformly as well. Let us
now define $\haz \bC_\e = \one + (\bC_\e {-} \one)/(2\e)$ and use
$\dot{\haz \bC_\e} = \dot \bC_\e/(2\e)$ in order to compute that 
\begin{align*}
  |\dot {\haz \bC_\e}| &\leq  |\dot {\haz \bC_\e} \bC_\e^{-1}| \,
  |\bC_\e| \leq c \widetilde R (\dot {\haz \bC_\e} \bC_\e^{-1})
  =\frac{c}{2\epsi} \widetilde R (\dot \bC_\e \bC_\e^{-1}) \\
&\leq
  \frac{c}{2\epsi}  2 D\big(\exp(2\e \bz_\e), \exp(2\e \haz \bz_\e)\big)  =
  c   D_\e
(\bz_\e,\haz \bz_\e)  \leq c.
\end{align*}
Up to a not relabeled subsequence we have that $\haz \bC_\e \to
\haz \bC$ weakly star in $W^{1,\infty}(0,1;\MM)$. By making use of the
lower-semicontinuity Lemma \ref{lem:semicon-tool} we conclude that 
\begin{align*}
 & \liminf_{\e \to 0}  D_\e
(\bz_\e,\haz \bz_\e) = \liminf_{\e \to 0} \frac{1}{2\e} D\big(\exp(2\e
\bz_\e), \exp(2\e \haz \bz_\e)\big) \geq \liminf_{\e \to 0} \frac{1}{2\e} \int_0^1 \widetilde
R(\dot \bC_\e \bC_\e^{-1}) \,\d t \\
&=  \liminf_{\e \to 0} \int_0^1 \widetilde
R(\dot {\haz \bC_\e} \bC_\e^{-1})\, \d t \geq \int_0^1 \widetilde R
(\dot {\haz \bC})\,\d t \geq \widetilde R( \haz \bC(1){-}  \haz \bC(0)).
\end{align*}
Relation \eqref{gamma_inf} follows by noting that 
$$\haz \bC(1){-} \haz \bC(0) = \lim_{\e \to 0} (\haz \bC_\e(1){-} \haz \bC_\e(0))
= \lim_{\e \to 0} \frac{\exp(2\e \haz \bz_\e) {-} \exp(2\e
  \bz_\e)}{2\e} = \haz \bz {-} \bz.$$

We now turn to the proof of the pointwise convergence.  Let 
$(\bz,\haz \bz) \in\MM \times \MM$ be given
and recall that
\begin{align*}
  D_\e( \bz, \haz \bz ) =  
  \frac{1}{2\e} D(\exp(2\e \bz),\exp(2\e \haz \bz ))  
\stackrel{\eqref{lipschitz_char}}{\leq}\widetilde R(\haz
\bz{-}\bz).
\end{align*}
By using also \eqref{gamma_inf} we conclude that $D_\e(\bz_\e,\haz
\bz_\e) \to \widetilde R(\haz
\bz{-}\bz) = D_0(\haz \bz {-} \bz)$.
\end{proof}
%
%
%  Let $\bz, \, \haz \bz \in
% \Rzd$. The pair $(\bz,\tilde \bz_\epsi)$ where
% $$ \tilde \bz_\e := \frac{1}{2\e}\left({\rm e}^{2\epsi(\haz \bz{-}
%     \bz)}(\one{+} 2\epsi \bz) - \one \right)$$ 
%   was used as a recovery sequence in \cite[Lemma 3.4]{MieSte13}. In
%   particular, $\tilde \bz_e \to \haz \bz$. This
%   choice is here not allowed as $\tilde \bz_\epsi$ need not be
%   symmetric. We hence define 
% $$\haz \bz_\e :=\frac{1}{2\e}\left({\rm e}^{\epsi(\haz \bz{-}
%     \bz)}(\one{+} 2\epsi \bz) {\rm e}^{\epsi(\haz \bz{-}
%     \bz)} - \one \right)$$
% which is symmetric and such that $|\haz \bz_\epsi {-} \tilde \bz_\e|
% ={\rm o}(\epsi)$. In particular, 
% $$\limsup_{\epsi\to 0}D_\epsi(\bz_\e,\haz \be_e) =  dots$$

%  As previously observed, there is no formal difference with the model considered in
% \cite{MieSte13}. In particular the same key properties (2.10a-c) and (2.12) of the
% cited paper are satisfied in this case (see Lemma \ref{lem:D}). Note however that in this
% case the range of $\bz$ is restricted to symmetric matrices. The recovery sequence
% $\widehat\bz_\e$ can also be chosen symmetric. In fact, the sequence
% $$
% \widetilde\bz_\e:=\frac{1}{\e}\left[ e^{\e(\widehat\bz-\bz)}(\one{+}\e\bz) -\one\right].
% $$
% was  proven to be a recovery sequence in \cite{MieSte13}. However
% $$
% \widetilde\bz_\e-\widehat\bz_\e=e^{\frac{\e}{2}(\widehat\bz-\bz)}[\bz,e^{\frac{\e}{2}
% (\widehat\bz-\bz) }] =O(\e).
% $$
% By the triangle inequality and  the Lipschitz continuity of $D$, it easily proved that
% also $\widehat\bz_\e$ is a recovery sequence.
% \end{proof} 

\begin{proof}[Proof of Theorem \emph{\ref{sdl}}]
We aim at applying  the evolutive $\Gamma$-convergence Lemma \ref{lem:evol}
to the sequence of functionals $(E_\e,D_\e)$ and the corresponding limit
$(E_0,D_0)$. Note that these  are defined  on the common
state space $  \Rzd$. Conditions \eqref{Diss_ass} readily follow for
all $D_\e$ and  $D_0$. The uniform compactness of the sublevels of
$E_\e$ is a consequence of the coercivity \eqref{coercivity_of_Wh2}
whereas the smoothness and the uniform control on the powers can be obtained by arguing
along the lines of the proof of Theorem \ref{thm:existence0}, in
particular as in \eqref{power_control}.

The $\Gamma$-$\liminf$ properties \eqref{gammainfD}-\eqref{gammainfE}
follow  from Lemmas \ref{lem:G-limit-E0}-\ref{lem:G-limit-D0}. As for
the closure condition \eqref{closure} assume to be given
$\bz_\e\in \calS(t_\e)$ so that $(\bz_\e, t_\e) \to (\bz,t)$ and $\haz
\bz \in \Rzd$. Then, by choosing the constant (mutual recovery) sequence $\haz \bz_\e=\haz \bz$ we readily
compute that 
\begin{align*}
  &E_\e(\haz \bz_\e, t_\e)- E_\e(\bz_\e,t_\e) + D_\e(\bz_\e,\haz \bz_\e)
  \\
&=
  E_\e( \haz \bz_\e, t_\e)- E_\e(\bz_\e,t_\e) +
  \frac{1}{2\e}D\big(\exp(2\e\bz_\e),\exp(2\e\haz \bz)\big) \\
&\stackrel{\eqref{lipschitz_char}}{\leq} E_\e(\haz \bz, t_\e)-
E_\e(\bz_\e,t_\e) + \widetilde R (\haz \bz {-} \bz_\e).
\end{align*}
In particular, by exploiting the local uniform convergence $E_\e \to E_0$
from Lemma \ref{lem:G-limit-E0}, the smoothness of $\We$ and $\bE$,
and the continuity of $\widetilde R$, we conclude for \eqref{closure}. 
\end{proof}

Before closing this section, let us  record here the convergence
of the energy densities,  which will turn out useful in Section \ref{sec:linear}.

\begin{lemma}[Energy densities convergence]\label{lem:gammaliminfineq}
   Under the assumptions \eqref{quad_behavior} we have the continuous convergences 
  \begin{align}
&\bA_\e \to \bA \ \Rightarrow \ \frac{1}{\e^2}\We( \one {+}\e
\bA_\e) \ \to \ \frac12 |\bA|^2_{\mathbb C}, \label{Gamma_manca_We}\\
&\bB_\e \to
\bB \  \Rightarrow \ \frac{1}{\e^2} \haz \Wh(\exp(2\e \bB_\e))
 \ \to \ \frac12 |\bB|^2_{\mathbb H}. \label{Gamma_manca_Wh}
  \end{align}
Indeed, the latter convergences are locally uniform.
\end{lemma}
\begin{proof}

Let $\bB_\e
\to \bB$, define $M= \sup |\bB_\e|$, and fix $\delta >0$. Due to the
quadratic behavior \eqref{quad_behavior} of $\haz \Wh$, for all $\e <
\delta /M$ we have 
$$\left| \frac{1}{\e^2} \haz \Wh(\exp(2\e\bB_\e)) {-}
  \frac12|\bB_\e|^2\right| \leq \delta M^2$$
and relation \eqref{Gamma_manca_Wh} follows by passing to the
$\limsup$ in $\e$ as $\delta$ is arbitrary.
The proof of \eqref{Gamma_manca_We} is analogous and follows by
considering the equivalent \eqref{quad_behavior2}.
\end{proof}

%===========================================
\section{Quasistatic evolution}\label{sec:quasi}
%===========================================
%
We turn now to the formulation of the quasistatic evolution problem. In
particular, we address the coupling of the constitutive model at the
material point and the boundary-value problem for the quasistatic
elastic response.

\subsection{Quasistatic equilibrium system}
 Let
$\Gamma_{\rm tr}, \,\Gamma_{\rm D} \subset\partial\O$  be measurable and relatively open in
$\partial \Omega$ so that $\Gamma_{\rm D}\subset\partial\O$ has positive surface measure, $\Gamma_{\rm tr}
\cap\Gamma_{\rm D} = \emptyset$, and $\Gamma_{\rm
  tr}\cup\Gamma_{\rm D}=\partial\O$. The set $\Gamma_{\rm tr}$
represents  the portion of the boundary where
traction is exerted. On the other hand,  
$\Gamma_{\rm D}$ is the boundary part
subject to no deformation. 

We neglect inertial effects and concentrate on a
quasistatic approximation of the time evolution.  The  evolution is hence expressed via the 
equilibrium system
\begin{equation}
\nabla{\cdot} \bm{\sigma}+b(t)=0 \quad \text{in}  \ \ \Omega \times (0,T).\label{eq:eq}
\end{equation}
along with the position of the {\it total stress}
$$
\bm{\sigma}:= \partial_{\nabla y}\We(\nabla y \bCp^{-1/2}) = \partial_{\bFe}\We(\nabla y \bCp^{-1/2}) \bCp^{-1/2},
$$
and of the boundary conditions
\begin{align}
y&={\rm id} \quad \text{in} \ \ \Gamma_{\rm D} \times (0,T),\label{eq:BC1}\\
 \bm{\sigma }  \nu  &={\tau}(t)\quad \text{in} \ \ \Gamma_{\rm tr} \times (0,T) \label{eq:BC2}
\end{align}
where $b : \Omega \times (0,T) \to \Rz^3$ and $\tau : \Gamma_{\rm tr}
\times (0,T) \to \Rz^3$ are given body force and traction,
respectively, and 
$  \nu $ indicates the  outward normal to $\partial \Omega$.
Let us comment that our choice of boundary conditions is dictated by simplicity. Nonconstant
imposed boundary deformations could be considered as well by
following \cite{Fra-Mie06}.
%
%-----------------------------
\subsection{Energy}
%-----------------------------
%
We specify the {\it total energy} of the medium as  
\begin{equation}\label{eq:energy1}
 \calE( y,\bCp,t)=\calW( y,\bCp)-\l\ell(t), y\r
\end{equation}  
namely as the sum of the {\it stored} energy $\calW( y,\bCp)$ and the
 work of  {\it external action}  $\l\ell(t), y\r$. 

The  stored energy $\calW( y,\bCp)$ is defined as 
\begin{equation}\label{eq:int-energy} 
 \calW( y,\bCp)=\int_\O
\left(\,\We(\nabla y\bCp^{-1/2})+\haz \Wh(\bCp)+\frac{\mu}{r}|\nabla\bCp|^r\,\right)\dx
\end{equation}
and results form the sum of the {\it elastic}, the  {\it plastic},
 and a
{\it gradient} energy term. The first two terms  have already been  introduced above.

As anticipated in the Introduction, at the quasistatic evolution
level the presence of the gradient term in
$\nabla \bCp$ with $\mu>0$ and $r \geq 1$ sets our problem within the class of gradient
  plasticity models \cite{Fleck1,Fleck2,Aifantis91}. A gradient term $\nabla \bFp$ is 
considered in the quasistatic evolution analysis in
\cite{Mai-Mie09} as well. Although the compactifying effects of the two terms $\nabla \bCp$ and
$\nabla \bFp$  are  comparable, note that such terms deliver different contributions. Indeed, one can compute 
\begin{align*}
  |\nabla \bCp| = |\nabla (\bFp^\TT\!\bFp)| = |(\bFp^\TT \nabla
  \bFp^\TT)^{\TT}  + \bFp^\TT \nabla \bFp| \leq |\bFp^\TT \nabla
  \bFp^\TT|+ |\bFp^\TT \nabla \bFp| \leq 2 |\bFp|\, |\nabla \bFp|.
\end{align*}
In particular, the term $|\nabla \bFp|$ controls $|\nabla \bCp|$
for bounded $\bFp$. On the other hand, the term $|\nabla \bCp|$
vanishes on $\SO$, hence cannot control $|\nabla
\bFp|$.

The  evolution is steered  by  external actions.  In
particular, the  external work reads  
\begin{equation*}\label{eq:ell}
 \l\ell(t), y\r:=\int_{\O} b(x,t){\cdot}  y(x)\, \dx+ \int_{\Gamma_{\rm tr}}\tau(x,t){\cdot} y(x)\,
\d S
\end{equation*}
where $\d S$ denotes the surface measure on $\partial \Omega$.% In particular, 
%$\Gamma_{\rm tr}$ denotes the portion of $\partial \Omega$ where
%traction $\tau$ is exerted. We shall impose a given deformation $y_{\rm D}$ on the
%portion $\Gamma_{\rm D}$ of the boundary so that $\overline\Gamma_{\rm tr}\cup \overline \Gamma_{\rm D}=\partial\O$.
% 
 
%
%---------------------
\subsection{Flow rule}
%---------------------
%
Let us introduce a notation for the {\it total energy density}  in \eqref{eq:int-energy} as
\begin{equation*}\label{eqn:W_C_Cp}
\haz W(\bC,\bCp,\nabla\bCp)= \widehat{\We}(\bCp^{-1/2}\bC\bCp^{-1/2})
 +\haz \Wh(\bCp)+\frac{\mu}{r}|\nabla\bCp|^r.
\end{equation*}
Then, the {\it second Piola-Kirchhoff stress} tensor $\bS$ is again
defined as
\begin{equation*}\label{eqn:piola2}
\bS:=2 \partial_{\bC}\haz W(\bC,\bCp,\nabla\bCp)=\bCp^{-1/2}{:} 2\,\partial_{\bCe}\widehat{\We}(\bCe){:}\bCp^{-1/2}\in\Rzs.
\end{equation*}
Let again $\bT$ denote the thermodynamic force conjugated to
$\bCp$. Because of the
gradient term $\nabla\bCp$, the partial derivative is replaced by the functional
variation and $\bT$ can be split into a {\it local} and {\it nonlocal} part
as follows
\begin{equation*}\label{eqn:T}
\bT:=-\delta_{\bCp}\haz W= -\partial_{\bCp} \haz
W+\nabla{\cdot}\partial_{\nabla\bCp}\haz W\in\Rzs.
\end{equation*}
Here, the symbol $\delta_{\bCp}$ refers to some suitable functional variation.
The flow rule \eqref{eq:flowrule22} takes here the form
\begin{equation*}\label{eq:flowrule2}
\partial_{\dot\bCp}\haz R(\bCp,\dot \bCp)  +  \delta_{\bCp}\haz W(y,\bCp) \ni \boldsymbol 0
\end{equation*}
where
$$
\delta_{\bCp}\haz W=\partial_{\bCp} W_e(\nabla
y\bCp^{-1/2})+\partial_{\bCp}\haz W_h(\bCp)-\mu\nabla{\cdot}( |\nabla\bCp|^{r-2}\nabla\bCp).
$$
Owing to the above introduced notation, we can formally summarize the
relations governing the quasistatic evolution problem as
\begin{align}
&\nabla{\cdot} \big(\partial_{\bFe} \We(\nabla y\bCp^{-1/2})
\bCp^{-1/2}\big)+ b(t)={\boldsymbol 0},\label{eq:diff-evolution}\\
& \partial_{\dot\bCp}\haz R(\bCp,\dot\bCp)+\partial_{\bCp}\We(\nabla
 y\bCp^{-1/2})+\partial_{\bCp}\haz \Wh(\bCp)-\mu\nabla{\cdot}( |\nabla\bCp|^{r-2}\nabla\bCp)\ni 0, \label{eq:diff-evolution2}
\end{align}
to be complemented with boundary and initial conditions.
This strong formulation \eqref{eq:diff-evolution}-\eqref{eq:diff-evolution2} of the quasistatic evolution problem seems at present inaccessible
from the point of view of the existence of solutions.  We are
hence forced in resorting to a weak-solution concept.

%%
%
%=====================================================================================
\section{Energetic solvability of the quasistatic-evolution problem}\label{sec:energetic}
%=====================================================================================
%

The aim of this section  is to present  the {energetic
  formulation} of system
\eqref{eq:diff-evolution}-\eqref{eq:diff-evolution2}.  
We start by introducing suitable functional spaces for the state
variables. Let  the coefficients 
$q_y,q_p,r>1$ be given,  specific assumptions are introduced
below.   We will ask  that   the
deformation $y$  belongs  to
\begin{equation*}\label{eq:calY}
\calY:=\{y\in W^{1,q_y}(\O,\R^3)\;|\; y ={\rm id} \ \text{on} \
\Gamma_{\rm D}\}, 
\end{equation*}
The state space for $\bCp$ is then defined by
\begin{equation*}\label{eq:spaceCp}
\calZ=\{\bCp\in L^{q_p}(\O,\R^{3\times 3})\cap W^{1,r}(\O,\R^{3\times
  3}) \;|\; \bCp(x)\in\MM \ \text{for a.e.} \ x \in \Omega\}.
\end{equation*}
$\calY$ and $\calZ$ are weakly closed subsets of separable and
reflexive Banach spaces.  Finally, we  set $\calQ := \calY \times \calZ$. 

The total energy $\calE: \calY\times \calZ \times [0,T] \to (-\infty,\infty]$ has been introduced in \eqref{eq:energy1}. As for the
dissipation metric we let $\calD: \calZ \times \calZ \to [0,\infty]$
be defined by 
\begin{equation*}\label{eq:D1}
\calD(\bCp,\hbCp):=\int_\Omega D(\bCp(x),\hbCp(x))\, \d x
\end{equation*}
where $D$ is the dissipation metric introduced in \eqref{eq:D}. The
local Lipschitz continuity of $D$, see Lemma \ref{lem:D0}, translates
to an analogous statement for $\calD$.

% We are hence interested in a energetic solution $(y,\bCp)$ corresponding to $(\mathcal
% Y\times\calZ,\calE,\calD)$ and starting from some initial value
% $(y_0,\bCp^0)$. This is 
% trajectory $(y,\bCp):[0,T]\rightarrow \mathcal Y\times\calZ$ such that
% $(y(0),\bCp(0))=(y_0,\bCp^0)$ and, for all $t \in [0,T]$, global
% stability \eqref{eq:st} and energy balance \eqref{eq:eb} hold. These
% conditions read in this case 
% \begin{align}
% &  (y(t),\bCp(t))\in\calS(t) :=\Big\{(y,\bCp)\in\mathcal
% Y\times\calZ\ \text{such that} \ \calE(y,\bCp,t)\leq\infty, \nonumber\\
% &\hspace{20mm}\calE(y,\bCp,t)\leq \calE(\widehat y,\widehat \bCp,t)+\calD(\bCp,\widehat \bCp)\quad
% \forall (\widehat y,\widehat\bCp)\in \mathcal Y\times\calZ\Big\}\label{eq:stability}
% \\[2mm]
% &\quad \calE(y(t),\bCp(t),t)+
% \textrm{Diss}_{[0,t]}(\bCp)=\calE(y(0),\bCp(0),0)+\int_0^t\partial_\tau
% \calE(y(\tau),\bCp(\tau),\tau)\,\d\tau \label{eq:energy}
% \end{align}
% where the \emph{total dissipation} now reads
% \begin{equation}
% \textrm{Diss}_{[0,t]}(\bCp):=\sup\left\{
%   \sum_{i=1}^N\calD(\bCp(t_{i-1}),\bCp(t_i)) , \quad 0=t_0\leq t_1\leq\ldots\ldots t_N=t \right\},\label{Diss}
% \end{equation}
%  the supremum being taken over
% all partitions of $[0,t]$.

% %
% %
% %=============================================================
% \section{Existence of energetic solutions}\label{sec:existence}
% %==============================================================
% %

The main result of this section reads as follows.

\begin{theorem}[Energetic solvability of the quasistatic
  system]\label{teo:existence1} Assume polyconvexity of $\We$ and
  coercivity of $\We$ and $\haz \Wh$, namely
  \begin{align} 
  & \We(\bFe)=\mathbb W(\bFe,{\rm cof}\bFe,\det\bFe) \quad \forall \bFe \in
    \GLp\nonumber\\
& \text{for some} \
    \ \bbW:\R^{3\times 3}\times \R^{3\times 3}\times
    \Rz_+\rightarrow\R \ \ \text{convex}\label{eq:poly},\\[2mm]
&\We(\bFe)\geq c_4 |\bFe |^{q_e} - \frac{1}{c_4},\qquad \haz \Wh(\bCp)\geq
c_4 |\bCp |^{q_p}  -\frac{1}{c_4} \label{coercivity_joint}
  \end{align}
for some positive constant $c_4$, and  $q_e, \, q_p >1$.  Moreover, assume that  
\begin{equation}\label{eq:indices}
\frac{1}{q_y}=\frac{1}{q_e}+\frac{1}{2q_p},\quad q_y>3,\qquad r>1.
\end{equation}
Eventually, let $\ell\in C^1([0,T];(W^{1,q_y}_{\Gamma_{\rm
    D}}(\O;\R^3))^*)$ and $(y_0,\bC_{\rm p,0})\in \calS(0)$ where $\calS(t)$
denotes the set of stable states at time $t \in [0,T]$ with respect to
$(\calQ,\calE,\calD)$ at time $t$.
 Then, there
exists an energetic solution corresponding to $(\calQ,\calE,\calD)$
starting from $(y_0,\bC_{\rm p,0})$.
More precisely, for all partitions
$\{0=t_0^k<t_1^k<\ldots<t_{N^k}^k=T\}$  with time step
$\tau^k=\max(t^k_i{-}t^k_{i-1})  $ the incremental minimization problems 
\begin{align*}
(y_i,\bC_{{\rm p},i})&=\textrm{\rm Argmin}\left\{\calE(y,\bCp,t_i^k) +
  \calD(\bC_{{\rm p},i-1},\bC_{{\rm p},i})\;|\; (y,\bCp) \in \calQ
\right\}\nonumber\\
& \text{for} \
i=1,\ldots,N^k
\label{IMP2}
\end{align*}
admit a solution $ \{(y_0,\bC_{{\rm p},0}),(y_1^k,\bC_{{\rm p},1}^k),\dots,(y_{N^k}^k,\bC_{{\rm
    p}, N^k}^k)\}$ and, as $\tau^k \to 0$, the corresponding piecewise
backward-constant interpolants $t \mapsto (\overline
y^k(t),\obCp^k(t))$  on the partition admit a not relabeled subsequence such that,
for all $t\in [0,T]$,
\begin{align*} &(\overline y^k(t),\obCp^k(t) )\to (y(t),\bCp(t)), \quad \textrm{\rm Diss}_{\calD,[0,t]}( \obCp^k) \to
\textrm{\rm Diss}_{\calD,[0,t]}(\bCp), \\
& \calE(\overline y^k(t), \obCp^k(t),t) \to
 \calE(y(t),  \bCp(t),t), \end{align*}
and $\partial_t \calE( \overline y^k(\cdot),\obCp^k(\cdot),\cdot)
\to \partial_t \calE(y(t),\bCp(\cdot),\cdot)  $ in $L^1(0,T)$ where
$(y,\bCp)$ in an energetic solution  corresponding to
$(\calQ,\calE,\calD)$. 
\end{theorem}

Before moving on let us  comment that  an elastic energy density
$\We$ satisfying
\eqref{eq:poly}-\eqref{coercivity_joint}  can be found, for instance, within the
class of
\emph{Ogden materials}    \cite[Sec. 4.9]{Ciarlet} corresponding indeed to
the choice
\begin{align*}%\label{eqn:Ogden}
&\We(\bFe):=\widehat{\We}(\bCe):=\sum_{i=1}^{n}a_i \tr \bCe^{\g_i/2}+\sum_{j=1}^{m}b_j \tr (\cof
\bCe)^{\delta_j/2}+\Gamma(\det\bCe^{1/2}),\nonumber \\
&n,m\geq 1 ,\quad a_i,b_j>0,\quad\g_i,\delta_j\geq 1,\quad s\mapsto\Gamma(s) \textrm{ convex
on }(0,\infty), \quad \lim_{s\rightarrow 0^+}\Gamma(s)=\infty.
\end{align*}
Clearly, frame indifference and isotropy are fulfilled. Moreover the function
$\We(\bFe)$ is 
polyconvex as $\g_i,\delta_j\geq 1$. Note  that $\g_i\geq 1$ for
some $i$ 
implies  
$q_e\geq 2$  in condition \eqref{coercivity_joint}.

We shall now prepare some lemmas to be used in the proof of Theorem
\ref{teo:existence1}.   

\begin{lemma}[Coercivity of the energy]\label{lem:coercive}
Under the assumptions of Theorem
\emph{\ref{teo:existence1}}, the energy $\calE$ is coercive in the following sense
\begin{equation}\label{eq:coercive}
\calE(y,\bCp,t)\geq c_5\|y\|^{q_y}_{W^{1,q_y}}+c_5\|\bCp\|^{q_p}_{L^{q_p}}-\frac{1}{c_5}
\end{equation}
where $c_5$ is a positive constant.
\end{lemma}

\begin{proof}
From the coercivity assumption \eqref{coercivity_joint}, an
application of Young's inequality gives
\begin{align*}
& \frac{1}{c_4}\We(\bFe)+\frac{1}{c_4^2}\geq |\nabla y\bCp^{-1/2}|^{q_e}\geq \left(|\nabla
y|/|\bCp^{1/2}|\right)^{q_e}\geq 3^{-1/2}(|\nabla
y|^{q_e}/|\bCp|^{q_e/2})\geq \\
&3^{-1/2}\left((1+t)\delta^{t/(t+1)}|\nabla y|^{q_e/(t+1)}-t \delta
|\bCp |^{q_e/2t}\right)
\end{align*}
for any $\delta$ and $t$ positive. Moreover, $\haz \Wh(\bCp)\geq c_4 |\bCp
|^{q_p}-1/c_4$. By taking
$\delta$ sufficiently small and $t$ such that $q_e/2t=q_p$, we obtain
$$
\We(\bFe)+\haz \Wh(\bCp)\geq c  |\nabla y |^{{2q_e
    q_p}/{(2q_p+q_e)}}+c  |\bCp |^{q_p} - c,
$$
therefore
$$
\calW(y,\bCp)\geq c\|\nabla y\|^{q_y}_{L^{q_y}}+
\|\bCp\|^{q_p}_{L^{q_p}} -c.
$$
As 
$
|\langle\ell,y\rangle |\leq \|\ell\|\|y\|_{W^{1,q_y}}
$,
by virtue of Korn's inequality the statement follows.
\end{proof}

 Note   that indeed the coercivity lower bound
\eqref{eq:coercive} holds under the weaker condition $
{1}/{q_y}={1}/{q_e}+{1}/{(2q_p)}<1$ as well.
%
%As for the kinematic hardening term $\haz \Wh$, we can refer for example to the choice in
%\cite{Reese&al08}
%\begin{equation}
%\haz \Wh(\bCp)=\frac{c}{2}\left[\textrm{Tr}(\bCp)-3-\ln(\textrm{det}\bCp)\right],
%\end{equation}
%which satisfies $\haz \Wh\geq\alpha|\bCp|-C$.\\
%** TUTTAVIA $q_y=1$ NON VA BENE A CAUSA DI \eqref{eq:indices}.
%

% The energetic formulation can admit solutions with jumps in time; if however, the
% energetic solution is differentiable in time, it is also a solution of the sub
% differential evolution system \eqref{eq:diff-evolution}.
% Conversely, a solution of the sub differential system is not, in general, an energetic
% solution if the energy functional is not convex \cite{Mie08}.\\

We shall now prove the weak semicontinuity of the energy. To establish
it, 
the following Lemma \ref{lem:minors} on the convergence of minors is
needed. We recall that for all $\bA\in \Rzn$, we have a $3\times 3$ matrix of order-one minors: $\bbM_1(\bA)=\bA$, a
$3\times 3$ matrix of order-two minors: $\bbM_2(\bA)=\textrm{cof} \bA$, and a scalar minor of
order three $\bbM_3(\bA)=\textrm{det}\bA$. By introducing the shorthand notation
$\bbM(\bA)=(\bbM_1(\bA),\bbM_2(\bA),\bbM_3(\bA))$, according to assumption \eqref{eq:poly} we can
write  $\We (\nabla y\bCp^{-1/2}) =\bbW(\bbM(\nabla y\bCp^{-1/2}))$. 
The next Lemma is a particular case of \cite[Prop. 5.1]{Mie08}.

\begin{lemma}[Convergence of minors]\label{lem:minors}
Let $y_k\rightharpoonup y$ in $W^{1,q_y}(\Omega;\Rz^3)$ and
$\bFpk\rightarrow {\bFp}$ in $L^p(\Omega;\SL)$ and
\begin{equation}\label{eq:index-lem}
 q_y > 3,\quad  \frac{1}{q_y}+\frac{2}{p}\leq 1.
\end{equation}
Then,
\begin{equation*}
 \mathbb M(\nabla y_k\bFpk^{-1})\rightharpoonup \bbM(\nabla
 y\bFp^{-1})\textrm{ in } L^1(\Omega;\Rzn{\times} \Rzn {\times} \Rz).
\end{equation*}
\end{lemma}
\begin{proof}
Since det$\bFpk=1$, we have
\begin{equation} \label{mino}
\mathbb M(\nabla y_k\bFpk^{-1})=\Big(\nabla y_k(\textrm{cof}\bFpk)^\TT, \textrm{cof}(\nabla
y_k)\bFpk^\TT,\textrm{det}(\nabla y_k)\Big).
\end{equation}
The desired convergence is obtained from the fact that 
$$
\frac{1}{\rho}+\frac{1}{\sigma}\leq 1 , \quad  \bA_k\overset{L^{\rho}}{\rightharpoonup} \bA,\quad \bB_k\overset{L^{\sigma}}{\rightarrow}
\bB \ \Rightarrow \ 
\bA_k \bB_k \overset{L^{1}}{\rightharpoonup} \bA \bB,
$$
applied to the three minors  in \eqref{mino}.  The classical weak continuity of gradient
minors \cite{Ball77}  for $q_y>3$ and 
$y_k\rightharpoonup y \textrm{ in } W^{1,q_y}$ yield
$$
 \nabla y_k\overset{L^{q_y}}\rightharpoonup
\nabla y,\quad \textrm{cof}(\nabla
y_k) \overset{L^{q_y/2}}\rightharpoonup \textrm{cof}(\nabla
y),\quad \textrm{det}(\nabla y_k)\overset{L^{q_y/3}}\rightharpoonup \textrm{det}(\nabla
y).
$$
Moreover, we clearly have that
$$
\bFpk^\TT\overset{L^p }{\longrightarrow} \bFp^\TT,\quad
\cof\bFpk\overset{L^{p/2}}{\longrightarrow} \cof\bFp. 
$$
In order to conclude, the following conditions on indexes need to be checked
$$
\quad q_y>3, \quad \frac{1}{q_y}+\frac{2}{p}\leq 1,\quad
\frac{2}{q_y}+\frac{1}{p}\leq 1.
$$
The first two are \eqref{eq:index-lem} and the last one is a direct consequence
of these.
\end{proof}
We can now establish the weak lower semicontinuity of the energy functional.
\begin{lemma}[Lower semicontinuity of the energy]\label{lem:lsc}
Under the assumptions of Theorem
\emph{\ref{teo:existence1}} the energy $\calE$ is 
weakly lower semicontinuous.
\end{lemma}

\begin{proof} 
Let $(y_k,{\bCpk})\rightharpoonup (y,\bCp)$ in  $\calQ$.
% We recall the expression of $\calE$:
% $$
% \calE(y,\bCp,t)=\int_\O\left[\,\mathbb W(\mathbb M(\nabla y\bCp^{-1/2})) +\haz \Wh(\bCp)+\frac{\mu}{r}|\nabla\bCp |^r\,\right]dx-\l\ell,y\r
% $$
% Let $(y_k,{\bCp}^{(k)})\rightharpoonup (y,\bCp)$ in  $\calY_0\times \calZ_0$, with 
%  $\calY_0=W^{1,q_y}(\O,\R^3)$ and $\calZ_0=L^{q_p}(\O,\R^{3\times 3})\cap
% W^{1,r}(\O,\R^{3\times 3})$.\\
%  The source term is weakly continuous.\\
The compact embedding of
$W^{1,r}\subset\subset L^r$ and the weak convergence of
${\bCpk}\rightharpoonup \bCp$ in $L^{q_p}\cap W^{1,r}$ entail strong convergence in
$L^s$, for all $s\in[1,q_p)$  if $r<q_p$ and all $s\in [1,r]$ if
$r\geq q_p$.  The term 
$$\bCp \mapsto \int_\Omega \left(\haz \Wh(\bCp)+\frac{\mu}{r}|\nabla\bCp|^r\right)\dx$$
is hence lower semicontinuous (see for example  \cite[Thm. 1.6,
p. 9]{Struwe}).  

The convergence in measure of ${\bCpk}$ implies the
convergence in measure of
${\bCpks}$. This follows from the local Lipschitz continuity result  in Lemma
\ref{lem:app-Lip}.
% In fact, the square root of a positive-definite matrix can be
% proved to be a Lipschitz-continuous map on any domain $D\subset \R^{3\times
% 3}_{\rm{sym}+}$ of matrices whose spectrum is lower bounded by a positive constant. On the
% other hand, for any element $\bC\in\MM$,  a bound in the matrix norm provides, by the
% unit determinant constrain, a positive lower bound for the spectrum: $\min \sigma(\bC)\geq
% |\bC|^{-2}$. 
Convergence in measure and the $L^{2q_p}$
boundedness of $\bCpks$, then yield the strong convergence
$\bCpks\rightarrow {\bC}^{1/2}_{\rm p}$ in $ L^p$
for $p\in[1,2q_p)$. 
 Thanks to Lemma \ref {lem:minors}, we have
$$
\mathbb M(\nabla y_k \bCpk^{-1/2})\rightharpoonup \bbM(\nabla
y\bCp^{-1/2})\textrm{ in } L^1.
$$
In fact, condition \eqref{eq:index-lem} is easily verified by checking that
$$
\frac{1}{q_y}+\frac{2}{2q_p}\overset{\eqref{eq:indices}}{=}\frac{1}{q_e}+\frac{3}{2q_p}
\overset{\eqref{eq:indices}}{<}1-\frac{2}{q_e},
$$
which implies ${1}/{q_y}+{2}/{p}\leq 1$ for $p=2q_p-\varepsilon\in[1,2q_p)$ and
$\varepsilon>0$ sufficiently small. As $\We(\nabla y \bCp^{-1/2})=\bbW(\mathbb M(\nabla
y\bCp^{-1/2}))$, with $\bbW$ convex, the lower semicontinuity of the
elastic energy term follows. Eventually, the time-dependent linear
term is weakly continuous.
\end{proof}

We are now ready to prove the existence of energetic solution for the
quasistatic evolution problem.

\begin{proof}[Proof of Theorem \emph{\ref{teo:existence1}}]
%-------------------------------------------------
%
We aim at checking the assumptions of the abstract
existence Lemma \ref{lem:abstract-ex}. The structural conditions
\eqref{Diss_ass} follow directly from the properties of the density $D$.
% The following Lemma establishes hypothesis (D) of theorem \ref{teo:abstract-ex}.
% %
% \begin{lemma}
% The  map $\calD:\calZ\times\calZ\rightarrow [0,\infty]$ defined by
% \eqref{eq:dist} 
% is a quasi-distance and is  weakly continuous. Moreover it satisfies the bound
% \begin{equation}\label{eq:D-bound}
% \calD(\bCp,\bCp')\leq c_0(c_1+\|\bCp\|_{L^1}+\|\bCp'\|_{L^1});
% \end{equation}
% in particular $\calD$ is finite: $\calD:\calZ\times\calZ\rightarrow [0,\infty[$.
% \end{lemma}
% \begin{proof}
% The quasi-distance properties of the integrand $D$ imply that $\calD$ is a quasi distance 
% on $\calZ$: the triangle inequality is trivial and $\calD=0$ implies $D(x)=0$ a.e\\
% Bound \eqref{eq:D-bound} is consequence of \eqref{eq:D-pointbound}.\\
% From continuity of $D$ (Lemma \ref{lem:D}) and the proven bound, the $L^1$-norm
% continuity 
% of $\calD$ is obtained through Lebesgue dominated convergence theorem. But
% ${\bCp}^{(k)}\rightharpoonup \bCp$ on $W^{1,r}$ implies $L^1$ convergence.
% \end{proof}
% %
Thanks to Lemmas \ref{lem:coercive} and \ref{lem:lsc}, the sublevels of
energy are bounded and weakly closed in $\calQ$, whence the
compactness  follows. The  power reads
$$
\partial_t \calE(q,t)=-\langle \dot\ell,y\rangle,\quad
\dot\ell \in C^0([0,T];(W^{1,q_y}_{\Gamma_{\rm D}}(\O,\R^3))^*).
$$
 In particular,
assumption \eqref{En_ass} follows.
Eventually, the closure of stable states \eqref{closure_ass} is a
consequence of the lower semicontinuity of the energy $\calE$ and the
continuity of the dissipation $\calD$. 
\end{proof}

Before moving on, let us remark that the weaker assumption
 \begin{equation*}
  \ell \in W^{1,1}([0,T];(W^{1,q_y}_{\Gamma_{\rm D}}(\O,\R^3))^*)
 \end{equation*}
on external actions would also suffice to prove an existence result. In this case, the
abstract result of Lemma \ref{lem:abstract-ex} has to be adapted 
and indeed simplified  to
the case of
linear applied loads,  see for instance \cite{AuMieSte08}.

%
%==================================================
\section{Small-deformation limit for quasistatic
evolution}\label{sec:linear}
%==================================================
%
Let us now turn to the proof of a linearization result for quasistatic
evolutions. The aim here is to present the three-dimensional
version of the main  result of Section \ref{sec:lin_const}, namely Theorem
\ref{sdl}. Again, this consists in a variational convergence 
argument. 

In
the stationary, multidimensional framework, the seminal contribution in this respect is
\cite{DNP02} where a variational justification of
linearization in elasticity is provided. Successive refinements
\cite{Agostiniani11}  and extensions
\cite{Paroni-Tomassetti09,Paroni-Tomassetti11,Schmidt08} of the
argument have been presented.  
In the evolutive case, the first result in plasticity with hardening is in
\cite{MieSte13}. Linearized plate models have been
derived from finite plasticity  in \cite{Davoli14,Davoli14b} and perfect plasticity is
considered in \cite{Giacomini13}. We follow here the general strategy
of \cite{MieSte13}, by adapting indeed many technical points to the
present symmetric situation. In particular, the choice of rescaled
variables and functionals is here different  from  \cite{MieSte13}
 and especially tailored to
cope with the nonlinear structure of $\MM$.  Note additionally
that \cite{MieSte13} proves the convergence of finite-plasticity
trajectories, whose existence under the assumptions of that paper is
not known.  The situation is here different
as energetic solutions do exist under the  assumptions of 
Theorem \ref{teo:existence1}. 
The price to pay
here is that we have to discuss the limiting behavior of the gradient
term in $\nabla \bCp$. In particular, our small-deformation analysis
is restricted to the case $r=2$ in \eqref{eq:int-energy} for the
quadratic character of the gradient energy term plays a crucial role.

Letting $\e>0$ and $(y,\bCp) \in \calQ$ be given we introduce the
equivalent variables
\begin{equation}\label{forma2}
u =  \frac{1}{\epsi}(y {-} {\rm id}), \quad \bz = \frac{1}{2\e}\log
 \bCp
\end{equation}
as well as the rescaled functionals for $r=2$
\begin{align*}
\calE_\e(u, \bz,t) &= \int_\Omega W_\e(\bE_\e,\bz) \,\d x +
\frac{\mu}{2\e^2}\int_\Omega |\nabla \exp(2\e \bz)|^2 - \langle\ell(t), u
\rangle,\quad \\
\calD_\e (\bz_1,\bz_2) &= \int_\Omega D_\e (\bz_1,\bz_2) \,\d x
\end{align*}
where the rescaled Green-Saint Venant strain $\be_\e$ is given by 
$$\bE_\e = \frac{1}{2\e}( (\one{+}\e \nabla u)^\TT
(\one{+}\e \nabla u) {-} \one) = \nabla u^{\rm sym} + \frac{\epsi}{2}
\nabla u^\TT \nabla u.$$

Assume now to be given $\ell \in C^1([0,T]; (W^{1,q_y}_{\Gamma_{\rm D}}
(\Omega; \Rz^3)))^*)$ and initial values $\bz_{0\e}$ such that
$\exp(2\epsi \bz_{0\e})\in \calS(0)$ where $\calS(t)$ denotes stable
states at time $t \in [0,T]$ corresponding to $(\calQ,\calE/ \e^2 ,\calD/(2\e))$. Owing to Theorem
\ref{teo:existence1} there exists an energetic solution
$(y_\e,\bC_{\rm p\e})
$ corresponding to $(\calQ,\calE/ \e^2 ,\calD/(2\e))$. Correspondingly, by
defining $(u_\e,\bz_\e)$ from $(y_\e,\bC_{\rm p\e})$ via
\eqref{forma2} we readily find that $(u_\e,\bz_\e)$ is an energetic
solution corresponding to $(\calQ_0,\calE_\e,\calD_\e)$ where the
space $\calQ_0$ can be chosen as
$$ \calQ_0 =  H^1_{\Gamma_{\rm D}} (\Omega;\Rz^3) \times H^1(\Omega;\Rzd)$$
by simply extending trivially the functionals. We shall refer to
$(u_\e,\bz_\e)$ as {\it finite-plasticity quasistatic evolutions} and
denote the corresponding set of stable states at time $t \in [0,T]$ by $\calS_\e(t)$.

We  are here  concerned with the  convergence  of finite-plasticity quasistatic evolutions
$(u_\e,\bz_\e)$ to the unique solution $(u,\bz)$ of
the linearized elastoplasticity system corresponding to
$(\calQ_0,\calE_0,\calD_0)$ where 
  \begin{align*}
\calE_0(u, \bz,t) &= \int_\Omega W_0(\nabla u^{\rm sym}, \bz) \,\d x +
2 \mu \int_\Omega |\nabla \bz |^2 - \langle\ell(t), u
\rangle,\quad \\
\calD_0 (\bz,\haz \bz) &= \int_\Omega D_0(\bz,\haz \bz) \, \d x =
\int_\Omega r|\bz_1 {-}\bz_2| \,\d x.
\end{align*}
In particular, $(u,\bz)$ is the unique strong solution of the
equilibrium system \eqref{eq:eq} with ${\boldsymbol \sigma} = \mathbb{C} (\nabla u^{\rm sym}{-}\bz)$ and boundary conditions
\eqref{eq:BC1}-\eqref{eq:BC2}, and of the constitutive
relation 
$$R\partial |\dot \bz| + (\mathbb{H}{+}\mathbb{C})\bz - 4\mu \Delta
\bz\ni \mathbb{C} \nabla u^{\rm sym}$$
along with the homogeneous Neumann condition $\mu \nabla
\bz   \nu  = {\boldsymbol 0}$ and an initial condition for
$\bz $ \cite{Han-Reddy}. We
refer to $(u,\bz)$ as {\it linearized-plasticity quasistatic evolution}.

We now state our main convergence result.

\begin{theorem}[Small-deformation limit of the quasistatic
  evolution]\label{sdl2}  Let the compact set  
\begin{equation}\label{eq:K}
 K:=\{\bC\in\MM:\;|\bC|\leq r\},\quad r^2>3 
\end{equation}
be given and assume $\haz \Wh$ to be coercive in the following sense
\begin{equation}
  \label{eq:32}
  \haz \Wh(\bCp) < \infty  \ \Leftrightarrow \  \bCp, \, \bCp^{-1} \in K.
\end{equation}
Moreover, assume the control \eqref{Kirchhoff}, the quadratic behavior at identity \eqref{quad_behavior}, and let
\begin{equation}\label{dist}
 \quad\We(\bF)\geq c_6\, \textrm{\rm dist}^2(\bF,\SO)\quad  \forall \bF\in\GLp
\end{equation}
for some positive constant $c_6$. Let $(u_\e, \bz_\e)$
  be finite-plasticity quasistatic evolutions starting from well-prepared initial data $(u_{0\e},\bz_{0\e}) \in \calS_\e(0)$, namely,
  \begin{equation*}\label{well-prepared2}
  (u_{0\e},\bz_{0\e})\to (u_0,\bz_0) \quad \text{and} \quad \calE_\e
  (u_{0\e},\bz_{0\e},0) \to \calE_0
  (u_{0},\bz_{0},0).
\end{equation*}
Then, for all $t\in [0,T]$,
\begin{align*}
  &(u_\e(t),\bz_\epsi(t)) \to (u(t),\bz(t)),\\ &{\rm
    Diss}_{\calD_\e,[0,t]}(\bz_\e) \to {\rm Diss}_{\calD_0,[0,t]}(\bz), \\
  &\calE_\e(u_\e(t),\bz_{\epsi}(t),t) \to \calE_0(u(t),\bz(t),t)
\end{align*}
where $(u,z)$ is the unique linearized-plasticity quasistatic
evolution starting from $(u_0,\bz_0)$.
\end{theorem}

 The coercivity assumption on $\haz \Wh$ is stronger
than the former \eqref{eq:3}, \eqref{coercivity_of_Wh2}, and
\eqref{coercivity_joint}. The assumption on the shape of $K$ from \eqref{eq:K} is
of technical nature and could probably be relaxed.  Let us however stress that it arises quite naturally in
modeling pseudoplastic processes in shape memory alloys
\cite{Eva09,Eva10,reese07}.  In addition,  note that the choice for the shape of $K$ is immaterial with respect to the
linearization limit $\epsi \to 0$.

Let us start by presenting a coercivity result.

\begin{lemma}[Coercivity of the energy]\label{lem:coerc}
 Under the assumptions of Theorem \emph{\ref{sdl2}} we have that 
 \begin{equation}
  \|\nabla
u\|_{L^2}^2+\|\bz\|_{L^2}^2+\|\nabla\exp(2\e\bz)/\e\|_{L^2}^2+\e\|\bz\|_{L^\infty}\leq
c_7(1{+}\calW_\e(u,\bz)) \label{eq:coerc0}
 \end{equation}
 for all $(u,\bz) \in \calQ$   where $c_7$ is a positive constant.
\end{lemma}
\begin{proof}
% Assume $\calW_\e(u,z)<\infty$. Then, $\exp(2\e \bz) \in K$ and we have
% that $\e \| \bz\|_{L^\infty} \leq c$.
% From
% coercivity \eqref{eq:3} we have that $\| \bz\|_{L^2}+ \| \nabla
% \exp(2\e\bz)/\e\|_{L^2}\leq
% c\calW_\e(u,z)$. Note that, for all $\alpha\in \Rz$
% \begin{equation}\exp(\alpha\e \bz_\e) = \one + \alpha\e \bz_\e + \alpha^2\e^2
% z^2\left(\sum_{k=2}^\infty \frac{(\alpha \e \bz)^{k-2}}{k!}
% \right)\label{eq:expa}
% \end{equation}
% As $\| \bz_\e\|_{L^2} + \|\e \bz_\e\|_{L^\infty}<c$ the sum in the
% right-hand side is bounded
% in $L^1$. On the other hand, the last term in the right-hand side of
% \eqref{eq:expa} is bounded in $L^2$. We have hence proved that 
% \begin{equation}\exp(\alpha\e \bz_\e) = \one + \alpha\e \bz_\e + \e \bL_\e
% \label{eq:expa2}
% \end{equation} 
% where 
% \begin{equation}\label{eq:expa3}\e \| \bL_\e\|_{L^\infty} + \| \bL_\e\|_{L^2} +
% \frac{1}{\e}\|\bL_\e\|_{L^1}\leq c.
% \end{equation}
Assume $\calW_\e(u,z)<\infty$. Then, $\exp(2\e \bz) \in K$  so 
that $\e \| \bz\|_{L^\infty} \leq c$. 
 In fact, by denoting by $\la_i$ the eigenvalues of $\bz$, 
condition $\tr \bz=0$ entails $\la_{\rm max} \geq \|
\bz\|_{\infty}/6$, hence $\| \exp(2\e \bz)\|_\infty \geq {\rm e}^{2\e
  \la_{\rm max}} \geq {\rm e}^{\e \| \bz\|_{L^\infty} /3}$.  
From  the  coercivity of $\haz \Wh $   we have that $\| \bz\|_{L^2}+ \| \nabla \exp(2\e\bz)/\e\|_{L^2}\leq
c\calW_\e(u,z)$. Note that, for all $\alpha\in \Rz$
\begin{equation}\exp(\alpha\e \bz) = \one + \alpha\e \bz + \e^2
\bz^2\bB_\e,\quad\bB_\e:=\left(\sum_{k=2}^\infty \frac{\alpha^k( \e \bz)^{k-2}}{k!}
\right).\label{eq:expa}
\end{equation}
As $  \|\e \bz\|_{L^\infty}\leq c$,  we have that 
$\|\bB_\e\|_{L^\infty}\leq c$.
% 
%  \\ Ho preferito separare chiaramente i termini di ordine $\e$ e $\e^2$ piuttosto che
% tenere $\one+\alpha\e\bz+\e\bL_\e$ con
% \begin{equation}\label{eq:expa3}\e \| \bL_\e\|_{L^\infty} + \| \bL_\e\|_{L^2} +
%  \frac{1}{\e}\|\bL_\e\|_{L^1}\leq c.
%  \end{equation}
%  La separazione dei contributi $\e$ e $\e^2$ e' importante a partire da \eqref{eq:CCtilde}
% in avanti. La $\bL_\e$ che ivi compare questa volta non soddisfa  la limitatezza $
% \e^{-1}\|\bL_\e\|_{L^1}\leq c$ ma solo \eqref{eq:expa4}. Percio' ho pensato che non fosse
% il caso di metterla in evidenza in modo particolare   
% 
%As $\| \bz_\e\|_{L^2} + \|\e 
%\bz_\e\|_{L^\infty}<c$ the sum in the
%right-hand side is bounded
%in $L^1$. On the other hand, the last term in the right-hand side of
%\eqref{eq:expa} is bounded in $L^2$. We have hence proved that 
%\begin{equation}\exp(\alpha\e \bz_\e) = \one + \alpha\e \bz_\e + \e \bL_\e
%\label{eq:expa2}
%\end{equation} 
%where 
%\begin{equation}\label{eq:expa3}\e \| \bL_\e\|_{L^\infty} + \| \bL_\e\|_{L^2} +
%\frac{1}{\e}\|\bL_\e\|_{L^1}\leq c.
%\end{equation}

The coercivity estimate on $\nabla u$ is based on a geometric rigidity
argument \cite{FJM02}  as in \cite{DNP02}. The first step is
to obtain an estimate of the distance of $\nabla y$ from $\SO$. By using
$
\bFe:=\nabla y\bCp^{-1/2}
$  and, recalling that 
$\bCp\in K$  is  bounded, one obtains
$$
|\nabla y{-}\bQ|^2=|(\bFe{-}\bQ)\bCp^{1/2}+\bQ(\bCp^{1/2}{-}\one)|^2\leq
c|\bFe{-}\bQ|^2+|\bCp^{1/2}{-}\one|^2
$$
almost everywhere.
We now use \eqref{eq:expa} with $\alpha =1$ on order to get that 
$\bCp^{1/2} = \one + \epsi \bz + \epsi^2\bz^2\bB_\e=\one + \epsi \bz \bB_\e'$, 
where $\bB_\e'=\one+(\e\bz)\bB_\e $.
Therefore
$$
|\bCp^{1/2}{-}\one|^2\leq \e^2|\bz|^2|\bB_\e'|^2
$$
and we conclude for
\begin{align*}%\label{eq:coerc1}
|\nabla y{-}\bQ|^2=|(\bFe{-}\bQ)\bCp^{1/2}+\bQ(\bCp^{1/2}{-}\one)|^2\leq
c(|\bFe{-}\bQ|^2+\e^2|\bz|^2|\bB_\e'|^2)
\end{align*}
where $ \|\bB_\e'\|_{L^\infty}\leq c$.
We now proceed as in \cite[Lemma 3.1]{MieSte13}. The last inequality
combined with the nondegeneracy condition \eqref{dist} 
yields
\begin{equation*}
 \int_\O \textrm{dist}^2(\nabla y, \SO) \dx\leq c\e^2(1{+}\calW_\e(u,\bz)).
\end{equation*}
Then, the Rigidity Lemma \cite[Thm. 3.1]{FJM02} entails
\begin{equation*}
\exists\haz\bQ\in\textrm{SO}(3):\quad \|\nabla
y{-}\haz\bQ\|_{L^2}^2\leq c\e^2(1{+}\calW_\e(u,\bz)).
\end{equation*}
As the rotation $\haz\bQ$ satisfies the estimate $|\haz\bQ{-}\one|^2\leq
c\e^2(1{+}\calW_\e(u,\bz))$ as a result of the boundary conditions,
 see 
\cite[Prop. 3.4]{DNP02},  we have
\begin{equation*}
 \e^2\|\nabla u\|_{L^2}^2=\|\nabla y{-}\one\|_{L^2}^2\leq 2\|\nabla
y{-}\haz\bQ\|_{L^2}^2+2\|\haz\bQ{-}\one\|_{L^2}^2\leq c\e^2(1{+}\calW_\e(u,\bz))
\end{equation*}
and the assertion follows.
\end{proof}

The next step toward to application of the evolutive
$\Gamma$-convergence Lemma \ref{lem:evol} is the proof of the
$\Gamma$-$\liminf$ inequalities \eqref{gammainfD}-\eqref{gammainfE}. 
We do this in the  next  two lemmas.

\begin{lemma}[$\Gamma$-$\liminf$ inequality for
  $\calE_\e$]\label{lem:G-limit-En}
 Under the assumptions of Theorem \emph{\ref{sdl2}} we have that 
   $$ \calE_0(u,\bz,t) \leq \inf \Big\{ \liminf_{\e \to 0}
  \calE_\e(u_\e,\bz_\e,t_\e) \;|\; (u_\e,\bz_\e,t_\e) \stackrel{\calQ
    \times [0,T]}{\weakto} (u,\bz,t) \Big\}.$$
\end{lemma}

\begin{proof}
As convergence for the linear external energy $\langle \ell(t_\e),
u_\e\rangle $ is trivial, we concentrate on the terms 
\begin{align*} I^1_\e&:=\frac{1}{\e^2} \int_\Omega \haz \We \big(\exp(-\e \bz_\e) (\one{+}2\epsi \bE_\e)
\exp(-\e \bz_\e)\big) \,\d x, \\
&= \frac{1}{\e^2} \int_\Omega  \We \big(
(\one{+}\epsi \nabla u_\e)
\exp(-\e \bz_\e)\big) \,\d x \\
I^2_\e&:= \frac{1}{\e^2} \int_\Omega \haz \Wh \big(\exp(2\e
\bz_\e)\big) \, \d x,\\
I^3_\e&:=  \frac{\mu}{2\e^2} \int_\Omega |\nabla \exp(2\e \bz_\e)|^2\,
\d x =  2\mu \int_\Omega |\nabla \exp(2\e \bz_\e)/(2\e)|^2\,
\d x .
\end{align*} 
Let
$(u_\e,\bz_\e)\overset{\calQ}{\rightharpoonup}(u,\bz)$. We can assume with no loss of generality that $\sup
\calE_\e(u_\e,\bz_\e,t_\e) <\infty$, so that the bound
\eqref{eq:coerc0} holds and $\exp(2\e \bz_\e) \in K $. In particular,
$\| \e \bz_\e \|_{L^\infty}\leq c$. 

Relation \eqref{eq:expa} for $\alpha=2$ entails that $\exp(2\e \bz_\e)
= \one + 2\e \bz_\e + \e^2\bze^2 \bB_\e$ with
$\|\bB_\e\|_{L^\infty}\leq c$.  As $\|\bz_\e\|_{L^2}\leq c$, we
compute 
$$\frac{1}{2\e} \Big( \exp(2\e \bz_\e)   {-} \one \Big)   - \bz_\e=
\frac12 \e\bz_\e^2 \bB_\e \stackrel{L^1}{\rightarrow} 0 $$
and, taking into account the $L^2$-boundedness of the same sequence, we have checked that $(\exp(2\e \bz_\e)   {-} \one)/(2\e) \weak \bz$ in
$L^2$. On the other hand, by \eqref{eq:coerc0} we also have the gradient bound $\|\nabla
(\exp(2\e \bz_\e)   {-} \one)/(2\e)\|_{L^2} \leq c$. Hence,
convergence 
\begin{align}
&\frac{1}{2\e}\Big( \exp(2\e \bz_\e)   - \one \Big) \stackrel{H^1}{\weak}  \bz\label{tocheck2}
\end{align}
 follows. In particular,
$$ 2\mu \int_\Omega |\nabla \bz|^2 \dx \leq \liminf_{\e \to 0} 2\mu
\int_\Omega |\nabla \exp(2\e \bz_\e) /(2\e)|^2 \dx  =  \liminf_{\e \to 0}
I^3_\e.$$
 Convergence 
\eqref{Gamma_manca_Wh} and Lemma \ref{lem:semicon-tool} directly entail that 
$$ \frac12  \int_\Omega | \bz|^2_{\mathbb H} \dx \leq \liminf_{\e \to 0}
I^2_\e.$$

Let us now turn to  integral $I_\e^1$.  We introduce, also for future reference, the
 shorthand notation 
\begin{equation}\label{eq:Aepsilon}
 \bA_\e=\frac{1}{\e}\Big((\one{+}\epsi \nabla u_\e)
\exp(-\e \bz_\e) - \one \Big),
\end{equation}
which allows to  rewrite $I_\e^1$ as   
$$I_\e^1=\frac{1}{\e^2}\int_\O \We(1{+}\e\bA_\e)\,\d x.$$
In view of  convergence  \eqref{Gamma_manca_We} and Lemma \ref{lem:semicon-tool}, it is sufficient to
prove 
\begin{equation}\label{eq:Aepsilon-lim}
 \bA_\e\stackrel{L^2}{\weak} \nabla u -\bz.
\end{equation}
By expanding $\exp(-\e \bz_\e) = \one{-}\e\bz_\e+\e^2\bze^2 \bB_\e$ according to 
\eqref{eq:expa} with $\alpha =-1$ we compute
\begin{align*}
 \bA_\e=(\nabla u_\e{-} \bz) +
(\bz{-}\bz_\e)+\e\bze^2\bB_\e+\e\nabla u_\e\bze(\e\bze\bB_\e{-}\one).
\end{align*}
As $\nabla u_\e \weak \nabla u$ in $L^2$, $\bz_\e \weak \bz$ in
$L^2$, and $\|\bB_\e\|_{L^\infty}+\|\e\bz_\e\|_{L^\infty}\leq c$, 
 the convergence \eqref{eq:Aepsilon-lim} follows.
% Let us now compute
% \begin{align*}
%   &\frac{1}{\e}\Big((\one{+}\epsi \nabla u_\e)
% \exp(-\e \bz_\e) - \one \Big) - (\nabla u_\e{-}
% \bz)  \\
% &=  \frac{1}{\e}  (\one{-}\e\bz_\e+\e \haz \bL_\e) - (\nabla u_\e {-}
% \bz)  \\
% & = \haz \bL_\e - \e \nabla u_\e \bz_\e + \e \nabla u_\e \haz \bL_\e
% \end{align*}
% where $\exp(-\e \bz_\e) = \one{-}\e\bz_e+\e \haz \bL_\e$ follows again
% by \eqref{eq:expa2} for $\alpha =-1$. As $\nabla u_\e \weak \nabla u$ in $L^2$, $\bz_\e \weak \bz$ in
% $L^2$, and the bounds
% \eqref{eq:expa3} hold, we deduce
% \begin{align}
% \frac{1}{\e}\Big((\one{+}\epsi \nabla u_\e)
% \exp(-\e \bz_\e) - \one \Big) \stackrel{L^2}{\weak} \nabla u -
% \bz.\label{tocheck}
% \end{align}
% In particular, the $\Gamma$-$\liminf$ inequality 
% $$ \frac12 \int_\Omega |\nabla u {-} \bz|^2_{\mathbb C} \dx \leq
% \liminf_{\e \to 0} I^1_\e$$
% follows from \eqref{Gamma_manca_We} by applying Lemma \ref{lem:semicon-tool}.
\end{proof}

\begin{lemma}[$\Gamma$-$\liminf$ inequality for
  $\calD_\e$]\label{lem:G-limit-Dn}
 Under the assumptions of Theorem \emph{\ref{sdl2}} we have that 
  $$ \calD_0(\bz,\haz \bz) \leq \inf \Big\{ \liminf_{\e \to 0}
  \calD_\e(\bz_\e,\haz \bz_\e) \;|\; (\bz_\e,\haz \bz_\e)
   \weakto (\bz,\haz \bz)  \ \ \text{in} \ L^2(\Omega;(\Rzd)^2)
    \Big\}.$$
\end{lemma}

\begin{proof}
  The assertion follows by Lemma \ref{lem:G-limit-D0} by 
applying the lower semicontinuity
  tool of Lemma \ref{lem:semicon-tool}.
\end{proof}

 Having established the above $\Gamma$-$\liminf$ 
inequalities,  the
next  ingredient of the  evolutive-$\G$-convergence argument is
the  
specification of a   \emph{mutual recovery sequence}.  This is
done within the following lemma. 
\begin{lemma}[Mutual recovery sequence]\label{lem:MRS}  Under the
  assumptions of Theorem \emph{\ref{sdl2}} \linebreak let
  $(u_\e,\bz_\e)\overset{\calQ}{\rightharpoonup} (u_0,\bz_0)$ be given with
$\sup_\e\calE_\e(t,u_\e,\bz_\e)<\infty$. Moreover, let  
$$
(\widehat u_0,\widehat bz_0)=(u_0,\bz_0)+(\widetilde u,\widetilde
\bz)$$
where $
(\widetilde u,\widetilde \bz) \in C^\infty_{\rm c} (\Omega; \Rz^3)
\times C^\infty_{\rm c}(\Omega; \Rzd)$. 
Then, there exists $(\widehat u_\e,\widehat \bz_\e)\in\calQ$ such that
$(\widehat u_\e,\widehat\bz_\e)\overset{\calQ}{\rightharpoonup}
(\widehat u_0,\widehat \bz_0)$ and
\begin{align}\label{eq:limsupD}
 &\limsup_{\e\rightarrow 0}\calD_\e(\bz_\e,\widehat\bz_\e)\leq
\calD_0(\bz_0,\widehat\bz_0)=R(\widetilde\bz),\\\label{eq:limsupE}
&\limsup_{\e\rightarrow 0}\Big(\calE_\e(\widehat u_\e,\widehat\bz_\e,t)-\calE_\e(
u_\e,\bz_\e,t)\Big)\leq \calE_0(\widehat u_0,\widehat\bz_0,t)-\calE_0(
u_0,\bz_0,t).
\end{align}
\end{lemma}
\begin{proof}
 We divide the proof into  several steps. 

\noindent\emph{ Step 1: Definition of the  recovery sequence.}  We
define 
\begin{align}
  \disp \widehat u_\e&:=u_\e+\widetilde u\circ(\textrm{id}{+}\e u_\e) 
  \displaystyle \label{rec_u}\\
\widehat\bz_\e&:=\frac{1}{2\e}\log\Big( \bPi \,(\exp(2\e(
\bz_\e{+}\widetilde \bz)))\Big)\label{rec_z}
\end{align}
where $  \bPi :\MM\rightarrow K$ is a contraction mapping onto the compact set
$K$ defined by \eqref{eq:K}. More precisely,  we  ask $\bPi$ to
have  the
following properties of  $ \bPi $:
\begin{equation}\label{eq:Pi}
 \bPi \left|_K\right.=\textrm{id}_K\quad \textrm{ and }\quad 
| \bPi (\bC_{\rm p1})- \bPi (\bC_{\rm
p2})|\leq|\bC_{\rm p 1}{-}\bC_{\rm p2}|\quad \forall\bC_{\rm p1},\bC_{\rm p2}\in\MM.
\end{equation}
 An explicit construction of a map $ \bPi $ fulfilling these
properties is provided in 
Appendix~B. 

The definition
of $\widehat u_\e$  can be rewritten  in terms of $\widehat
y_\e=\textrm{id}+\e \widehat u_\e$  as 
\begin{equation}\label{eq:recovery-y}
\widehat y_\e=\textrm{id}+\e u_\e +\e\widetilde u\circ (\textrm{id}{+}\e u_\e)=\widetilde
y
\circ y_\e,
\end{equation}
where  now  $\widetilde y=\textrm{id}+\e\widetilde u$. 
This choice for  the recovery sequence $\widehat u_\e $
 corresponds  to the one used  in
\cite{MieSte13}.  Note in particular that 
$$\det \nabla \widehat y_\e =\det (\nabla  \widetilde y (y_\e)
\nabla y_\e) = \det \nabla \widetilde y (y_\e) \, \det \nabla y_\e
>0$$
for small $\e$, as $\det \nabla \widetilde y (y_\e)\to 1$ uniformly.
That is
$\one +\e
\nabla \haz u_\e \in \GLp$ almost everywhere in $\Omega$ for all $\e$
sufficiently small. Moreover, we immediately check that 
\begin{equation}
  \label{eq:u-conv0}
  \haz u_\e \stackrel{L^2}{\to} \haz u_0.
\end{equation}

 The  recovery sequence $\haz\bze$ is
different from the one used in \cite{MieSte13} in two respects. 
At first, the choice is tailored to have a recovery sequence made of
{\it symmetric} tensors whereas no symmetry  of the recovery
sequence  is of course  imposed  in \cite{MieSte13}. Secondly, we address
here the additional intricacy of keeping the gradient of the
recovery sequence bounded in $L^2$ while gradient terms were not
discussed in  \cite{MieSte13}. Note that by neglecting $\log$, $\bPi$, and
$\exp$ in the definition of $\haz \bz_\e$ we would retrieve the classical
choice $\haz \bz_\e = \bz_\e+\tilde \bz$ which is well-suited for
quadratic energies in the linear-space setting \cite{Mie08}. The
actual definition of $\haz \bz_\e$ is hence an adaptation of the
latter to the nonlinear structure of $\MM$. 

 In the following we will use the  shorthand notations 
$$
\bCpe:=\exp(2\e\bz_\e),\quad
\bCpetild:=\exp(2\e(\bz_\e{+}\widetilde\bz)),\quad
\bCpehat:= \bPi (\widetilde\bC_\e)=\exp(2\e\widehat \bz_\e).
$$

\noindent\emph{ Step 2: Preliminary results.} 
By the coercivity Lemma \ref{lem:coerc}, we have the bound
\begin{equation}
\|\bze\|_{L^2}^2+\frac{1}{\e}\| \nabla \exp(2\e\bze)\|_{L^2}^2+\e\|\bze\|_{L^\infty}\leq c. 
\end{equation}
 We now use the uniform Lipschitz  continuity  of the
logarithm on  the set $K$ 
$$
2\e|\nabla\bz_\e|=|\nabla \log \bCpe|\leq c|\nabla\bCpe|.
$$
which is proved in Lemma 
\ref{lem:app-Lip} in Appendix A. In particular, we deduce that 
\begin{equation}\label{eq:gradz}
\|\nabla \bze\|_{L^2}^2\leq c. 
\end{equation} 
% In the following, we will denote collectively by $\bL_\e$ any  sequence of matrix
% functions which satisfy the uniform bound, that is $\bL_\e$ behaves in
% this respect like $\bze$. 
% According to this definition, $(\e\bL_\e)^n=\e\bL_\e $ for any positive integer $n$.\\
For any $\alpha\in\Rz$, by expanding the exponential
$\bCpetild^\alpha=\exp(2\alpha\e(\bz_\e{+}\widetilde\bz))$, we obtain the useful
 expression 
\begin{equation}\label{eq:CCtilde}
 \bCpetild^\alpha=2\alpha\e\widetilde\bz+\bCpe^\alpha+\e^2\bL_\e
\end{equation}
where
$$
\bL_\e=\sum_{n=2}^{\infty}\frac{(2\alpha)^n\e^{n-2}}{n!}\big((\bz_\e{+}\widetilde\bz)^n-\bz_\e^n
\big)
$$
satisfies the bound
\begin{equation}
\|\e\bL_\e\|_\infty+\|\bL_\e\|_{L^2}\leq c.
\label{eq:expa4}
\end{equation}
In fact, for  all  $n\geq 2$, the  highest  power of
$\bz_\e$  in $\bL_\e$ is  controlled by   $
c\e^{n-2}|\bz_\e|^{n-1}$  and
$\|\e\bze\|_\infty+\|\bze\|_{L^2}\leq c$. 
Moreover, for $\alpha=1$, we have 
\begin{equation}\label{eq:nablaC}
 \nabla\bCpetild=2\e\nabla\widetilde\bz+\nabla\bCpe+\e^2\nabla\bL_\e
\end{equation}
 and we can prove that $\|\nabla\bL_\e\|_{L^2}\leq c$.  In fact,
 the $\bze$-dependent terms in the expansion of $\bL_\e$ 
behave as  $\e^{n-2}\bze^{k} $, with $1\leq k\leq n-1$ for any
$n\geq 2$.  Correspondingly,  the gradient terms fulfill
$$
\|\e^{n-2}\nabla \bze^k\|_{L^2} \leq   k  \e^{n-2}\|(\nabla
\bze)\bze^{k-1}\|_{L^2}= k  \e^{n-k-1}\|(\nabla
\bze)(\e\bze)^{k-1}\|_{L^2}
$$
and the bound $\|\nabla\bL_\e\|_{L^2}\leq c$ follows  from  $
\|\e\bze\|_\infty+\|\nabla\bze\|_{L^2}\leq c$. 

We  next  define the sets
$$
K_\e:=\{x\in\O\;|\;\bCpetild\in K\}=\{x\in\O\;|\;\bCpehat=\bCpetild\}.
$$
In particular,  note that 
$$
\bzehat-\widetilde\bz-\bze =0 \ \ \text{on}\  K_\e.
$$
The complement of $K_\e$ has small measure. Indeed, from
\eqref{eq:expa} and \eqref{eq:CCtilde}, it follows that
$\|\bCpetild{-}\one\|_{L^2}^2\leq c\e^2$.  Moreover, one has that
 $|\bCpetild(x){-}\one|\geq  r/\sqrt{3}-1$ for
$\bCpetild(x)\in\MM {\setminus}
{K}$, that is  for  $x\in \Omega{\setminus} K_\e$.  Hence,
\begin{align*}
 |\O {\setminus} K_\e|=\int_{\Omega\setminus K_\e}\dx\leq
  \frac{1}{(r/\sqrt{3}{-}1)^2}  \frac{}{}\int_{\Omega\setminus K_\e}|\bCpetild{-}\one|^2\dx\leq
 \frac{1}{(r/\sqrt{3}{-}1)^2}  \|\bCpetild{-}\one\|_{L^2}^2\leq c\e^2.
\end{align*}

The following convergences will be used in the estimate of the $\limsup$ of 
the hardening  terms 
\begin{align}
 &\widehat\bz_\e-\bz_\e\overset{L^2}{\longrightarrow} \widetilde\bz,\label{eq:bz-conv1}\\
 &\widehat\bz_\e+\bz_\e\overset{L^2}{\rightharpoonup}
\widehat\bz_0+\bz_0.\label{eq:bz-conv2}
\end{align}
 Indeed, on  $K_\e$ we have $ \widehat\bz_\e{-}\bz_\e= \widetilde\bz$ and
$\widehat\bz_\e+\bz_\e=\widetilde\bz+2\bz_\e$ with
$\bz_\e\overset{L^2}{\rightharpoonup}\bz_0 $.  Hence, convergences
\eqref{eq:bz-conv1}-\eqref{eq:bz-conv2} follow from  
$|\O{\setminus} K_\e|<C\e^2$ and the $L^2$-boundedness of 
$\widehat\bz_\e$ and $\bz_\e$.  In particular, by taking the sum of
\eqref{eq:bz-conv1} and \eqref{eq:bz-conv2} we conclude that 
\begin{equation}\label{eq:bz-conv3}
\haz
\bz_\e \stackrel{L^2}{\weak} \haz \bz_0.
\end{equation}

\noindent\emph{ Step 3: The $\limsup$ inequality for the dissipation.} 
Let us  decompose 
\begin{align*}
 \calD_\e(\bz_\e,\widehat\bz_\e)
&=\frac{1}{2\e}\int_{\O\setminus
K_\e} 
D(\bCpe,\bCpehat)\,\d x+\frac{1}{2\e}\int_{K_\e}D\big(\exp(2\e\bz_\e),\exp(2\e
\widehat\bz_\e)\big)\,\d x.
\end{align*}
Taking into account 
the uniform Lipschitz continuity of $D$ on $K$, we have  
\begin{align*}
\frac{1}{\e} D(\bCpe,\bCpehat)&=\frac{1}{\e}D\big( \bPi (\bCpe),
 \bPi (\bCpehat)\big)\leq \frac{c}{\e} | \bPi (\bCpe){-} \bPi (\bCpehat)|\\
&\leq \frac{c}{\e}|\bCpe{-}\bCpehat|\stackrel{\eqref{eq:CCtilde}}{=}
c|2\widetilde\bz+\e\bL_\e| 
\end{align*}
 and the right-hand side  is uniformly bounded  in $L^{\infty}$. Since $ |\O{\setminus} K_\e|<c\e^2$, it
follows  that 
$$
\limsup_{\e\rightarrow 0}\calD_\e(\bz_\e,\widehat\bz_\e)\leq \limsup_{\e\rightarrow
0}\int_{K_\e}D_\e(\bz_\e,
\widehat\bz_\e) \, \d x.
$$
On the other hand, by recalling \eqref{lipschitz_char},  on the
set $K_\e$ we have that
 
$$
\frac{1}{2\e}D\big(\exp(2\e\bz_\e),\exp(2\e\haz\bz_\e)\big)\leq
\widetilde R(\haz\bz_\e{-}\bz_\e)=\widetilde R(\widetilde\bz),
$$
hence
$$
\limsup_{\e\rightarrow 0}\calD_\e(\bz_\e,\widehat\bz_\e)\leq\int_\O
R(\widetilde\bz)\,\d x=\calD_0(\bz_0,\widehat\bz_0).
$$

\noindent\emph{ Step 4: The $\limsup$ inequality for the gradient.}
We aim  at showing  that
\begin{align}
&\limsup_{\e\rightarrow 0}
\frac{\mu}{2\e^2}\left(\int_\O|\nabla \exp(
2\e\widehat\bz_\e)|^2\d x-\int_\O|\nabla \exp(2\e\bz_\e)|^2\dx
\right)\nonumber\\
&\leq 2\mu\int_\O|\nabla 
\widehat \bz_0 | ^2\d x-2 \mu \int_\O|\nabla\bz_0|^2 \dx .\label{EXP}
\end{align}
By the  contractive character of $\bPi$, see \eqref{eq:Pi}, we
have 
$$
| \nabla
\exp(2 \e \haz \bz_\e) |=|\nabla( \bPi  (\bCpetild))|\leq|\nabla \bCpetild|.
$$

By using   the decomposition  \eqref{eq:nablaC} and the bound $\|\nabla
\bCpe\|_{L^2}\leq c\e  $, from the coercivity condition \eqref{eq:coerc0} we compute
\begin{align*}
&\frac{1}{\e^2}
(\|\nabla\bCpetild\|_{L^2}^2-\|\nabla\bCpe\|_{L^2}^2)=\|2\nabla \widetilde\bz+
\e^{-1}\nabla\bCpe+\e\nabla\bL_\e\|_{L^2}^2-\|\e^{-1}\nabla\bCpe\|_{L^2}^2\\
&\leq 4\|\nabla
\widetilde\bz\|_{L^2}^2+2\e^{-1}\|\nabla\widetilde\bz\,\nabla\bCpe+\nabla\bCpe
\nabla\widetilde\bz\|_{L^1}+c\e.
\end{align*}
%By deriving \eqref{eq:expa} and taking into account $\|\bze\|_{L^2}+\|\nabla
%\bze\|_{L^2}\leq  c$, one obtains
%$$
%\e^{-1}\nabla\bCpe=2\nabla\bze+\nabla\bL_\e,\quad \|\nabla\bL_\e\|_{L^1}\leq c\e
%$$
%Therefore
%$$
%\e^{-2}
%(\|\nabla\bCpetild\|_{L^2}^2-\|\nabla\bCpe\|_{L^2}^2)\leq 4\|\nabla
%\widetilde\bz\|_{L^2}^2+4\|\nabla\widetilde\bz\,\nabla\bze+\nabla\bze \nabla\widetilde\bz\|_{L^1}+c\e.
%$$
 Owing to convergence \eqref{tocheck2} we have that
$(2\e)^{-1}\nabla\bCpe\stackrel{L^2}{\weak}\nabla\bz_0 $, so that 
\begin{align*}
&\limsup_{\e\rightarrow 0}\left(\frac{\mu}{2\e^2}\left(
\|\nabla\bCpehat\|_{L^2}^2-\|\nabla
\bCpe\|_{L^2}^2\right)\right)\\ &\leq 2\mu\|\nabla
\widetilde\bz\|_{L^2}^2+2\mu\|\nabla\widetilde\bz\,\nabla \bz_0 +\nabla \bz_0\nabla\widetilde\bz\|_{L^1}=\\
&=2\mu\|\nabla(\widetilde\bz+\bz_0)\|_ { L^2 } ^2-2\mu\|\nabla \bz_0\|_{L^2}^2=2\mu\|\nabla \widehat\bz\|_ { L^2 } ^2-2\mu\|\nabla \bz_0\|_{L^2}^2
\end{align*}
which corresponds to \eqref{EXP}.

\noindent { \it Step 5: the $\limsup$ inequality for the elastic energies.}
Let $\bA_\e$ be defined by \eqref{eq:Aepsilon} and $\haz\bA_\e$
 have an  analogous 
expression in terms of   $\haz u_\e$ and $\haz\bz_\e$.  We
 aim at proving  
that
\begin{align*}
& \limsup_{\e\rightarrow 0}\left(\int_\O \We^{\e}(\widehat\bA_\e)\, \d x-\int_\O
\We^{\e}(\bA_\e) \, \d x\right)\\
&\leq\frac12 \int_\O \left|\nabla\widehat
u_0^{\rm sym}{-}\widehat\bz_0\right|_\bbC^2 \d x-\frac12\int_\O
\left|\nabla u_0^{\rm sym}{-}\bz_0\right|_\bbC^2 \d x,
\end{align*}
where  we have used the short-hand notation 
 $\We^\e(\bA):=\e^{-2}\We(\one{+}\e\bA)$.  We preliminarily
 observe that 
% and
% $$
% \bA_\e:=\e^{-1}\left[(\one{+}\e\nabla
% u_\e)\bCpe^{-1/2}-\one\right]\quad \widehat\bA_\e:=\e^{-1}\left[(\one{+}\e\nabla
% \widehat u_\e)\bCpehat^{-1/2}-\one\right].
% $$ 
\begin{equation}\label{eq:A-L2-bound}
 \|\bA_\e\|_{L^2}+ \|\widehat\bA_\e\|_{L^2}\leq c.
\end{equation}
 Indeed, by using the decomposition  \eqref{eq:expa}  we have
$$
\bCpe^{-1/2}=\bm I+\e\bL_\e%,\qquad \bCpehat^{-1/2}=\bCpetild^{-1/2}=\exp[-\e(\widetilde\bz+\bze)]=\bm
%I+\e\haz\bL_\e,
$$
with $\bL_\e$  %,\;\widehat\bL_\e$ 
satisfying the bound  \eqref{eq:expa4}.
% $$
% |\bCpehat{-}\one|=| \bPi (\bCpetild)- \bPi (\one)|\leq
% |\bCpetild{-}\one|=|\e\bL_\e|.
% $$
 In particular, one has 
$$
\bA_\e=\nabla
 u_\e+\bL_\e+\nabla
 u_\e(\e\bL_\e)%,\qquad
%\widehat \bA_\e=\nabla
%\widehat u_\e+\widehat\bL_\e+\nabla
%\widehat u_\e(\e\widehat\bL_\e)\textrm{ on }K_\e.
$$
so that the bound for $\|\bA_\e\|_{L^2}\leq c$ follows. The
control of $\haz \bA_\e $ is analogous. On the set $K_\e$ we have that
$$\widehat \bA_\e=\nabla
\widehat u_\e+\widehat\bL_\e+\nabla
\widehat u_\e(\e\widehat\bL_\e)$$ for some $\haz \bL_\e$ fulfilling
\eqref{eq:expa4} and we use the fact that $|\Omega\setminus K_\epsi|
\leq c\epsi^2$.

We next remark that 
\begin{equation}\label{eq:u-conv}
 \nabla\widehat u_\e - \nabla u_\e\overset{L^2}{\rightarrow}\nabla\widetilde u.
\end{equation}
 Indeed,  by computing
\begin{align}
 \nabla \haz u_\epsi &= \frac{1}{\epsi} \big( \nabla
 \tilde y(y_\epsi) \nabla y_\epsi {-} \one\big) 
 = \frac{1}{\epsi}
 \big( (\one{+}\epsi \nabla \tilde u)(y_\epsi) \nabla y_\epsi {-}
 \one\big)  \nonumber\\
&= \frac{1}{\epsi}
 \big( \nabla y_\epsi {+} \epsi \nabla \tilde u(y_\epsi) \nabla y_\epsi {-}
 \one\big) = \nabla u_\epsi +  \nabla \tilde u(y_\epsi)  + \epsi  \nabla
\tilde u(y_\epsi)  \nabla u_\epsi\nonumber
\end{align}
we obtain that 
\begin{align}
\|(\nabla \haz u_\epsi {-} \nabla u_\epsi) - \nabla \tilde u\|_{L^2}
&\leq \| \nabla \tilde u (y_\epsi) {-} \nabla \tilde u\|_{L^2} +
 \| \epsi \nabla \tilde u (y_\epsi) \nabla u_\epsi\|_{L^2}
\nonumber\\
&\leq c \epsi+ c \epsi \| \nabla
u_\epsi\|_{L^2}\leq c\epsi.\nonumber
\end{align}
On the other hand, we readily check that 
$$ \nabla \haz u_\e + \nabla u_\e = 2 \nabla u_\e + \nabla \tilde
u(y_\e) + \e \nabla \tilde u(y_\e) \nabla u_\e$$
so that the convergence 
\begin{equation}
  \label{eq:u-conv2}  \nabla \haz u_\e + \nabla u_\e \stackrel{L^2}{\weakto} \nabla \haz
  u_0 + \nabla u_0  
\end{equation}
follows. By combining \eqref{eq:u-conv0} and \eqref{eq:u-conv}-\eqref{eq:u-conv2} we
obtain that 
\begin{equation}
  \haz u_\e \stackrel{H^1}{\weakto} \haz u_0.\label{eq:u-conv3}
\end{equation}
 As $\| \nabla \haz \bz_\e\|_{L^2}$
is bounded, convergences \eqref{eq:bz-conv3} and \eqref{eq:u-conv3}
entail that
$$(\haz u_\e, \haz \bz_\e) \stackrel{\mathcal Q}{\weak} (\haz u_0,\haz
\bz_0)$$
as required by Lemma \ref{lem:MRS}.
%
%======================
% PROOF TAKEN FROM MS13
%
% In fact, by definition of $\widehat u_\e$
% $$
% \nabla\widehat u_\e=\nabla u_\e+\nabla\widetilde u\circ y_\e\cdot(\one+\e\nabla
% u_\e),\quad y_\e:=\textrm{id}+u_\e,
% $$
% hence
% \begin{eqnarray*}
%  \|\nabla\widehat u_\e - \nabla u_\e-\nabla\widetilde u\|_{L^2}\leq \|\nabla\widetilde u
% \circ y_\e -
% \nabla\widetilde u\|_{L^2}+\e\|(\nabla\widetilde u\circ y_\e)\cdot \nabla u_\e\|_{L^2}\leq
% c\e+
% c\e
% \|\nabla u_\e\|_{L^2}\leq c\e.
% \end{eqnarray*}
% We have used the fact that $\nabla\widetilde u$ is Lipschitz-continuous and  $
% \|y_\e-\textrm{id}\|_{L^2}=\e\|u_\e\|_{L^2}\leq c\e$.\\
%=========================
%

 We now turn to the proof of the two convergences 
\begin{align}\label{eq:strong-conv-A}
 \widehat\bA_\e-\bA_\e&\stackrel{L^2}{\rightarrow} \nabla\widetilde
u-\widetilde\bz=(\nabla\widehat
u_0{-}\widehat\bz_0)-(\nabla u_0{-}\bz_0),\\
 \widehat\bA_\e+\bA_\e&\stackrel{L^2}{\rightharpoonup} (\nabla\widehat
u_0{-}\widehat\bz_0)+(\nabla u_0{-}\bz_0).\label{eq:weak-conv-A}
\end{align}

On the
set $K_\e$ we use
\eqref{eq:CCtilde} for $\alpha=-1/2$ in order to get that 
\begin{align*}
 &\widehat\bA_\e{-}\bA_\e= \frac{1}{\e} \left((\one{+}\e\nabla
\widehat u_\e)\bCpetild^{-1/2}-(\one{+}\e\nabla
u_\e)\bCpe^{-1/2}\right)\\
&= \frac{1}{\e}  \left(\bCpetild^{-1/2}{-}\bCpe^{-1/2}
\right)+\nabla\widetilde
u\bCpe^{-1/2}+(\nabla\widehat u_\e{-}\nabla u_\e{-}\nabla\widetilde
u)\bCpe^{-1/2}+\nabla\widehat  u_\e\left(\bCpetild^{-1/2}{-}\bCpe^{-1/2}\right)\\
&=\left(-\widetilde\bz+\e\bL_\e\right)+\nabla\widetilde
u\bCpe^{-1/2}+(\nabla\widehat u_\e{-}\nabla u_\e{-}\nabla\widetilde
u)\bCpe^{-1/2}+\e\nabla\widehat
u_\e\left(-\widetilde\bz+\e\bL_\e\right) 
\end{align*} 
with $\|\bL_\e\|_{L^2}\leq c$. 
The first term  in the above right-hand side  converges $L^2$-strongly to $-\widetilde\bz$. Since 
$\bCpe^{-1/2}\overset{L^2}{\longrightarrow}\one $ by \eqref{eq:expa}, the second term
 strongly  converges in $L^2$ to
$ \nabla\widetilde u$. The last two terms are easily seen to be
strongly  $L^2 $  convergent to
zero.
Since $|\O{\setminus} K_\e|\leq c\e^2 $, the bound
\eqref{eq:A-L2-bound}  yields the convergence \eqref{eq:strong-conv-A}. 

The proof of the weak convergence  \eqref{eq:weak-conv-A} results
as a combination of the same argument for 
\eqref{eq:Aepsilon-lim} on $K_\e$  and the 
$L^2$-boundedness of $\bA_\e$ and $ \widehat\bA_\e$.

 We now   define for  all   $\delta>0$ the sets 
$$
U_\e^\delta:=\{x\in \O \; | \: |\e\bA_\e(x)|+|\e\widehat\bA_\e(x)|\leq
 \tilde c_\delta\}
$$
 with  $\tilde c_\delta$ from \eqref{quad_behavior2}.  On these sets,
 also by using the bound 
\eqref{eq:A-L2-bound}, we have 
\begin{eqnarray}
 \We^\e(\widehat\bA_\e)-\We^\e(\bA_\e)&\leq&
(1{+}\delta)|\widehat\bA_\e|_{\bbC}^2-(1{-}\delta)|\bA_\e|_{\bbC}^2\leq
|\widehat\bA_\e|_{\bbC}^2-|\bA_\e|_{\bbC}^2+2c \delta  \tilde
c^2_\delta  =\nonumber\\
&=&\frac{1}{2}(\widehat\bA_\e{-}\bA_\e){:}\bbC(\widehat\bA_\e{+}\bA_\e)+2c \delta
 \tilde
c^2_\delta \quad
\textrm{on} \quad U^\delta_\e.
\label{eq:quad-trick}
\end{eqnarray}
 We can easily estimate the measure of the sets $U_\e^\delta$ as
follows 
\begin{equation}\label{eq:O-U}
 |\O{\setminus} U_\e^\delta|=\int_{\O\setminus U_\e^\delta}\d x\leq
 \frac{1}{\tilde c_\delta^2} \int_{\O\setminus
  U_\e^\delta}(|\e\bA_\e(x)|+|\e\widehat\bA_\e(x)|)\, \d x\leq 
\frac{c\e^2}{\tilde
c^2_\delta}.
\end{equation}

Thanks  convergences  \eqref{eq:strong-conv-A} and\eqref{eq:weak-conv-A}, 
estimates \eqref{eq:quad-trick} and \eqref{eq:O-U}  entail 
\begin{align*}\nonumber
&\limsup_{\e\rightarrow
0}\int_{U_\e^\delta} 
\left(\We^\e(\widehat\bA_\e)-\We^\e(\bA_\e)\right)\,\d x\\\nonumber
&\leq \limsup_{\e\rightarrow
0}\left( 2c \delta
\tilde c^2_\delta|\O| +\frac{1}{2}\int_{U_\e^\delta}
(\widehat\bA_\e{-}\bA_\e){:}\bbC(\widehat\bA_\e{+}\bA_\e)\, \d x\right)\\\nonumber
&\leq  2c \delta
\tilde c^2_\delta|\O| +\limsup_{\e\to
0}\left(\frac{1}{2}\int_{U_\e^\delta}
(\nabla\widetilde u{-}\widetilde\bz){:}\bbC(\nabla\widehat
u_0{-}\widehat\bz_0{+}\nabla
u_0{-}\bz_0)\,\d x\right)\\%\label{eq:U-estimate}
&\leq  2c \delta
\tilde c^2_\delta|\O| +\frac{1}{2}\int_\O \left(
|\nabla\widehat u_0^{\rm sym}{-}\widehat\bz_0|_\bbC^2-|\nabla
u_0^{\rm sym}{-}\bz_0|_\bbC^2\right)\d x\qquad
\end{align*}
where the  minor-symmetry  property $ \bbC\bA=\bbC\bA^{\rm
  sym}$ has been used. 

 We shall now discuss the contribution of the the sets $\Omega
\setminus U_\e^\delta $. We show that the corresponding
elastic-energies terms  are uniformly bounded by $c\e$. 
Consider relation 
$$
\nabla\widehat y_\e\bCpehat^{-1/2}=(\nabla\widehat y_\e\nabla y_\e^{-1})(\nabla
y_\e\bCpe^{-1/2})(\bCpe^{1/2}\bCpehat^{-1/2}).
$$
 This can be rewritten as 
\begin{equation*}
 1+\e\widehat\bA_\e=\bG_{1,\e}(1{+}\e\bA_\e)\bG_{2,\e}
\end{equation*}
with
\begin{equation*}
 \bG_{1,\e}:=\nabla\widehat y_\e\nabla y_\e^{-1},\qquad
\bG_{2,\e}:=\bCpe^{1/2}\bCpehat^{-1/2}.
\end{equation*}
Recalling  the choice  \eqref{eq:recovery-y}, we have 
$$
\bG_{1,\e}{-}\one=\nabla(\widetilde y\circ y_\e)\nabla y_\e^{-1}-\one=\nabla\widetilde
y-\one=\e(\nabla \widetilde u)\circ y_\e.
$$
Now we consider
\begin{eqnarray*}
\bG_{2,\e}{-}\one=\bCpe^{1/2}\left(( \bPi (\bCpetild))^{-1/2}-\bCpe^{-1/2}\right).
\end{eqnarray*}
 By using   the Lipschitz-continuity of the matrix  square
root,  see Lemma \ref{lem:app-Lip}), we have 
\begin{eqnarray*}
|\bG_{2,\e}{-}\one|\leq c| \bPi (\bCpetild){-}\bCpe|\leq c
|\bCpetild{-}\bCpe|\leq c\e|2\widetilde\bz+\e\bL_\e|.
\end{eqnarray*}
 The uniform bounds  
\begin{equation}\label{eq:G-bounds}
\|\bG_{1,\e}{-}\one\|_{L^\infty}\leq c\e,\qquad \|\bG_{2,\e}{-}\one\|_{L^\infty}\leq c\e
\end{equation}
 then follow. These bounds allow us to use
the following estimate \cite[Lemma 4.1]{MieSte13} 
\begin{align}
|\We(\bG_1\bF\bG_2){-}\We(\bF)|\leq
c(1{+}\We(\bF))(|\bG_1{-}\one|+|\bG_2{-}\one|) \quad \forall
|\bG_1|, \,  |\bG_2|\leq \delta\label{eq:MieSte13}
\end{align}
for some constants $c,\, \delta >0$. 
By combining this with  the bounds \eqref{eq:G-bounds} one has
that 
\begin{align*}
 &\int_{\O\setminus U_\e^\delta}(\We^\e(\widehat\bA_\e)-\We^\e(\bA_\e))\,\d x=\frac{1}{\e^2}
 \int_{\O\setminus U_\e^\delta}(\We(\bG_{1,\e}\bF_\e\bG_{1,\e})-\We(\bF_\e))\,\d
x\\&\stackrel{\eqref{eq:MieSte13}}{\leq}  \frac{c}{\e^2}\int_{\O\setminus
U_\e^\delta}(1{+}\We(\bF_\e))(|\bG_1{-}\one|+|\bG_2{-}\one|)\,\d
x\\
&\stackrel{\eqref{eq:G-bounds}}{\leq}\frac{c}{\e} \int_{\O\setminus U_\e^\delta}(1{+}\We(\bF_\e))\,\d
x\stackrel{\eqref{eq:O-U}}{\leq} c\e,
\end{align*}
 which completes the proof. 

\noindent \emph{Step 6: The $\limsup$ inequality for the plastic
  energy.} We aim at showing that 
\begin{equation}
 \limsup_{\e\rightarrow 0}\left(\int_\O   \Wh^{\e}(\widehat\bz_\e)\,\d x-\int_\O
 \Wh^{\e}(\bz_\e)\,\d x \right)\leq \frac12\int_\O |\widehat\bz_0|_\bbH^2
\,\d x-\frac12\int_\O
|\bz_0|_\bbH^2 \,\d x,\label{WLS}
\end{equation}
 where we have used the short-hand notation  $\Wh^{\e}(\bz):=\e^{-2}\Wh(\exp(2\e\bz)) $.
The strategy is similar to the one used for the elastic energy.
First, we define the sets 
$$
Z_\e^\delta:=\left\{x\in \O\;|\;|\e\bze(x)|+|\e\haz\bze(x)|\leq 
  c_\delta  \right\},
$$
 with $c_\delta$ from \eqref{quad_behavior}, so that   
\begin{align}\label{eq:Wh1}
\Wh^\e(\widehat\bz_\e)-\Wh^\e(\bz_\e)\leq
\frac{1}{2}(\widehat\bz_\e{-}\bz_\e){:}\bbH(\widehat\bz_\e{+}\bz_\e)+2c \delta
  c_\delta^2 \quad
\textrm{on} \ \ Z^\delta_\e.
\end{align}
 Arguing exactly as in Step 5, we can prove
that the complementary sets $\O{\setminus}
Z_\e^\delta$ fulfill  $$|\O{\setminus}
Z_\e^\delta|\leq \frac{c\e^2}{
  c_\delta^2} .$$ 
 Owing to the Lipschitz-continuity  of $\Wh$ on
$K$, the contraction property of $ \bPi $ and \eqref{eq:CCtilde}
\begin{align}\nonumber
 \int_{\O\setminus Z_\e^\delta}\left( \Wh^{\e}(\widehat\bz_\e){-}
\Wh^{\e}(\bz_\e)\right)\d x=\frac{1}{\e^2}\int_{\O\setminus Z_\e^\delta}\left(\haz
\Wh(\bCpehat){-}\haz
\Wh(\bCpe)\right)\d x\leq\\\label{eq:Wh2}
\frac{c}{\e^2}\int_{\O\setminus Z_\e^\delta}|\bCpehat{-}\bCpe|\d x\leq
\frac{c}{\e^2}\int_{\O\setminus Z_\e^\delta}|\bCpetild{-}\bCpe|\d x\leq
\frac{c}{\e^2}|\O\setminus Z_\e^\delta| \e=c\e.
\end{align}
By combining \eqref{eq:Wh1}-\eqref{eq:Wh2} and
\eqref{eq:bz-conv1}-\eqref{eq:bz-conv2}, 
 the $\limsup$ condition \eqref{WLS} follows  from  $\delta$ being
 arbitrary. 
\end{proof}

\begin{proof}[Proof of Theorem \emph{\ref{sdl2}}]  Having established Lemmas
  \ref{lem:G-limit-En}, \ref{lem:G-limit-En}, and \ref{lem:MRS}, we
  are in the position of applying
  the abstract Lemma \ref{lem:evol}. Although Lemma \ref{lem:MRS}
  deals with smooth and compactly supported competitors only, note
  that the full strength of condition \eqref{closure} can be easily recovered
  by density.

The pointwise strong convergence of $(u_\e,\bz_\e)$ and the
convergence of energies and dissipation follow at once from the uniform
convexity of the linearized energy $\calE_0$ along the same lines as
\cite[Cor. 3.8 and Cor. 3.9]{MieSte13}.
\end{proof}

\appendix

\section{Local  Lipschitz continuity}
%==================
%
 We comment here on the local  Lipschitz continuity of  
the matrix logarithm and the matrix fractional power on $\MM$. 
This is a  consequence of the unit determinant constraint, which
allows to control the  moduli  of the matrix eigenvalues and
their reciprocals  in terms of the matrix
norm.

\begin{lemma}[Local Lipschitz continuity]\label{lem:app-Lip}
  We have that  
\begin{equation}\label{eq:app}
 |\log \bC_{1}{-}\log\bC_2|\leq
c(1+(|\bC_{1}|\vee |\bC_2|)^{2})|\bC_{1}{-}\bC_2|\quad 
\forall \bC_{1},\,\bC_{2}\in\MM
\end{equation}
 for  some positive constant $c>0$. 
In particular,  given  any compact $K\subset \MM$  there exists $c_K>0$ such that
$|\log\bC_{1}{-}\log\bC_2|\leq c_K |\bC_{1}{-}\bC_2|$ for all 
$\bC_{1},\bC_2\in K$.
 Moreover, for  all $\alpha\in\Rz$, we have that 
\begin{equation}\label{eq:app2}|\bC_{1}^{\,\alpha}{-}\bC_2^{\,\alpha}|\leq
   c_{K\alpha} |\bC_{1}{-}\bC_2|\quad 
\forall \bC_{1},\,\bC_{2}\in K\end{equation}
for some positive constant $  c_{K\alpha}$.
\end{lemma}

\begin{proof}
 Let $\s_i\subset (0,\infty)$ be the spectrum  of $\bC_i$,
  for $i=1,2$,  and
$\la_0=\min\{\s_1\cup\s_2\}>0$. Since $\det \bC_i=1$  one easily
checks  that
\begin{equation}\label{eq:la0}
 \la_0\geq |\bC_{1}|^{-2}\wedge |\bC_2|^{-2}.
\end{equation}
The  logarithm  of
$\bC_i$ can be calculated  via  the Cauchy Integral Formula
(for operators) \cite[Ch. 7]{Dun-Schw} 
 $$
 \log\bC_i=\int_{\g}\frac{\log z}{z\one-\bC_i}\,\d z,
 $$
 where $\g$ is a closed contour in the analyticity region of 
 $\log z$ (one can take $\gamma\subset \{ \textrm{Re }z>0\}$, for
 instance)  and  winds  one time around $\s_1\cup \s_2$. Therefore
\begin{align*}
 &\log\bC_{1}{-}\log\bC_2=(\bC_{1}{-}\bC_2)\int_{\g}\frac{\log z}{(z\one{-}
\bC_{1})(z\one{-}\bC_2)}\,\d z
\\
&=(\bC_{1}{-}\bC_2)\int_{\bar \g}\frac{\log z}{(z\one{-}\bC_{1})(z\one{-}\bC_2)} \,\d z,
\end{align*}
where, in the last equality,   we have replaced  $\g$  the infinite straight line $\bar
\g=\{x_0+i
t\,|\,t\in\R\}$, $x_0\in (0, \la_0)$,   since the modulus of the
integrand  behaves like 
$z \mapsto |\log z|\,|z|^{-2}$ at infinity.  For all $z \in \bar
\g$ we have 
$$
\textrm{Re } z=x_0<\la_0\Rightarrow \left|\frac{1}{z\one-\bC_i}\right|\leq
\frac{\sqrt{3}}{|z-\la_0|}.
$$
 We hence compute that 
\begin{eqnarray*}
  |\log\bC_{1}{-}\log\bC_2|\leq |\bC_{1}{-}\bC_2|\int_{\bar
\g}\frac{3|\log z|}{|z-\la_0|^2}
\,\d z\leq \frac32|\bC_{1}{-}\bC_2|\int_{-\infty}^{\infty}\frac{|\log(x_0^2+t^2)|+\pi}{
(x_0-\la_0)^2+t^2}
\dt.
\end{eqnarray*}
The last inequality follows from the elementary control
$$|\log(x_0{+}it)| \leq \frac12 |\log(x_0^2{+}t^2)| + |\vartheta|
\leq \frac12 \left(|\log(x_0^2{+}t^2)|  +\pi\right)$$
for $\vartheta := \arctan (t/x_0) \in (-\pi/2,\pi/2)$.
 As this  estimate holds for any $x_0\in (0,\la_0)$,  by letting $x_0\rightarrow 0$ we obtain
\begin{eqnarray*}
| \log\bC_{1}{-}\log\bC_2|\leq 3 |\bC_{1}{-}\bC_2|
\int_{0}^{\infty}\frac{|\log t^2|+\pi}{\la_0^2+t^2}\dt.
\end{eqnarray*}
 We can now  elementarily compute  that 
$$
\int_{0}^{\infty}\frac{\pi}{\la_0^2+t^2}\dt=\frac{\pi^2}{2\la_0},\quad \int_{0}^{\infty}\frac{|\log t^2|}{\la_0^2+t^2}\dt\leq-2\int_{0}^{1}\frac{\log t}{\la_0^2}\dt+\int_{1}^{\infty}\frac{\log t^2}{t^2}dt=\frac{2}{\la_0^2}+c.
$$
 Eventually, we have proved that 
$$
|\log\bC_{1}{-}\log\bC_2|\leq c\left(1{+}\frac{1}{\la_0^2}\right)|\bC_{1}{-}\bC_2|\stackrel{\eqref{eq:la0}}{\leq} c\left(1+(|\bC_{1}|\vee |\bC_2|)^{2}\right)|\bC_{1}{-}\bC_2| .
$$
As for the  matrix  power $\bC \mapsto \bC^\alpha$, we simply use
 $\bC^{\,\alpha}=\exp(\alpha\log \bC)$ and  recall that the 
exponential  map  is uniformly Lipschitz on compact sets.
\end{proof}
\section{The map $ \bPi $}
%==================
%
 We collect here some comment on  the existence  of a map
$$\bPi:\MM \to K$$
having properties \eqref{eq:Pi}, to be used in the definition of the
recovery sequence \eqref{rec_z}. Recall that 
$$K=\{\bC\in\MM\;|\; |\bC|\leq r\},\qquad r>|\one|=\sqrt{3}$$
 and let  the flux  $\Phi_t$, $t\geq 0$,  be 
associated to the following differential equation on  $\GLp$ 
\begin{equation}\label{eq:diff-eq}
 \dot\bC=-\left(\bC{-}3|\bC^{-1}|^{-2}\bC^{-1}\right).
\end{equation}
 In particular, $t \mapsto \Phi_t(\bC)$ is  be the solution of the differential equation
\eqref{eq:diff-eq} with initial datum $\bC$.
 Note that the manifold $\MM$ is  invariant under the flux
$\Phi_t$. In fact, along solutions $\bC(t)$ of
the equation \eqref{eq:diff-eq} symmetry and determinant constraint 
are preserved as  the symmetry of $\bC$ induces that of $\dot \bC$
and 
$$\tr(\bC(t)^{-1}\dot\bC(t))=-\tr(\one-3|\bC^{-1}|^{-2}\bC^{-2}) =
-(3{-} 3|\bC^{-1}|^{-2}\bC^{-2}{:}\one) =0.
$$
 Moreover, the flux $\Phi_t $ is norm-contractive  for we  
readily check that  
\begin{align*}
 \frac{1}{2}\frac{\d}{\d t}|\bC(t)|^2&=\tr(\bC\dot\bC)=-\tr(\bC^2{-}3|\bC^{-1}|^{-2}
\one)\\&=-(|\bC|^2{-}9|\bC^{-1}|^{-2})\leq 3-|\bC|^2\leq 0.
\end{align*}
 We have  here  used the fact that $\bC^{-1}\in\MM$ and
 $|\bC^{-1}|^2\geq 3$.  More precisely, as $|\bC|^2 \geq 3$ on
 $\MM$ (with equality corresponding to $\bC=\one$), we have checked
 that 
 \begin{equation}
   \label{perillimes}
   |\bC| > \sqrt{3} \ \Rightarrow \ \frac12 \frac{\d}{\d t}
   |\bC(t)|^2<0. 
 \end{equation}
 Let us record some additional properties of the flux $\Phi_t$ in the following lemma.

\begin{lemma}\label{lemma:B1}  The flux $\Phi_t$ satisfies the following properties
\begin{itemize}
 \item [i)] Let $\bC,\,\bC_{0}\in\MM$. Then
\begin{equation}\label{eq:Phi1}
 | \Phi_t(\bC) |\geq |\bC_{0}| \ \Rightarrow  \ \frac{\d}{\d t}| \Phi_t(\bC) {-}\bC_{0}|\leq 0.
\end{equation}
\item [ii)] For all $t \geq 0$, $\Phi_t$ is a contraction on $\MM$, namely
\begin{equation}\label{eq:Phi2}
|\Phi_t(\bC_{1}){-}\Phi_t(\bC_2)|\leq|\bC_{1}{-}\bC_2|\quad
\forall \bC_{1},\,\bC_2\in\MM .
\end{equation}
\end{itemize}
\end{lemma}
\begin{proof}
 Ad i). Let $\bC(t) = \Phi_t(\bC)$ for some $\bC\in\MM$. The differential equation \eqref{eq:diff-eq} entails that 
\begin{align*}
& \frac{1}{2}\frac{\d}{\d t}|\bC(t){-}\bC_{0}|^2=\tr\left(\dot\bC(\bC{-}\bC_{0})\right)\\
&=-\tr\left(\bC(\bC{-}\bC_{0})\right)+\frac{3}{|\bC^{-1}|^{2}}\left(3-{}\tr(\bC^{-1}
\bC_{0})\right)
\end{align*}
 By the invariance of the trace by  cyclic permutation  of
the factors we have that 
$$
\tr(\bC^{-1}\bC_{0})=\tr(\bC^{-1/2}\bC_{0}\bC^{-1/2})\geq 3
$$
since $\bC^{-1/2}\bC_{0}\bC^{-1/2} \in \MM$. Moreover,
$$
|\bC|\geq |\bC_{0}|\Rightarrow |\bC|^2=\tr(\bC^2)\geq \tr(\bC\bC_{0}),
$$
which proves the statement.

 Ad ii). Let $\bC_i(t) = \Phi_t(\bC_i)$ for $\bC_1,\,\bC_2\in\MM$.  By using again  the differential equation  we get 
\begin{align*}
&\frac{1}{2}
\frac{\d}{\d t}|\bC_{1}(t){-}\bC_2(t)|^2=\tr\left[(\dot\bC_{1}{-}
\dot\bC_2)(\bC_{1}{-}\bC_2)\right ]=
\\
&=-\tr\left[\left(\bC_{1}-\frac{3}{|\bC_{1}^{-1}|^2}\bC_{1}^{
-1 }
-\bC_2+\frac{3}{|\bC_2^{-1}|^2}\bC_2^{-1}\right)(\bC_{1}{-}\bC_2)\right ]=\\
&=-|\bC_{1}{-}\bC_2|^2+3\left(\frac{3-\tr(\bC_{1}^{-1}\bC_2)}{|\bC_{1}^{-1}|^2}+\frac{
3-\tr(\bC_2^{
-1}\bC_{1})}{|\bC_2^{-1}|^2}\right)\leq 0
\end{align*} 
proving the assertion.
\end{proof}
%
% Since $|\bC|^2\geq 3$ on $\MM$ (with  equality corresponding to
%    $\bC=\one$) it follows that
% \begin{equation}
% \lim_{t\rightarrow\infty}\Phi_t(\bC)=\one \quad
% \forall\bC\in\MM.\label{star}
% \end{equation}
For  any initial datum
$\bC\in\MM$, by  \eqref{perillimes}  there exists a minimum time $t_0\geq 0$ for which
$|\Phi_{t_0}(\bC)|=r>\sqrt{3}$,  namely
\begin{equation}\label{eq:min-time}
t_0(\bC)=\min\{t\geq0\;|\; \Phi_{t_0}(\bC)\in K\}<\infty.
\end{equation}
 We define the map $ \bPi :\MM\rightarrow K$ as follows
 \begin{equation}\label{eq:def-Pi}
    \bPi (\bC):=\Phi_{t_0(\bC)}(\bC).
 \end{equation} 
Of course, $ \bPi \left|_K\right.=\text{id}$, so that the
first condition in \eqref{eq:Pi} is satisfied. 

Let $\bC_1, \, \bC_2 \in \MM$ be given and assume with no loss of
generality that  $t_1:=t_0(\bC_{1})\leq t_0(\bC_2)=:t_1+\delta$. We can write
 \begin{align*}
  | \bPi (\bC_{1})- \bPi (\bC_2)|&= |\Phi_{t_1}(\bC_{1})-\Phi_{t_2}(\bC_2)|=|\Phi_{t_1}(\bC_{1})-\Phi_{\delta}(\Phi_{t_1}(\bC_2))|\\
 & \stackrel{\eqref{eq:Phi1}}{\leq}
|\Phi_{t_1}(\bC_{1})-\Phi_{t_1}(\bC_2)|\stackrel{\eqref{eq:Phi2}}{\leq}|\bC_{1}
- \bC_2|.
 \end{align*} 
The map $\bPi$ is hence contractive in $\MM$, which is 
the
second condition in \eqref{eq:Pi}. 
%
% 
% \begin{remark}\label{altra_forma} Before closing this discussion let
%   us explicitly remark that the theory  can  be adapted to other forms
%   of the compact constraining set $K$. In particular, by using
%   directly the formulation in terms of $\log \bCp$, one could consider
%   the compact subset of $ \Rzd$  given by 
% $$
% \mathcal K=\{\bA\in\Rzd\,: \,|\bA |\leq R\},\quad R>0
% $$
% Then, the correpsonding set $ K:=\exp(\mathcal K)\subset\MM $
% is also compact. 
% We define the projection map $\bP:\Rzd\rightarrow\mathcal K $ by
% \begin{equation}
%   \bP(\bA):=(1\wedge R|\bA|^{-1}) \bA
% \end{equation}
% It can be easily proved that $\bP$ is a contraction. Let us now define $ \bPi: \MM\rightarrow K$ by:
% $$
% \bPi\left(\bCp\right)=\exp(\bP(\log\bCp)).
% $$

%  qui non ho capito come si conclude 
%  With this choice $\hat\bze=P(\bze+\widetilde\bz)$.  
% \\
% The map $\bPi$ is a contraction on $\MM$: $|\bPi(\bC_1)-\bPi(\bC_2)|\leq
% |\bC_1-\bC_2| $ for any $\bC_1,\bC_2\in\MM$.(?)

% \end{remark}
% %
% 
%
%
%===================
\section{Lower semincontinuity tool}
%===================
%
For the sake of completeness, we report here the lower-semicontinuity
tool which has been repeatedly used above. The lemma is in the spirit
of \cite[Thm. 1]{Balder84} and  \cite{Ioffe77}. A proof can be found
in \cite[Thm
4.3, Cor. 4.4]{be}, in 
\cite[Lemma 3.1]{Mielke-et-al08}  
in one dimension, and in \cite{Li96} in case of local uniform convergence.
 
\begin{lemma}[Lower-semicontinuity tool]\label{lem:semicon-tool}
Let $f_0, \, f_\epsi: \Rz^n \to [0,\infty]$ be lower semicontinuous, 
$$ \quad  f_0(v_0) \leq \inf\big\{\liminf_{\epsi
  \to 0} f_\epsi(v_\epsi) \ | \
v_\epsi \to v_0\big\} \quad \forall v_0 \in \Rz^n $$
and $ w_\epsi  \weakto w_0$ in $L^1(\Omega;\Rz^n)$. By denoting by
$ \zeta  $
the Young measure generated by $ w_\epsi  $ we have that 
$$ \int_\Omega \left( \int_{\Rz^n} f_0(w) {\rm d}  \zeta_x (w)\right) \dx
\leq \liminf_{\epsi \to 0} \int_\Omega
f_\epsi(w_\epsi) \dx.$$
In particular, if $f_0 $ is convex we have 
$$ \int_\Omega f_0(w_0) \dx \leq \liminf_{\epsi \to 0} \int_\Omega
f_\epsi(w_\epsi) \dx.$$
\end{lemma}

%
%
%
%=========================
\section*{Acknowledgement}
%=========================
%
U.S. acknowledges support by the CNR-JSPS grant {\it VarEvol} and by the Austrian Science Fund (FWF) project P 27052-N25. This work has been funded by the Vienna Science and Technology Fund (WWTF)
through Project MA14-009. Stimulating discussions with Gianni Dal Maso,
Alexander Mielke, Stefanie Reese, and Filip Rindler are gratefully
acknowledged.

%
%============================
\section*{Statement on conflict of interests}

The authors declare that they have no conflict of interests.

\end{document}